\documentclass[prl,aps,twocolumn,showpacs,floatfix]{revtex4-1}
\usepackage{amsfonts}
\usepackage{amssymb}
\usepackage{hyperref}
\usepackage{graphicx}
\usepackage{dcolumn}
\usepackage{bm,amsmath,verbatim}
\usepackage{mathrsfs}
\usepackage{color}

\usepackage[T1]{fontenc}
\usepackage[latin9]{inputenc}
\usepackage{booktabs}

\hypersetup{colorlinks,
linkcolor=blue,          citecolor=blue,        filecolor=blue,      urlcolor=blue           }

\newcommand{\ket}[1]{|#1\rangle}
\newcommand{\bra}[1]{\langle #1 |}

\newcommand{\ii}{\text{i}}

\begin{document}

\title{Non-Hermitian Absorption Spectroscopy}
\author{Kai Li$^{1}$}
\author{Yong Xu$^{1,2}$}
\email{yongxuphy@tsinghua.edu.cn}
\affiliation{$^{1}$Center for Quantum Information, IIIS, Tsinghua University, Beijing 100084, People's Republic of China}
\affiliation{$^{2}$Shanghai Qi Zhi Institute, Shanghai 200030, People's Republic of China}

\begin{abstract}
While non-Hermitian Hamiltonians have been experimentally realized in cold atom systems,
it remains an outstanding open question of how to experimentally measure their complex energy spectra
in momentum space for a realistic system with boundaries. The existence of non-Hermitian skin effects
may make the question even more difficult to address given the fact that
energy spectra for a system with open boundaries are dramatically different from those in momentum space;
the fact may even lead to the notion that momentum-space band structures are not experimentally accessible
for a system with open boundaries.
Here, we generalize the widely used radio-frequency spectroscopy to measure both real and imaginary parts of
complex energy spectra of a non-Hermitian quantum system for either bosonic or fermionic atoms.
By weakly coupling the energy levels of a non-Hermitian system to auxiliary energy levels,
we theoretically derive a formula showing that the decay of atoms on the auxiliary energy levels reflects the real and
imaginary parts of energy spectra in momentum space.
We further prove that measurement outcomes are independent of boundary conditions in the thermodynamic limit,
providing strong evidence that the energy spectrum in momentum space
is experimentally measurable.
We finally apply our non-Hermitian absorption spectroscopy protocol to the Hatano-Nelson model and
non-Hermitian Weyl semimetals to demonstrate its feasibility.
\end{abstract}
\maketitle

Measurements based on spectroscopy providing band structure information play a key role in
identifying various topological phases in condensed matter and cold atom systems~\cite{Torma2016Review, Ding2019NRP, Takagi2020JCMP, NHLinearResponse, Shen2021RMP, Zwierlein2021NP}.
In the past few years, non-Hermitian topological physics has seen a rapid advance~\cite{ChristodoulidesNPReview,XuReview,ZhuReview,UedaReview,BergholtzReview}.
Such systems usually exhibit complex band structures with exceptional points or rings~\cite{Zhen2015nat,Xu2017PRL,Nori2017PRL,Kozii2017,Zyuzin2018PRB,Zhou2018,Cerjan2018PRB,Yoshida2018PRB,Zhao2018PRB,Carlstrom2018PRA,
HuPRB2019,Wang2019PRB,Yoshida2019PRB,Ozdemir2019,Cerjan2019nat,Kawabata2019PRL,Zhang2019PRL,Zhang2020PRL,
Chuanwei2020PRL,Yang2020PRL,Wang2021PRL,Nagai2020PRL,Nori2021PRL,Yuliang2021,An2022PRB}.
In cold atom systems, non-Hermitian Hamiltonians have been experimentally
realized by introducing atom loss~\cite{LuoNC2018,JoarXiv2021,Gadway2019NJP,Yan2020PRL,Takahashi2020PTEP,Esslinger2021arxiv,Esslinger2021PRX}. The parity-time ($\mathcal{PT}$) symmetry breaking has been observed
by measuring the population of an evolving state of a system through quench
dynamics~\cite{LuoNC2018,JoarXiv2021,WeiZhang2021PRL}. However, such a method is very hard to generalize to a generic case without
$\mathcal{PT}$ symmetry. In fact, while there are practical proposals
on how to realize non-Hermitian topological phases in cold atom systems~\cite{Xu2017PRL,Yang2021,Cui2021,Lang2022}, such as
non-Hermitian Weyl semimetals, it remains an outstanding open question of how to
experimentally measure both the real and imaginary parts of energy spectra
in a non-Hermitian cold atom system.
Moreover, one of the most important phenomena in non-Hermitian systems is the non-Hermitian skin effects (NHSEs); 
with such effects, band structures under open boundary conditions (OBCs) are
dramatically different from those under periodic boundary conditions (PBCs)~\cite{Yao2018PRL1,TonyLee,Xiong2018JPC,Torres2018PRB,Kunst2018PRL,Qibo2020,Lang2019PRB,Okuma2020PRL,Slager2020PRL,ChenFang2020PRL}.
The difference naturally leads to a question of whether the existence
of NHSEs makes it impossible to measure the energy
spectra in momentum space for a system with open boundaries.
The question is important in light of the fact that most
experiments are performed in a geometry with boundaries.

\begin{figure}[t]
\includegraphics[width=3.4in]{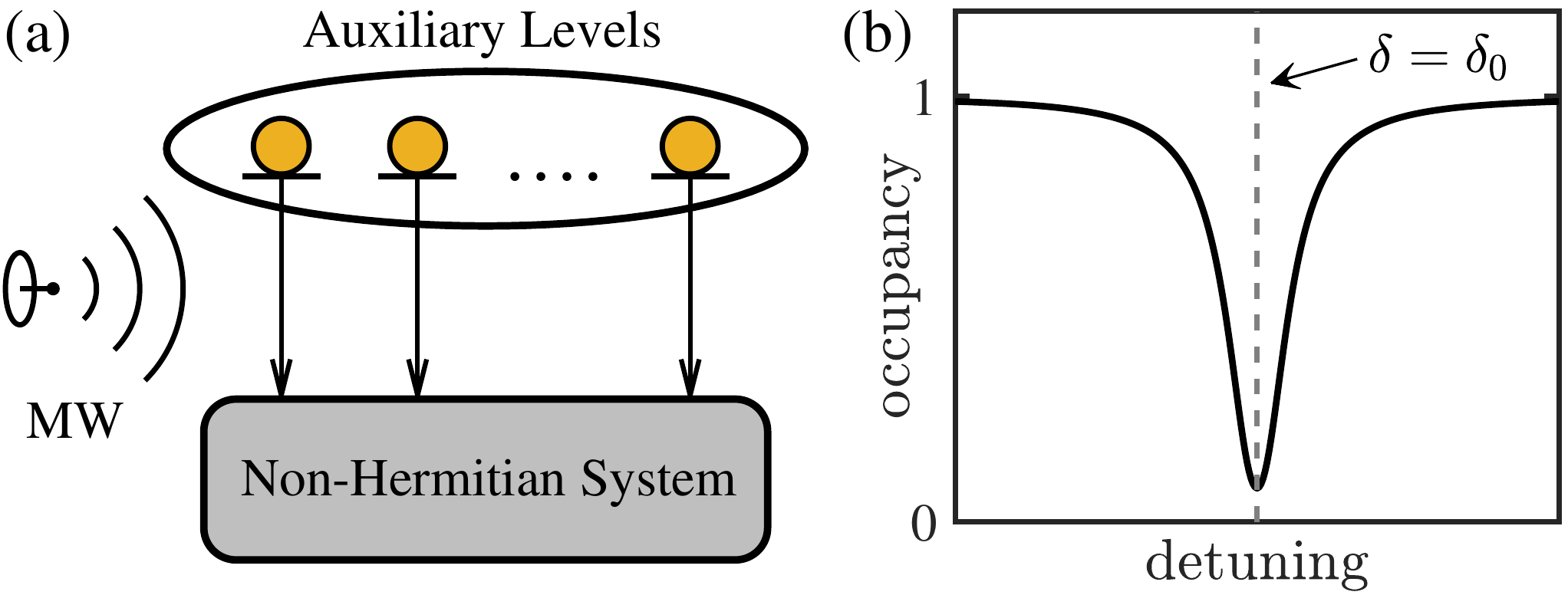}
\caption{(a) Schematics of non-Hermitian absorption spectroscopy where a non-Hermitian system is coupled to auxiliary energy levels
by a microwave field or a laser beam.
Atoms are initially prepared on the auxiliary levels, and their occupancy on the auxiliary levels is
finally measured after a period of time.
(b) A typical example of the measured occupancy as a function of the detuning $\delta$ with an occupancy dip at $\delta=\delta_0$.
The center and the width of the dip are related to the real and the imaginary part of an eigenenergy of a non-Hermitian system, respectively.
}
\label{fig1}
\end{figure}

A widely used spectroscopy in cold atom systems is the
radio-frequency (RF) spectroscopy~\cite{Torma2016Review,Zwierlein2021NP}. There, auxiliary energy levels are
weakly coupled to system energy levels so that atoms will either be driven
from occupied auxiliary levels to empty system levels or from
occupied system levels to empty auxiliary levels, when
the frequency of radio waves or microwaves match the energy difference.
By imaging the transmitted atoms,
such a spectroscopy allows us to map out the energy band dispersion,
similar to angle resolved photoemission spectroscopy (ARPES) in solid-state
materials. In the paper, we generalize the RF spectroscopy to a non-Hermitian
system to allow for measurements of both the real and imaginary parts of
a complex energy spectrum in cold atom systems (see Fig.~\ref{fig1}).
For a non-Hermitian system in cold atoms, there always exists a total loss of atoms on the
system levels, making it impossible to image any atoms on the system levels after a long period of time.
We thus consider initial preparations of atoms on the auxiliary levels followed by
measurements of atoms on these levels instead of system levels.
Using linear response theory,
we derive a formula describing the population of auxiliary levels, based on which
one can extract not only real parts of the system's band structure in momentum space but also its imaginary parts.
We further prove that measurement results are independent of boundary conditions in the thermodynamic limit
despite the existence of NHSEs;
this is in stark contrast to the results in Ref.~\cite{Sato2021PRL} showing that ARPES might be
sensitive to skin effects.
Finally, we utilize the Hatano-Nelson (HN) model and
a non-Hermitian Weyl semimetal to demonstrate the feasibility of the spectroscopy.

\emph{Momentum-resolved non-Hermitian Absorption spectroscopy.}---We start by considering a generic translation-invariant non-Hermitian system described by a Hamiltonian
$\hat{\mathcal{H}}_\text{s} = \sum_{i\alpha,j\beta} [H_\text{s}]_{i\alpha,j\beta} \hat{c}_{i\alpha}^\dagger \hat{c}_{j\beta}$ with $\hat{c}_{i\alpha}^\dagger$ ($\hat{c}_{i\alpha}$) being 
either the fermionic or bosonic creation (annihilation) operator acting on the $\alpha$th degree of freedom of the $i$th unit cell.
We assume that the non-Hermitian Hamiltonian $\hat{\mathcal{H}}_\text{s}$ is purely dissipative, i.e. $\text{Im}(\lambda) < 0$ for any eigenvalue $\lambda$ of $H_\text{s}$.
To measure the system's complex energy spectra, we couple the first degree of freedom of each unit cell to an auxiliary energy level by a microwave field such that the full Hamiltonian under rotating wave approximations becomes ($\hbar = 1$)
\begin{equation}
\label{full_Hamiltonian}
\begin{aligned}
\hat{\mathcal{H}} =
\hat{\mathcal{H}}_\text{s} + \hat{\mathcal{H}}_\text{a} +
({\Omega}/{2}) \sum\nolimits_{j}(\hat{c}_{j1}^{\dagger}\hat{a}_{j}+\hat{a}_{j}^{\dagger}\hat{c}_{j1})
\end{aligned}
\end{equation}
with $\hat{\mathcal{H}}_\text{a} =
\sum_{j} \big[ J (\hat{a}_{j}^{\dagger} \hat{a}_{j+1} + \hat{a}_{j+1}^{\dagger} \hat{a}_{j})
- (\omega - \omega_\text{a}) \hat{a}_{j}^{\dagger} \hat{a}_{j} \big]$ describing the auxiliary levels, where
$\hat{a}_{j}^{\dagger}$ ($\hat{a}_{j}$) is the creation (annihilation) operator acting on the $j$th auxiliary level,
$J$ is the hopping strength between nearest-neighbor auxiliary levels, $\omega_\text{a}$ is the energy
of the auxiliary level of an atom measured relative to the first system energy level of the atom, and $\omega$ is
the frequency of the microwave field ($\delta = \omega - \omega_\text{a}$ is the detuning).
The final term depicts the coupling between system and auxiliary
levels with $\Omega$ being the Rabi frequency of the microwave field. Note that $J$, $\omega$, $\omega_a$
and $\Omega$ are all real numbers.

To measure the energy spectrum of the system, we first prepare a cloud of fermionic atoms on the first band of
auxiliary levels described by the many-body state $|\psi_0^{(M)}\rangle=\prod_k\hat{a}_k^\dagger|0\rangle$
with $\ket{0}$ being the vacuum state at a low temperature. For bosonic atoms, we consider a finite temperature ensemble (see Supplemental Material S-1 B
for detailed discussions).
We then switch on the coupling between system and auxiliary levels by shining a microwave field on
the atoms.
After a long period of time, we perform the time-of-flight measurement to obtain the atom population on the auxiliary levels
at each momentum $k$.
While the dynamics of a dissipative cold atom system is usually described by the master equation,  
	in Supplemental Material S-1, we have proved that
	the atom number on auxiliary levels at momentum $k$ is given by
	$N_{\text{a},k}= \text{Tr}[ \rho(t) \hat{a}_k^\dagger \hat{a}_k ]
	=N_{\text{0},k} \bra{0} \hat{a}_k e^{\ii \hat{\mathcal{H}}^\dagger t} \hat{a}_k^\dagger \hat{a}_k e^{-\ii \hat{\mathcal{H}} t} \hat{a}_k^\dagger \ket{0}$, 
	where $N_{\text{0},k}= \text{Tr}[ \rho(0) \hat{a}_k^\dagger \hat{a}_k ]$. Here $\rho(t)$ is the density matrix evolving from either the initial state
	$\rho(0)=|\psi_0^{(M)}\rangle \langle \psi_0^{(M)}|$ for fermions ($N_{\text{0},k}=1$ in this case) or a finite temperature ensemble for bosons.
	The result indicates that the dynamics is completely determined by the non-Hermitian Hamiltonian $\hat{\mathcal{H}}$.
	This allows us to use a single-particle state
	$\ket{\psi_\text{a}^k} = \hat{a}_k^\dagger|0\rangle$ ($\ket{\psi_\text{a}^k}$ is an eigenstate of $\hat{\mathcal{H}}_\text{a}$ with eigenenergy
	$E_k = -\delta + 2J \cos k $) on auxiliary levels
	as an initial state to derive the protocol.

With an initial state prepared as $\ket{\psi_\text{a}^k}$ on auxiliary levels, we use the linear response theory~\cite{TormaBook} to derive the population of auxiliary levels at time $t$ under PBCs as~\cite{Supplementary} 
\begin{equation}
\label{theor1}
N_\text{a} (t)
= N_\text{a} (0) \exp (-\kappa t)
\end{equation}
with
$
\kappa =- \frac{\Omega^{2}}{2} \sum_m \frac{a_{km}^{(1)}  \gamma_{km} - b_{km}^{(1)} \Delta_{km} }{\Delta_{km}^2 + \gamma_{km}^2}
$ and $N_\text{a} (0)$ being the initial occupancy of the single-particle state of the auxiliary levels at momentum $k$.
The result holds for both bosons and fermions. For simplicity, we will consider $N_\text{a} (0)=1$ henceforth.
Here, $\Delta_{km} = E_k - \varepsilon_{km}$ with $\varepsilon_{km}$ ($-\gamma_{km}$) denoting the real (imaginary) part of eigenenergies
of ${\mathcal{H}}_\text{s}$, which are labeled by the momentum $k$ and the band index $m$;
$a_{km}^{(1)}=\mathrm{Re} [c_{km}^{(1)}]$ and $b_{km}^{(1)}=\mathrm{Im} [c_{km}^{(1)}]$ where
$c_{km}^{(\alpha)} = \langle \psi_\text{s}^{k\alpha} | u_\text{R}^{km} \rangle \langle u_\text{L}^{km} | \psi_\text{s}^{k\alpha} \rangle$
with $\ket{\psi_\text{s}^{k \alpha}} = \hat{c}_{k\alpha}^\dagger\ket{0}=\frac{1}{\sqrt{N}} \sum_j e^{\ii k j} \hat{c}_{j\alpha}^\dagger \ket{0}$,
$\ket{u_\text{R}^{km}}$ being the right eigenstate of the system Hamiltonian, i.e., $\hat{\mathcal{H}}_\text{s} \ket{u_\text{R}^{km}} = (\varepsilon_{km}-\ii \gamma_{km}) \ket{u_\text{R}^{km}}$, and $\bra{u_\text{L}^{km}}$ being the corresponding left one.
In the derivation, we have assumed that $\Omega$ is sufficiently small compared with the decay rates of system states
and $t\gg 1/\gamma_{km}$.
One can also see the similarity between $-4\ln [N_\text{a}(t)]/(\Omega^2 t)$ given by Eq.~(\ref{theor1})
and
the ${\mathcal{H}}_\text{s}$'s spectral function,
$A(k,\omega) = -2\text{Im} \, \text{Tr} [(\omega - {\mathcal{H}}_\text{s})^{-1}] = \sum_m {2\gamma_{km}}/ {[(\omega-\varepsilon_{km})^2 + \gamma_{km}^2]}$
given the fact that $\sum_{\alpha}c_{km}^{(\alpha)}=1$.

We now briefly summarize the derivation of Eq.~(\ref{theor1}) (the details can be found in Supplemental Material S-2).
We first write the full Hamiltonian (\ref{full_Hamiltonian}) as $\hat{\mathcal{H}} = \hat{\mathcal{H}}_\text{0} + \hat{\mathcal{V}}$ with $\hat{\mathcal{H}}_\text{0} = \hat{\mathcal{H}}_\text{s} + \hat{\mathcal{H}}_\text{a}$.
In the interaction picture,
the state at time $t$ is given by $\ket{\psi^I (t)} = \hat{U}^I (t,0) \ket{\psi_\text{a}^k}$ where $\hat{U}^I (t,t_0) = 1 - \ii \int_{t_0}^{t} dt' \hat{\mathcal{V}}^I (t') + \mathcal{O} (\Omega^2)$ is the time evolution operator.
Through careful derivations, we find that $\dot{N}_\text{a} (t) = \bra{\psi^I (t)} \hat{\dot{N}}_\text{a}^I (t) \ket{\psi^I (t)} = -\kappa$.
Considering the fact that $N_\text{a} (t)$ decreases with time, we obtain
$
	\dot{N}_\text{a}(t) = -\kappa N_\text{a}(t)
$, which yields Eq.~(\ref{theor1}) after integration.
The result has also been numerically confirmed.
We note that different from Refs.~\cite{NHLinearResponse,Zhai2021PRL} where a non-Hermitian perturbation is added in a
Hermitian system,
we here include a Hermitian perturbation to measure a non-Hermitian system's properties.

Based on Eq.~(\ref{theor1}), we see that the spectral line of $-\ln [N_\text{a}(t)]$ as a function of $-\delta$ consists of multiple peaks centered at $-\delta = \varepsilon_{km} - 2J \cos k$ with half widths approximated by $2\gamma_{km}$.
In fact, Eq.~(\ref{theor1}) allows us to obtain both $\varepsilon_{km}$ and $\gamma_{km}$ (as well as the quantities $a_{km}$ and $b_{km}$)
by fitting the spectral line using this formula.

\begin{figure}[t]
\includegraphics[width=3.4in]{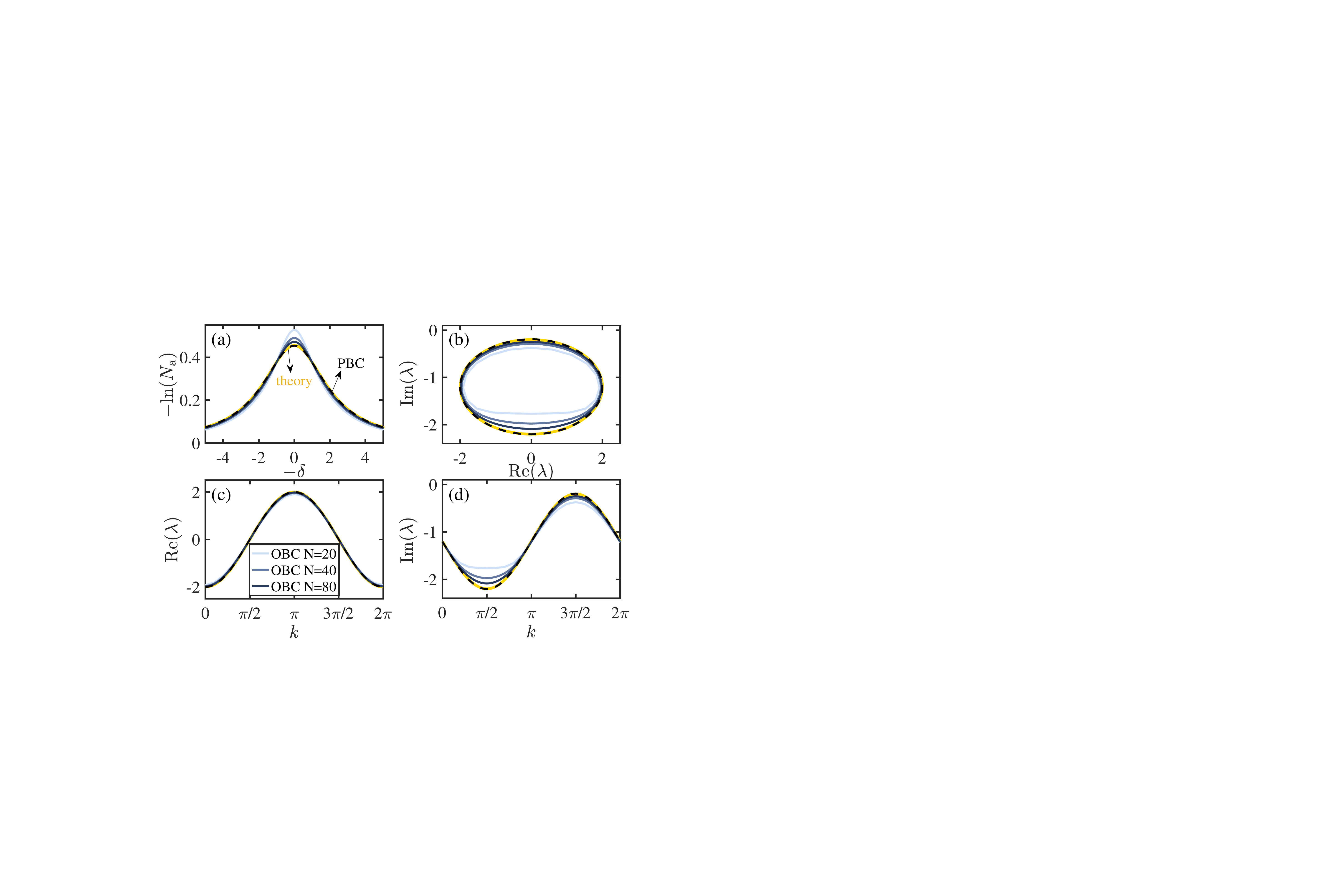}
\caption{
(a) The spectral lines of the HN model at $k=\pi/2$
with the yellow line obtained from Eq.~(\ref{logNa_HN}) and the dashed black and solid blue lines
numerically computed for a system under PBCs and OBCs, respectively.
Energy spectra in the complex energy plane (b), real (c) and imaginary parts (d) of energies with respect to $k$,
where $k$ denotes the Bloch momentum as we consider a Bloch state as an initial state.
The energies are extracted by fitting the numerically calculated spectral lines
under PBCs (dashed black lines) and OBCs (solid blue lines) based on Eq.~(\ref{logNa_HN}) in comparison with
the theoretical result, $\varepsilon_k=2J_s\cos k+2\ii g\sin k-2\ii \gamma$ (the yellow lines).
Here, $J_s=-1$, $J=-0.02$, $g=-0.5$, $\gamma=0.6$, $\Omega=0.1$ and $t=200$.
}
\label{fig2}
\end{figure}

Before applying our method to several paradigmatic models,
we wish to prove that our conclusion is independent of boundary conditions in the thermodynamic limit, although we derive the results under PBCs.
Here we briefly summarize the proof; the detailed one can be found in Supplemental Material S-3.
We first prove that $D = \bra{k \alpha} e^{-\ii \hat{\mathcal{H}}_\text{o} t} - e^{-\ii \hat{\mathcal{H}}_\text{p} t} \ket{k' \alpha'} \propto 1/N$
where $\ket{k\alpha}$ is the k-space basis vector, and letters `o' (`p') are used to denote quantities under OBCs (PBCs).
Assuming the hopping range is finite, we obtain $\bra{k\alpha} \hat{B} \ket{k' \alpha'} = f_{\alpha \alpha'}(k,k')/N$ where $\hat{B} = \hat{\mathcal{H}}_\text{p} - \hat{\mathcal{H}}_\text{o}$ and $f_{\alpha \alpha'} (k,k')$ is independent of $N$.
We then prove that each term in $D_n = \bra{k \alpha} (\hat{\mathcal{H}}_\text{p} - \hat{B})^n - (\hat{\mathcal{H}}_\text{p})^n \ket{k' \alpha'}$ is proportional to $1/N$
and thus conclude that $D = \sum_{n=1}^{\infty} \frac{(-\ii t)^n}{n!} D_n \propto 1/N$.
Since $\dot{N}_\text{a}^\text{b} (t) = (\Omega^2/4) \int_{0}^{t} dt' \Gamma^{\text{b}} (t,t') + \text{H.c.}$ and $\Gamma^{\text{b}} (t,t') = - \sum_{q \tilde{q}} \bra{\psi_\text{a}^k} e^{ \ii \hat{\mathcal{H}}_\text{a}^\text{b} t} \ket{\psi_\text{a}^{q}} \bra{\psi_\text{s}^{q 1}} e^{-\ii \hat{\mathcal{H}}_\text{s}^\text{b} (t - t')} \ket{\psi_\text{s}^{\tilde{q} 1}} \bra{\psi_\text{a}^{\tilde{q}}} e^{-\ii \hat{\mathcal{H}}_\text{a}^\text{b} t'} \ket{\psi_\text{a}^k}$ ($\text{b}=\text{o},\text{p}$), we derive that $(\Gamma^\text{o} - \Gamma^\text{p}) (t,t') \propto 1/N$,
yielding $\dot{N}_\text{a}^\text{o} (t) - \dot{N}_\text{a}^\text{p} (t) \propto 1/N$.
Taking the infinite size limit, we obtain $N_\text{a}^\text{o} (t) = N_\text{a}^\text{p} (t)$,
implying that
the result is independent of boundary conditions in the thermodynamic limit.

\emph{Hatano-Nelson model.}---To demonstrate the feasibility of our spectroscopy protocol, we apply it to the HN model~\cite{HN1996PRL} (the simplest model that supports NHSEs):
\begin{equation}
\label{HNmodel}
\hat{\mathcal{H}}_\text{s}^\text{HN} = \sum\nolimits_j  [(J_s + g) \hat{c}_j^\dagger \hat{c}_{j+1} + (J_s - g) \hat{c}_{j+1}^\dagger \hat{c}_{j}
-2 \ii \gamma \hat{c}_j^\dagger \hat{c}_j ],
\end{equation}
where $J_s$ and $g$ are real parameters describing the hopping strength between nearest-neighbor sites. When $g\neq 0$,
the NHSE occurs due to the asymmetric hopping.
Here we add an onsite dissipation term $-2\ii \gamma \sum_j \hat{c}_j^\dagger \hat{c}_j$ with  $\gamma > |g|$ to ensure that $\hat{\mathcal{H}}_\text{s}^\text{HN}$ is purely dissipative.
The full Hamiltonian with each system site coupled to an auxiliary level is given by
$\hat{\mathcal{H}}^\text{HN} =
\hat{\mathcal{H}}_\text{s}^\text{HN} +
\sum_j [ J ( \hat{a}_{j}^{\dagger} \hat{a}_{j+1} + \hat{a}_{j+1}^{\dagger} \hat{a}_{j} )
 - \delta \hat{a}_{j}^{\dagger} \hat{a}_{j}
+ ({\Omega}/{2}) (\hat{c}_{j}^{\dagger}\hat{a}_{j}+\hat{a}_{j}^{\dagger}\hat{c}_{j}) ]$.
Since $\ket{u_\text{R}^k} = \ket{u_\text{L}^k} = \ket{\psi_\text{s}^k} $ for the HN model, we obtain $a_k=1$, $b_k=0$ and thus
\begin{equation}
\label{logNa_HN}
-\ln [N_\text{a} (t)]
=  \frac{\Omega^{2}t}{2} \frac{ \gamma_k }{(-\delta+2J \cos k -\varepsilon_k)^2 + \gamma_k^2},
\end{equation}
which is exactly the spectral function $A(k,\omega)$ up to a constant factor.

To verify our theory, we numerically calculate $N_\text{a} (t)$ with respect to $\delta$
and find that the results under PBCs agree very well with Eq.~(\ref{logNa_HN}) [see Fig.~\ref{fig2}(a)].
For open boundaries, while the numerical results slightly deviate from the theoretical ones,
the deviation becomes smaller as the system size increases, which is in agreement with our proof that
$N_\text{a} (t)$ under OBCs approaches the result under PBCs as we increase the system size even when the
system exhibits NHSEs.

With each $N_\text{a} (t)$ as a function of $\delta$, one can extract the energy information
by fitting the function based on Eq.~(\ref{logNa_HN}).
Figure~\ref{fig2}(b-d) illustrates the extracted energy spectra under PBCs and OBCs in comparison with
the momentum space energy spectra of the system Hamiltonian $\hat{\mathcal{H}}_\text{s}^\text{HN}$.
We see that the results under PBCs agree perfectly well with the system's complex energy spectra.
For open boundaries, while the extracted energy spectra are slightly different from
the theoretical results, the discrepancy becomes smaller as the system size becomes larger.
We thus conclude that non-Hermitian absorption spectroscopy allows us to extract both real
and imaginary parts of complex energy spectra of a non-Hermitian system in momentum space even when
the non-Hermitian system has open boundaries and NHSEs.

To further confirm that the spectroscopy measures the energy spectra in momentum space,
we study the non-Hermitian Rice-Mele model~\cite{NHRM2020PRL} with NHSEs induced by onsite dissipations in  Supplemental Material S-4,
which is more relevant to cold atom experiments (also see the proposals in Refs.~\cite{Yang2021,Cui2021}).
We show that the boundary effects can always be significantly reduced by increasing system sizes.

\emph{Non-Hermitian Weyl semimetal.}---Next, we study a three-dimensional Weyl semimetal with onsite dissipations~\cite{Xu2017PRL}:
\begin{equation}
\label{WER_Ham}
\hat{\mathcal{H}}_\text{s}^\text{NHWS} = \sum\nolimits_{\bm k}  \hat{c}_{\bm k}^\dagger h_{\bm k} \hat{c}_{\bm k},
\end{equation}
where $\hat{c}_{\bm k}^\dagger = (\hat{c}_{{\bm k}\uparrow}^\dagger \ \, \hat{c}_{{\bm k}\downarrow}^\dagger)$
and $h_{\bm k} = 2 t_\text{SO} (\sin k_x \sigma_x + \sin k_y \sigma_y) + [m_z - 2t_z \cos k_z - 2 t_1 (\cos k_x + \cos k_y)] \sigma_z + \ii \gamma (\sigma_z - \sigma_0)$ with $t_{\text{SO}}$, $t_z$, $t_1$ and $m_z$ being real system parameters,
$\gamma$ depicting the atom loss rate on a hyperfine level~\cite{Xu2017PRL}, and $\{\sigma_{\nu}\}$ ($\nu=x,y,z$) being a set of Pauli matrices.
Without onsite dissipations ($\gamma=0$), this model~\cite{Yong2016PRA,Yong2019PRB,Liu2020SB} has been experimentally realized in
cold atom systems~\cite{3DWeylband2021Science}.
When $\gamma>0$, each Weyl point develops into a Weyl exceptional ring consisting of exceptional points
on which $h_{\bm k}$ is nondiagonalizable~\cite{Xu2017PRL}.
For example, the Weyl point at $(k_x,k_y,k_z) = \{ 0,0,\arccos[ ({m_z - 4t_1})/({2t_z}) ] \}$
deforms into a Weyl exceptional ring which can be approximated by $k_x^2+k_y^2={\gamma^2}/({4 t_\text{SO}^2})$ and
$k_z = \arccos[ ({m_z-4t_1})/({2t_z}) + {t_1 \gamma^2}/({8 t_z t_\text{SO}^2}) ]$ when $\gamma \ll 2 t_\text{SO}$.

\begin{figure}[t]
\includegraphics[width=3.4in]{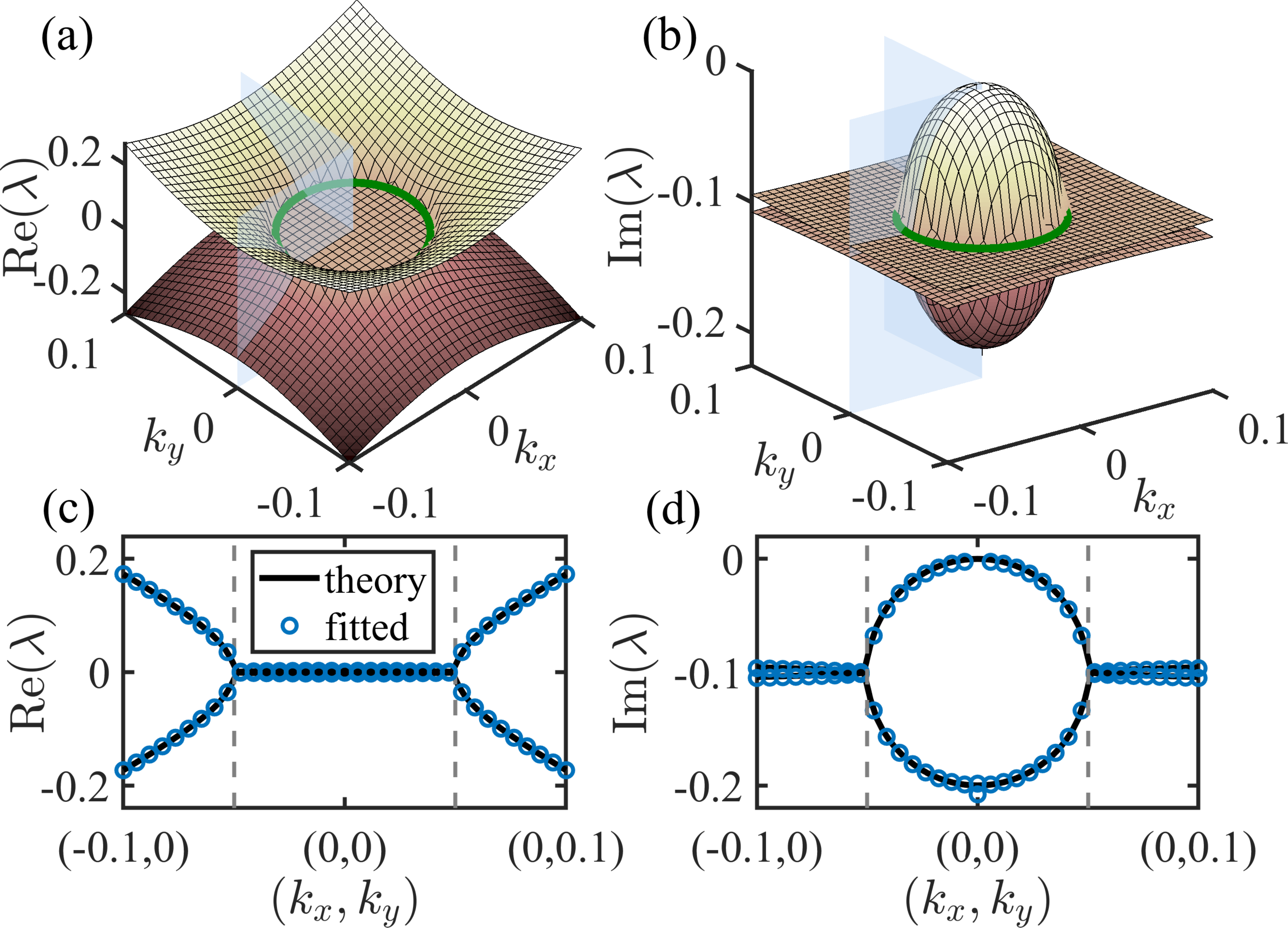}
\caption{
Real (a) and imaginary parts (b) of the extracted eigenenergy of $\hat{\mathcal{H}}_\text{s}^\text{NHWS}$
in the $(k_x,k_y)$ plane obtained by fitting the numerically simulated spectral lines based on Eq.~(\ref{theor1}).
The Weyl exceptional ring is highlighted as a green circle.
(c),(d) The sectional view (blue circles) of the fitted energy spectra on the blue planes in (a) and (b), respectively.
The black lines refer to the eigenenergies of $h_{\bm k}$.
The locations of exceptional points are marked out by vertical dashed lines.
Without loss of generality, we here set $J=0$ in light of the fact that nonzero $J$ only causes a shift of the spectral line by $2J \sum_{i=x,y,z} \cos k_i$.
Here, $t_1=t_z=t_\text{SO} = 1$, $m_z=4$, $\gamma = 0.1$ and $k_z = \arccos({1}/{800})$.
}
\label{fig3}
\end{figure}

To measure the Weyl exceptional ring, we couple the spin down degree of freedom of each atom to an auxiliary level.
Since the measurement protocol is independent of boundary conditions,
we consider the Hamiltonian under PBCs, which reads
$
\hat{\mathcal{H}}^\text{NHWS} = \sum_{\bm k} \hat{c}_{\bm k}^\dagger h_{\bm k} \hat{c}_{\bm k}
+ ({\Omega}/{2})
(\hat{c}_{{\bm k}\downarrow}^\dagger \hat{a}_{\bm k} + \hat{a}_{\bm k}^\dagger \hat{c}_{{\bm k}\downarrow})
- \delta \hat{a}_{\bm k}^\dagger \hat{a}_{\bm k} .
$
For each initial state $\ket{\psi_\text{a}^{\bm k}}$, the full Hamiltonian in k-space is a $3\times3$ matrix in the basis
$\beta_{\bm k} = \{ \hat{c}_{{\bm k}\uparrow}^\dagger \ket{0}, \hat{c}_{{\bm k}\downarrow}^\dagger \ket{0}, \hat{a}_{\bm k}^\dagger \ket{0} \}$.

Figure~\ref{fig3}(a) and (b) show the extracted energy spectra in the $(k_x,k_y)$ plane by fitting the results of $-\ln N_\text{a} (t)$ versus $\delta$,
where a Weyl exceptional ring at $k_x^2+k_y^2={\gamma^2}/{4 t_\text{so}^2}$ is highlighted as a green circle.
We see that the eigenenergies are approximately purely real (imaginary) up to a constant $-\ii \gamma$ outside (inside) the Weyl exceptional ring.
Such features can also be clearly observed in the sectional view [see Fig.~\ref{fig3}(c) and (d)] of the fitted energy spectra on the blue planes in Fig.~\ref{fig3}(a) and (b).
The sectional view further reveals the existence of exceptional points at $k_x^2+k_y^2=(0.05)^2$, the positions of which are
highlighted by vertical dashed lines. The figure illustrates that the fitted (measured) energy spectra are in excellent agreement with
the eigenenergies of $h_{\bm k}$.

\begin{figure}[t]
\includegraphics[width=3.4in]{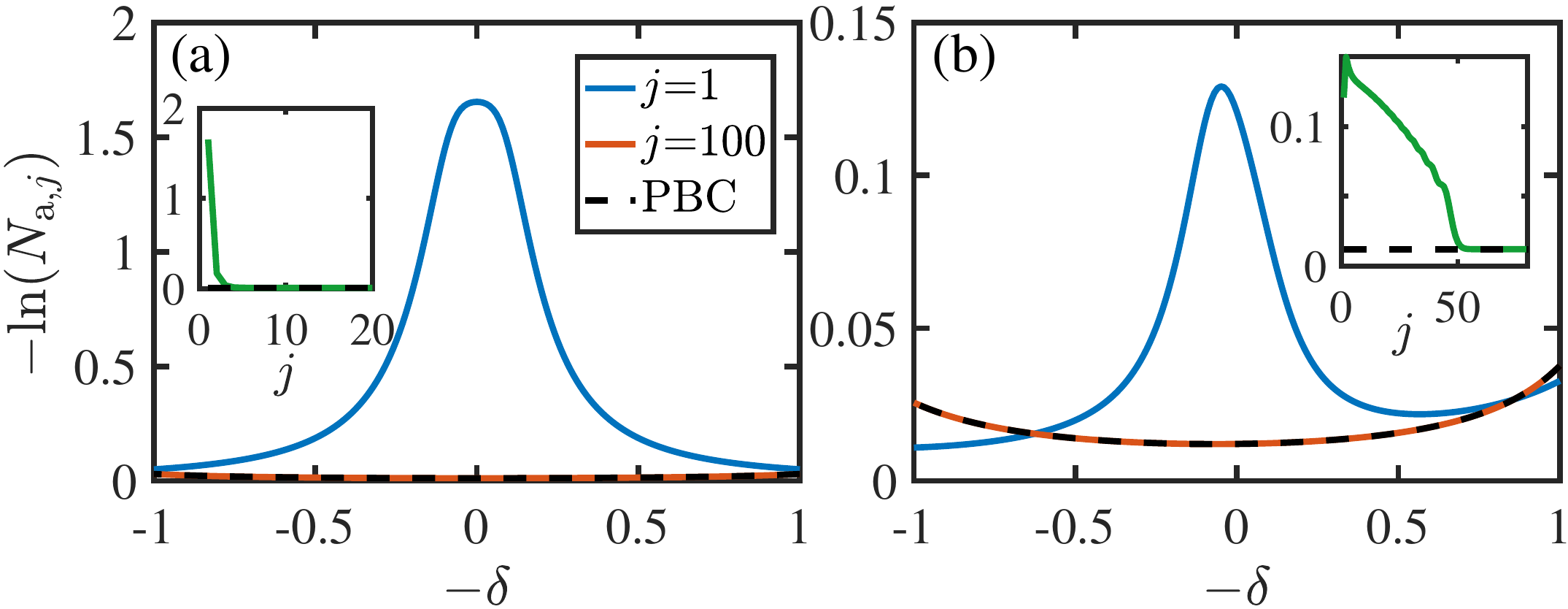}
\caption{Indicator of amounts of local absorbed atoms, $-\ln (N_{\text{a},j})$ versus $-\delta$ on the $j$th auxiliary level
when (a) $J=0$ and (b) $J=-0.05$.
The insets plot $-\ln (N_{\text{a},j})$ as a function of site $j$ when $\delta=0$.
The dashed black lines and other ones are calculated for a system (with the system size $N=200$) under PBCs and OBCs, respectively.
Here, $J_s=-1$, $J_2=-2.5$, $g=-0.1$ and $\gamma=0.1$.
}
\label{fig4}
\end{figure}

\emph{Site-resolved non-Hermitian absorption spectroscopy.}---We have shown that the energy spectra in momentum space
can be extracted by performing time-of-flight measurements of atoms on auxiliary levels.
In this section, we will demonstrate that topological edge modes, such as zero-energy modes, in non-Hermitian systems
can also be measured by probing the local occupancy $N_{\text{a},j}$
after a period of time $t$~\cite{footnote1}.
We expect that the amount of absorbed atoms at boundaries exhibits a peak near $\delta=0$
for a topological system with zero-energy edge modes.

To demonstrate our method, we consider the non-Hermitian Su-Schrieffer-Heeger (SSH) model~\cite{TonyLee,Yao2018PRL1} with two sublattices $A$ and $B$
described by
\begin{eqnarray}
	\label{NHSSH}
		\hat{\mathcal{H}}_\text{s}^\text{NHSSH} =&&
		\sum\nolimits_j [
		(J_s + g) \hat{c}_{jA}^\dagger \hat{c}_{jB}
		+ (J_s - g) \hat{c}_{jB}^\dagger \hat{c}_{jA} \\
		&& + J_2 (\hat{c}_{jB}^\dagger \hat{c}_{j+1\,A} \! + \! \text{H.c.})
		-  \ii \gamma (\hat{c}_{jA}^\dagger \hat{c}_{jA} \! + \! \hat{c}_{jB}^\dagger \hat{c}_{jB}) ]. \nonumber
\end{eqnarray}
We couple the $A$ site of each unit cell to an auxiliary level and calculate the local occupancy $N_{\text{a},j}$
in a nontrivial regime with two topological zero-energy modes localized at both edges.
The site-resolved spectrum in Fig.~\ref{fig4} displays an absorption peak around $\delta=0$
at the left boundary, which does not exist in the bulk sites, revealing the existence
of zero-energy modes. 
Such a feature is further illustrated by the amount of local absorbed atoms
on each auxiliary level when $\delta=0$ (see the insets),
showing that only the left boundary has a significant response 
(the right edge state can be probed if we couple $B$ sites to the auxiliary levels; see Supplemental Material S-6 B for details).
We also find that by turning on $J$, the topological zero modes can be detected in a wider range near the left boundary ($j<50$)
as shown in Fig.~\ref{fig4}(b).
We attribute this to the hopping between boundary sites which balances the particle distribution, while the bulk is unaffected and agrees with the
results for periodic boundaries.

In summary, we have generalized the widely used radio-frequency spectroscopy to a non-Hermitian quantum system
and demonstrated that it
can be employed to measure both real and imaginary parts of complex energy spectra. We theoretically prove and numerically confirm that such measurement results are independent of boundary conditions
even when a non-Hermitian quantum system exhibits NHSEs, thereby providing strong evidence that band structures
in momentum space are experimentally measurable in a generic non-Hermitian system.
In cold atom systems, we may consider either bosonic atoms, such as $^{87}$Rb atoms, or fermionic atoms, such as $^{173}$Yb or $^{40}$K
(see Supplemental Material S-1 C for more details).
Our methods are in fact not limited to cold atom systems, but can also be used in other quantum systems,
such as trapped ions~\cite{WeiZhang2021PRL} or solid-state spin systems~\cite{Du2018Science,Deng2021PRL,Deng2022NC}.
Such a spectroscopy may also be generalized to a dissipative interacting system (see Supplemental Material S-7 for detailed discussions). 
Given the similarity between the radio-frequency spectroscopy and ARPES in solid-state materials,
our results may also have important implications in condensed matter systems.

\begin{acknowledgments}
We thank T. Qin, Y.-L. Tao and J.-H. Wang for helpful discussions.
This work is supported by the National Natural Science Foundation of China (Grant No. 11974201)
and Tsinghua University Dushi Program.
\end{acknowledgments}

\emph{Note added.}---Recently, we became aware of a related work where topological edge states are experimentally measured in ultracold atoms~\cite{Yan2022}.


\begin{thebibliography}{99}

\bibitem{Torma2016Review}
P. T{\"o}rm{\"a},
Phys. Scr. \textbf{91}, 043006 (2016).

\bibitem{Ding2019NRP}
B. Lv, T. Qian, and H. Ding,
Nat. Rev. Phys. \textbf{1}, 609 (2019).

\bibitem{Takagi2020JCMP}
C.-L. Lin, N. Kawakami, R. Arafune, E. Minamitani, and N. Takagi,
J. Condens. Matter Phys. \textbf{32}, 243001 (2020).

\bibitem{NHLinearResponse}
L. Pan, X. Chen, Y. Chen, and H. Zhai,
{Nat. Phys.} \textbf{16}, 767 (2020).

\bibitem{Shen2021RMP}
J. A. Sobota, Y. He, and Z.-X. Shen,
Rev. Mod. Phys. \textbf{93}, 025006 (2021).

\bibitem{Zwierlein2021NP}
C. J. Vale and M. Zwierlein,
Nat. Phys. \textbf{17}, 1305 (2021).



\bibitem{ChristodoulidesNPReview}
R. El-Ganainy, K. G. Makris, M. Khajavikhan, Z. H. Musslimani, S. Rotter, and D. N. Christodoulides,
Nat. Phys. \textbf{14}, 11 (2018).

\bibitem{XuReview}
Y. Xu,
Front. Phys. \textbf{14}, 43402 (2019).

\bibitem{ZhuReview}
D.-W. Zhang, Y.-Q. Zhu, Y. X. Zhao, H. Yan, and S.-L. Zhu,
Adv. Phys. \textbf{67}, 253 (2019).

\bibitem{UedaReview}
Y. Ashida, Z. Gong, and M. Ueda,
Adv. Phys. \textbf{69}, 249 (2020).

\bibitem{BergholtzReview}
E. J. Bergholtz, J. C. Budich, and F. K. Kunst,
Rev. Mod. Phys. \textbf{93}, 015005 (2021).



\bibitem{Zhen2015nat}
B. Zhen, C. W. Hsu, Y. Igarashi, L. Lu, I. Kaminer, A. Pick, S.-L. Chua, J. D. Joannopoulos, and M. Solja\v{c}i\'{c},
Nature (London) \textbf{525}, 354 (2015).

\bibitem{Xu2017PRL}
Y. Xu, S.-T. Wang, and L.-M. Duan,
Phys. Rev. Lett. \textbf{118}, 045701 (2017).

\bibitem{Nori2017PRL}
D. Leykam, K. Y. Bliokh, C. Huang, Y. D. Chong, and F. Nori,
Phys. Rev. Lett. \textbf{118,} 040401 (2017).

\bibitem{Kozii2017} 
V. Kozii and L. Fu,
arXiv:1708.05841 (2017).

\bibitem{Zyuzin2018PRB} 
A. A. Zyuzin and A. Y. Zyuzin,
Phys. Rev. B \textbf{97}, 041203(R) (2018).

\bibitem{Zhou2018}
H. Zhou, C. Peng, Y. Yoon, C. W. Hsu, K. A. Nelson, L. Fu, J. D. Joannopoulos, M. Solja\v{c}i\'{c}, and B. Zhen,
Science \textbf{359}, 1009 (2018).

\bibitem{Cerjan2018PRB}
A. Cerjan, M. Xiao, L. Yuan, and S. Fan,
Phys. Rev. B \textbf{97}, 075128 (2018).

\bibitem{Yoshida2018PRB} 
T. Yoshida, R. Peters, and N. Kawakami,
Phys. Rev. B \textbf{98}, 035141 (2018).

\bibitem{Zhao2018PRB} 
P.-L. Zhao, A.-M. Wang, and G.-Z. Liu,
Phys. Rev. B \textbf{98}, 085150 (2018).

\bibitem{Carlstrom2018PRA} %
J. Carlstr\"{o}m and E. J. Bergholtz,
Phys. Rev. A \textbf{98}, 042114 (2018).

\bibitem{HuPRB2019} %
Z. Yang and J. Hu,
Phys. Rev. B \textbf{99}, 081102(R) (2019).

\bibitem{Wang2019PRB} %
H.-Q. Wang, J.-W. Ruan, and H.-J. Zhang,
Phys. Rev. B \textbf{99}, 075130 (2019).

\bibitem{Yoshida2019PRB} 
T. Yoshida, R. Peters, N. Kawakami, and Y. Hatsugai,
Phys. Rev. B \textbf{99}, 121101(R) (2019).

\bibitem{Ozdemir2019} %
{\c{S}}. K. \"{O}zdemir, S. Rotter, F. Nori, and L. Yang,
Nat. Mater. \textbf{18}, 783 (2019).

\bibitem{Cerjan2019nat} %
A. Cerjan, S. Huang, M. Wang, K. P. Chen, Y. Chong, and M. C. Rechtsman,
Nat. Photon. \textbf{13}, 623 (2019).

\bibitem{Kawabata2019PRL} %
K. Kawabata, T. Bessho, and M. Sato,
Phys. Rev. Lett. \textbf{123}, 066405 (2019).

\bibitem{Zhang2019PRL} %
X. Zhang, K. Ding, X. Zhou, J. Xu, and D. Jin,
Phys. Rev. Lett. \textbf{123}, 237202 (2019).

\bibitem{Zhang2020PRL} %
X.-X. Zhang and M. Franz,
Phys. Rev. Lett. \textbf{124}, 046401 (2020).

\bibitem{Chuanwei2020PRL}
J. Hou, Z. Li, X.-W. Luo, Q. Gu, and C. Zhang,
Phys. Rev. Lett. \textbf{124}, 073603 (2020).

\bibitem{Yang2020PRL} 
Z. Yang, C.-K. Chiu, C. Fang, and J. Hu,
Phys. Rev. Lett. \textbf{124}, 186402 (2020).

\bibitem{Wang2021PRL} %
K. Wang, L. Xiao, J. C. Budich, W. Yi, and P. Xue,
Phys. Rev. Lett. \textbf{127}, 026404 (2021).

\bibitem{Nagai2020PRL} 
Y. Nagai, Y. Qi, H. Isobe, V. Kozii, and L. Fu,
Phys. Rev. Lett. \textbf{125}, 227204 (2020).

\bibitem{Nori2021PRL}
T. Liu, J. J. He, Z. Yang, and F. Nori,
Phys. Rev. Lett. \textbf{127}, 196801 (2021).

\bibitem{Yuliang2021}
Y.-L. Tao, T. Qin, and Y. Xu,
arXiv:2111.03348 (2021).

\bibitem{An2022PRB}
H. Wu and J.-H. An,
Phys. Rev. B \textbf{105}, L121113 (2022).

\bibitem{LuoNC2018}
J. Li, A. K. Harter, J. Liu, L. de Melo, Y. N. Joglekar, and L. Luo,
Nat. Commun. \textbf{10}, 855 (2019).

\bibitem{JoarXiv2021}
Z. Ren, D. Liu, E. Zhao, C. He, K. K. Pak, J. Li, and G.-B. Jo,
Nat. Phys. \textbf{18}, 385 (2022).

\bibitem{Gadway2019NJP}
S. Lapp, J. Ang'ong'a, F. A. An, and B. Gadway,
New J. Phys. \textbf{21}, 045006 (2019).

\bibitem{Yan2020PRL}
W. Gou, T. Chen, D. Xie, T. Xiao, T.-S. Deng, B. Gadway, W. Yi, and B. Yan,
Phys. Rev. Lett. \textbf{124}, 070402 (2020).

\bibitem{Takahashi2020PTEP}
Y. Takasu, T. Yagami, Y. Ashida, R. Hamazaki, Y. Kuno, and Y. Takahashi,
Prog. Theor. Exp. Phys. \textbf{2020}, 12A110 (2020).

\bibitem{Esslinger2021PRX}
F. Ferri, R. Rosa-Medina, F. Finger, N. Dogra, M. Soriente, O. Zilberberg, T. Donner, and T. Esslinger,
Phys. Rev. X \textbf{11}, 041046 (2021).

\bibitem{Esslinger2021arxiv}
R. Rosa-Medina, F. Ferri, F. Finger, N. Dogra, K. Kroeger, R. Lin, R. Chitra, T. Donner, and T. Esslinger,
arXiv:2108.11888 (2021).

\bibitem{WeiZhang2021PRL}
L. Ding, K. Shi, Q. Zhang, D. Shen, X. Zhang, and W. Zhang,
Phys. Rev. Lett. \textbf{126}, 083604 (2021).



\bibitem{Yang2021}
S. Guo, C. Dong, F. Zhang, J. Hu, and Z. Yang,
arXiv:2111.04220 (2021).

\bibitem{Cui2021}
L. Zhou, H. Li, W. Yi, and X. Cui,
arXiv:2111.04196 (2021).

\bibitem{Lang2022}
Z.-C. Xu, Z. Zhou, E. Cheng, L.-J. Lang, and S.-L. Zhu,
arXiv:2201.01216 (2022).



\bibitem{Yao2018PRL1}
S. Yao and Z. Wang,
Phys. Rev. Lett. \textbf{121}, 086803 (2018).

\bibitem{TonyLee}
T. E. Lee,
Phys. Rev. Lett. \textbf{116}, 133903 (2016).

\bibitem{Xiong2018JPC}
Y. Xiong,
J. Phys. Commun. \textbf{2}, 035043 (2018).

\bibitem{Torres2018PRB}
V. M. Martinez Alvarez, J. E. Barrios Vargas, and L. E. F. Foa Torres
Phys. Rev. B \textbf{97}, 121401(R) (2018).

\bibitem{Kunst2018PRL}
F. K. Kunst, E. Edvardsson, J. C. Budich, and E. J. Bergholtz,
Phys. Rev. Lett. \textbf{121}, 026808 (2018).

\bibitem{Qibo2020}
Q.-B. Zeng, Y.-B. Yang, and Y. Xu,
Phys. Rev. B \textbf{101}, 020201(R) (2020).

\bibitem{Lang2019PRB}
H. Jiang, L.-J. Lang, C. Yang, S.-L. Zhu, and S. Chen, 
Phys. Rev. B \textbf{100}, 054301 (2019).

\bibitem{Okuma2020PRL}N. Okuma, K. Kawabata, K. Shiozaki, and M. Sato,
{Phys. Rev. Lett. \textbf{124,} 086801 (2020).}

\bibitem{Slager2020PRL}
D. S. Borgnia, A. J. Kruchkov, and R.-J. Slager,
{Phys. Rev. Lett. 124, 056802, (2020).}

\bibitem{ChenFang2020PRL}K. Zhang, Z. Yang, and C. Fang,
{Phys. Rev. Lett. \textbf{125,} 126402 (2020).}

\bibitem{Sato2021PRL}
N. Okuma and M. Sato,
Phys. Rev. Lett. \textbf{126}, 176601 (2021).

\bibitem{TormaBook}
P. T\"orm\"a,
Spectroscopies---theory, in \textit{Quantum Gas Experiments---Exploring Many-Body States},
edited by P. T\"orm\"a and K. Sengstock (Imperial College Press, London, 2015).

\bibitem{Supplementary}
See Supplemental Material for more details on the proof that
the many-body dynamics described by the master equation is determined by the single-particle dynamics governed by the effective non-Hermitian 
Hamiltonian in Section S-1,
the derivation of the population of the auxiliary levels based on linear response theory in Section S-2,
the proof that the population is independent of boundary conditions in the thermodynamic limit in Section S-3,
the influence of non-Hermitian skin effects (NHSEs) on non-Hermitian absorption spectroscopy
based on the non-Hermitian Rice-Mele model in Section S-4,
the analysis of the possibility of probing non-Bloch energies for systems that exhibit NHSEs and its problems in Section S-5,
more results about the non-Hermitian SSH model in Section S-6,
and the application of the absorption spectroscopy in interacting systems in Section S-7,
which includes Refs.~\cite{Prosen2008NJP,Wang2019PRL,Ueda2018PRX}.

\bibitem{Prosen2008NJP}
T. Prosen,
{New J. Phys.} \textbf{10}, 043026 (2008).

\bibitem{Wang2019PRL}
F. Song, S. Yao, and Z. Wang,
{Phys. Rev. Lett.} \textbf{123}, 170401 (2019).

\bibitem{Ueda2018PRX}
Z. Gong, Y. Ashida, K. Kawabata, K. Takasan, S. Higashikawa, and M. Ueda, 
{Phys. Rev. X} \textbf{8}, 031079 (2018).

\bibitem{Zhai2021PRL}
T.-S. Deng, L. Pan, Y. Chen, and H. Zhai,
Phys. Rev. Lett. \textbf{127}, 086801 (2021).

\bibitem{HN1996PRL}
N. Hatano and D. R. Nelson,
Phys. Rev. Lett. \textbf{77}, 570 (1996).



\bibitem{NHRM2020PRL}
Y. Yi and Z. Yang,
Phys. Rev. Lett. \textbf{125}, 186802 (2020).




\bibitem{Yong2016PRA}
Y. Xu and L.-M. Duan,
Phys. Rev. A \textbf{94}, 053619 (2016).

\bibitem{Yong2019PRB}
Y. Xu and Y. Hu,
Phys. Rev. B \textbf{99}, 174309 (2019).

\bibitem{Liu2020SB}
Y.-H. Lu, B.-Z. Wang, and X.-J. Liu,
Sci. Bull. \textbf{65}, 2080 (2020).

\bibitem{3DWeylband2021Science}
Z.-Y. Wang, X.-C. Cheng, B.-Z. Wang, J.-Y. Zhang, Y.-H. Lu, C.-R. Yi, S. Niu, Y. Deng, X.-J. Liu, S. Chen, and J.-W. Pan,
Science \textbf{372}, 271 (2021).

\bibitem{footnote1}{
If the non-Hermitian Hamiltonian is quadratic, then we can treat each atom independently so that the local occupancy is given by
$N_{\text{a},j} = \sum_x \bra{0} \hat{a}_x e^{\ii \hat{\mathcal{H}}^\dagger t} \hat{a}_j^\dagger \hat{a}_j e^{-\ii \hat{\mathcal{H}} t} \hat{a}_x^\dagger \ket{0}$~\cite{Supplementary}. }

\bibitem{Du2018Science}
Y. Wu, W. Liu, J. Geng, X. Song, X. Ye, C.-K. Duan, X. Rong, and J. Du,
Science \textbf{364}, 878 (2018).

\bibitem{Deng2021PRL}
W. Zhang, X. Ouyang, X. Huang, X. Wang, H. Zhang, Y. Yu, X. Chang, Y. Liu, D.-L. Deng, and L.-M. Duan,
Phys. Rev. Lett. \textbf{127,} 090501 (2021).

\bibitem{Deng2022NC}
Y. Yu, L.-W. Yu, W. Zhang, H. Zhang, X. Ouyang, Y. Liu, D.-L. Deng, and L.-M. Duan,
arXiv:2112.13785 (2021).

\bibitem{Yan2022}
Q. Liang, D. Xie, Z. Dong, H. Li, H. Li, B. Gadway, W. Yi, and 
B. Yan,
arXiv:2201.09478 (2022).

\end{thebibliography}
\end{document}


\title{Supplemental Material for Non-Hermitian Absorption Spectroscopy}
\author{Kai Li$^{1}$}
\author{Yong Xu$^{1,2}$}
\email{yongxuphy@tsinghua.edu.cn}
\affiliation{$^{1}$Center for Quantum Information, IIIS, Tsinghua University, Beijing 100084, People's Republic of China}
\affiliation{$^{2}$Shanghai Qi Zhi Institute, Shanghai 200030, People's Republic of China}

\begin{abstract}
\end{abstract}
\maketitle

\setcounter{equation}{0} \setcounter{figure}{0} \setcounter{table}{0} %
\renewcommand{\theequation}{S\arabic{equation}} \renewcommand{\thefigure}{S%
	\arabic{figure}}

In the Supplemental Material, we will
map the many-body dynamics described by the master equation to the single-particle dynamics governed by the effective non-Hermitian Hamiltonian in Section S-1,
derive the population of the auxiliary levels based on linear response theory in Section S-2,
prove that the population is independent of boundary conditions in the thermodynamic limit in Section S-3,
discuss the influence of non-Hermitian skin effects (NHSEs) on non-Hermitian absorption spectroscopy
based on the non-Hermitian Rice-Mele model in Section S-4,
show the possibility of probing non-Bloch energies for systems that exhibit NHSEs and its problems in Section S-5,
provide more results about the non-Hermitian SSH model in Section S-6,
and finally discuss the application of the absorption spectroscopy in interacting systems in Section S-7.

\section{S-1. Equivalence between the dynamics of correlation functions governed by the master equation and a non-Hermitian Hamiltonian} 
For a dissipative cold atom system, its dynamics is described by the following master equation, 
\begin{equation}
	\begin{aligned}
		\label{MasterEquation}
		\frac{d\rho(t)}{dt}
		&= -\ii [\hat{\mathcal{H}}_h, \rho(t)] + \sum_\mu (2\hat{L}_\mu \rho(t) \hat{L}_\mu^\dagger - \{ \hat{L}_\mu^\dagger \hat{L}_\mu, \rho(t) \} ) \\
		&=  -\ii [\hat{\mathcal{H}} \rho(t) - \rho(t) \hat{\mathcal{H}}^\dagger] + \sum_\mu 2\hat{L}_\mu \rho(t) \hat{L}_\mu^\dagger,
	\end{aligned}
\end{equation}
where $\rho(t)$ is the density matrix, $\hat{L}_\mu$ is the Lindblad operator (we here consider $\hat{L}_\mu = \sum_i D_{\mu i} \hat{c}_i$ with $D$ being an $L \times L$ matrix), and
$\hat{\mathcal{H}}_h$ describes a generic non-interacting Hermitian system with
$L$ degrees of freedom so that it can be written in a quadratic form as $\hat{\mathcal{H}}_h = \sum_{i,j=1}^L [H_h]_{ij} \hat{c}_i^\dagger \hat{c}_j$
($H_h$ is an $L \times L$ Hermitian matrix), and 
$\hat{\mathcal{H}} = \hat{\mathcal{H}}_h - \ii \sum_\mu \hat{L}_\mu^\dagger \hat{L}_\mu = \sum_{ij} H_{ij} \hat{c}_i^\dagger \hat{c}_j $
is the effective non-Hermitian Hamiltonian with $H_{ij} = [H_h]_{ij} - \ii \sum_\mu D_{\mu i}^* D_{\mu j}$
($H$ is an $L \times L$ non-Hermitian matrix). Here, $\hat{c}_j$ ($\hat{c}_j^\dagger$) is either a fermionic or a bosonic annihilation (creation) operator for the $j$th degree of freedom. In our specific case for measurements, $\hat{\mathcal{H}}_h$ consists of the system Hamiltonian and auxiliary levels.
Because we are mainly interested in the dynamics of the occupancy on the auxiliary levels (or two-point correlation functions in a more general case),
we now ask whether we can reduce the dynamics governed by the master equation to the dynamics governed by the non-Hermitian Hamiltonian $H$.

\subsection{A. Case 1}
We first consider an initial state $\rho_0$ with all auxiliary levels occupied,  
i.e., $\rho_0 = \ket{\psi_0^{(M)}}\bra{\psi_0^{(M)}}$ with $|\psi_0^{(M)}\rangle=\prod_j\hat{a}_j^\dagger|0\rangle=\prod_k\hat{a}_k^\dagger|0\rangle$.
We note that while the results are valid for both bosonic and fermionic cases, it is more realistic to consider a finite temperature
ensemble as an initial state for the bosonic case (see the detailed discussion in the next subsection).

\begin{theorem}\label{Them0}
	If the initial state is a product state $\ket{\psi_0^{(M)}} = \prod_{j=1}^{N_1} \hat{c}_j^\dagger\ket{0}$ with $N_1$ fermionic or bosonic particles ($N_1\le L$), then from Eq.~(\ref{MasterEquation}), we can derive the correlation function at time $t$ as
	\begin{equation}
		C_{ij}=\text{Tr}[\rho(t) \hat{c}_i^\dagger \hat{c}_j ]
		= \sum_{x=1}^{N_1} \bra{x} e^{\ii \hat{\mathcal{H}}^\dagger t} \hat{c}_i^\dagger \hat{c}_j e^{-\ii \hat{\mathcal{H}} t} \ket{x},
	\end{equation}
where $C$ is an $L\times L$ matrix.
The result tells us that one can reduce the dynamics described by the master equation to the dynamics of single-particle states $\ket{x} = \hat{c}_x^\dagger\ket{0}$
governed by the single-particle non-Hermitian Hamiltonian $H$. For a hybrid system consisting of auxiliary and system levels as in our case,
we can encode the auxiliary levels in the first $N$ degrees of freedom ($N_1=N$) and system levels in other degrees of freedom.
In this case, 
given a many-particle initial state $|\psi_0^{(M)}\rangle=\prod_j\hat{a}_j^\dagger|0\rangle=\prod_k\hat{a}_k^\dagger|0\rangle$ on auxiliary energy levels 
($\hat{a}_k^\dagger=\frac{1}{\sqrt{N} }\sum_{x=1}^{N}e^{ikx}\hat{a}_x^\dagger$),
	the atom number on the $j$th auxiliary level is given by $C_{jj} (t)
	=\text{Tr}(\rho(t) \hat{a}_j^\dagger \hat{a}_j )=
	\sum_{x=1}^N \bra{0}\hat{a}_x e^{\ii \hat{\mathcal{H}}^\dagger t} \hat{a}_j^\dagger \hat{a}_j e^{-\ii \hat{\mathcal{H}} t} \hat{a}_x^\dagger\ket{0}$.
	Under PBCs, the atom number on the auxiliary levels at momentum $k$ is given by
	$N_{\text{a},k} = \text{Tr}(\rho(t) \hat{a}_k^\dagger \hat{a}_k)= \bra{0} \hat{a}_k e^{\ii \hat{\mathcal{H}}^\dagger t} \hat{a}_k^\dagger \hat{a}_k e^{-\ii \hat{\mathcal{H}} t} \hat{a}_k^\dagger \ket{0}$, allowing us to use a single-particle state
	$|\psi_0\rangle=\hat{a}_k^\dagger|0\rangle$ as an initial state. For OBCs, we still have this result in the thermodynamic limit. 
\end{theorem}
	
\begin{proof}
We start by deriving the time evolution of the correlation function $C_{ij}(t)$ as 
\begin{equation}
	\label{derivation_CMatrixDynamics_first}
\frac{dC_{ij} (t)}{dt}=\text{Tr}(\frac{d\rho(t)}{dt} \hat{c}_i^\dagger \hat{c}_j)
=-i\text{Tr}[(\hat{\mathcal{H}} \rho(t)-\rho(t) \hat{\mathcal{H}}^\dagger ) \hat{c}_i^\dagger \hat{c}_j ]
+2\sum_\mu \text{Tr}[\hat{L}_\mu \rho(t) \hat{L}_\mu^\dagger \hat{c}_i^\dagger \hat{c}_j].
\end{equation}
The first term can be reduced to
\begin{equation}
-i\text{Tr}[(\hat{\mathcal{H}} \rho(t)-\rho(t) \hat{\mathcal{H}}^\dagger ) \hat{c}_i^\dagger \hat{c}_j ]
=[XC]_{ij}+[CX^\dagger]_{ij} \pm 2\sum_{j_1 j_2} M_{j_1 j_2} \text{Tr}[\rho(t) \hat{c}_i^\dagger \hat{c}_{j_1}^\dagger \hat{c}_j \hat{c}_{j_2}],
\end{equation}
where $M_{j_1 j_2}=\sum_\mu D_{\mu j_1}^* D_{\mu j_2}$, $X_{ij} = \ii [H_h]_{ji} - M_{ji}=i[H^*]_{ij}$, and we take $+$ and $-$ for fermions and bosons, respectively. Besides the first two terms, there still exist four-point correlation terms. If we neglect the quantum jump term, then the four-point correlation terms arise if the initial state contains more than one particles (if there is only one particle, this term vanishes). In fact, the four-point correlation terms can be exactly canceled by the quantum jump term, which can be simplified to 
\begin{equation}
2\sum_\mu \text{Tr}[\hat{L}_\mu \rho(t) \hat{L}_\mu^\dagger \hat{c}_i^\dagger \hat{c}_j]
=\mp 2\sum_{j_1 j_2} M_{j_1 j_2} \text{Tr}[\rho(t) \hat{c}_i^\dagger \hat{c}_{j_1}^\dagger \hat{c}_j \hat{c}_{j_2}],
\end{equation}
where we take $-$ and $+$ for fermions and bosons, respectively. 
Thus for the quadratic open quantum system described by \Eq{MasterEquation}, the dynamics of the correlation function $C_{ij}(t)$ is given by \cite{Prosen2008NJP,Wang2019PRL}
\begin{equation}
\label{CMatrixDynamics}
\frac{dC(t)}{dt} = XC(t) + C (t) X^\dagger.
\end{equation}

Now we flatten the matrix $C(t)$ to a column vector $V_C (t)$, transforming \Eq{CMatrixDynamics} to
\begin{equation}
\label{Eq42}
\frac{dV_C (t)}{dt} = (I \otimes X + X^* \otimes I) V_C (t).
\end{equation}
From \Eq{Eq42}, we directly obtain
\begin{equation}
V_C (t) = e^{(I \otimes X + X^* \otimes I) t} \, V_C (0)
= e^{(I \otimes X) t} e^{ (X^* \otimes I) t} \, V_C (0)
= (I \otimes e^{Xt})(e^{X^* t} \otimes I) V_C (0)
= (e^{X^* t} \otimes e^{Xt}) V_C (0).
\end{equation}
Since $C(0) = \text{diag} \{1,1,...,1,0,...,0 \}$ with $N_1$ non-zero entries,
the vectorized initial state is given by $[V_C (0)]^T = (\mathbf{e}_1^T , ... , \mathbf{e}_{N_1}^T , \mathbf{0}^T , ..., \mathbf{0}^T )$, where $\mathbf{e}_i$ is an $L$-dimensional column vector satisfying $[\mathbf{e}_i]_j = \delta_{ij}$.
Let $F=e^{Xt}$, we have
\begin{equation}
F^* \otimes F =
\begin{pmatrix}
F_{11}^* F & \cdots & F_{1L}^* F \\
\vdots     & \ddots & \vdots \\
F_{L1}^* F & \cdots & F_{LL}^* F
\end{pmatrix}
\end{equation}
and
\begin{equation} \label{Vcevolution}
V_C (t) = (F^* \otimes F) V_C (0) =
\begin{pmatrix}
\sum_{x=1}^{N_1} F_{1x}^* F \mathbf{e}_x \\
\vdots \\
\sum_{x=1}^{N_1} F_{Lx}^* F \mathbf{e}_x
\end{pmatrix}.
\end{equation}
Then
\begin{equation}
C (t) = \big( \sum_{x=1}^{N_1} F_{1x}^* F \mathbf{e}_x , ...  , \sum_{x=1}^{N_1} F_{Lx}^* F \mathbf{e}_x \big)
\end{equation}
and
\begin{equation}
C_{ij} (t)
= \sum_{x=1}^{N_1} F_{jx}^* [F \mathbf{e}_x]_i
= \sum_{x=1}^{N_1} F_{jx}^* \sum_k F_{ik} [\mathbf{e}_x]_k
= \sum_{x=1}^{N_1} F_{jx}^* \sum_k F_{ik} \delta_{xk}
= \sum_{x=1}^{N_1} F_{jx}^* F_{ix}.
\end{equation}
Finally, since
\begin{equation}
\begin{aligned}
\sum_{x=1}^{N_1} \bra{x} e^{\ii \hat{\mathcal{H}}^\dagger t} \hat{c}_i^\dagger \hat{c}_j e^{-\ii \hat{\mathcal{H}} t} \ket{x}
& = \sum_{x=1}^{N_1} \bra{x} e^{\ii \hat{\mathcal{H}}^\dagger t} \ket{i} \bra{j} e^{-\ii \hat{\mathcal{H}} t} \ket{x}
= \sum_{x=1}^{N_1} [e^{\ii H^\dagger t}]_{xi} [e^{-\ii H t}]_{jx} \\
& = \sum_{x=1}^{N_1} [e^{X^T t}]_{xi} [e^{X^* t}]_{jx}
= \sum_{x=1}^{N_1} F_{ix} F_{jx}^*,
\end{aligned}
\end{equation}
we conclude that 
\begin{equation}
C_{ij} (t) = \sum_{x=1}^{N_1} \bra{x} e^{\ii \hat{\mathcal{H}}^\dagger t} \hat{c}_i^\dagger \hat{c}_j e^{-\ii \hat{\mathcal{H}} t} \ket{x}.
\end{equation}

We remark that the conclusion holds for a master equation including a quantum jump term. If we neglect the quantum jump term,
the conclusion usually does not hold for a many-particle system since four point-correlation functions
(which are canceled exactly by the contribution from the quantum jump term) arise
in Eq.~(\ref{CMatrixDynamics}).
We can also show this by projecting the master equation Eq.~(\ref{MasterEquation}) to the $N_1$-particle subspace using the projection operator $\hat{P}_{N_1}$.
Since the initial state $\ket{\psi_0^{(M)}}$ is an $N_1$-particle many-body state and the Lindblad operator contains only annihilation operators, we have 
$\hat{P}_{N_1} \hat{L}_\mu \rho(t) \hat{L}_\mu^\dagger \hat{P}_{N_1} = 0$ and
\begin{equation}
\frac{d\rho^{(N_1)}(t)}{dt}
=  -\ii (\hat{\mathcal{H}} \rho^{(N_1)}(t) - \rho^{(N_1)}(t) \hat{\mathcal{H}}^\dagger)
\end{equation}
with $\rho^{(n)}(t) = \hat{P}_n \rho(t) \hat{P}_n$.
Then $\rho^{(N_1)}(t) = e^{-\ii \hat{\mathcal{H}} t} \rho_0 e^{\ii \hat{\mathcal{H}}^\dagger t} = \ket{\psi(t)} \bra{\psi(t)}$ 
with $\ket{\psi(t)}\equiv e^{-\ii \hat{\mathcal{H}} t} \ket{\psi_0^{(M)}}$,
and $\text{Tr}(\rho^{(N_1)}(t) \hat{c}_i^\dagger \hat{c}_j) = \bra{\psi (t)} \hat{c}_i^\dagger \hat{c}_j \ket{\psi(t)}$.
However, $C_{ij}(t)=\text{Tr}(\rho(t) \hat{c}_i^\dagger \hat{c}_j)
=\sum_{n,m=0}^{N_1} \text{Tr}(\hat{P}_{n}\rho(t) \hat{P}_{m}\hat{c}_i^\dagger \hat{c}_j)
=\sum_{n=0}^{N_1} \text{Tr}(\rho^{(n)}(t) \hat{c}_i^\dagger \hat{c}_j)
 \neq \text{Tr}(\rho^{({N_1})}(t) \hat{c}_i^\dagger \hat{c}_j)$ and thus $C_{ij} (t) \neq \bra{\psi (t)} \hat{c}_i^\dagger \hat{c}_j \ket{\psi(t)}$.
When we consider the single-particle case, that is, ${N_1}=1$, we have $C_{ij}(t) = \text{Tr} (\rho^{(1)}(t) \hat{c}_i^\dagger  \hat{c}_j) = \bra{\psi (t)} \hat{c}_i^\dagger \hat{c}_j \ket{\psi(t)}$ because $\text{Tr} (\rho^{(0)}(t) \hat{c}_i^\dagger  \hat{c}_j) = 0$.

We are now in a position to apply the results to a hybrid system consisting of auxiliary and system levels
by treating the first ${N_1}$ degrees of freedom as auxiliary levels ($N_1=N$), i.e., $\hat{a}_j^\dagger = \hat{c}_j^\dagger$.
Given the fact that all the auxiliary levels are initially occupied, the measured local occupancy is given by
$N_{\text{a},j} = \text{Tr}(\rho(t) \hat{a}_j^\dagger \hat{a}_j) = \sum_{x=1}^N \bra{0} \hat{a}_x e^{\ii \hat{\mathcal{H}}^\dagger t} \hat{a}_j^\dagger \hat{a}_j e^{-\ii \hat{\mathcal{H}} t} \hat{a}_x^\dagger \ket{0}$.
We also show that by performing $k$-resolved measurements, we can obtain the occupancy of momentum states which is given by
\begin{equation}
\begin{aligned}
N_{\text{a},k} &= \text{Tr}(\rho(t) \hat{a}_k^\dagger \hat{a}_k)
= \frac{1}{N} \sum_{j j'} e^{\ii k (j-j')} \text{Tr} (\rho(t) \hat{a}_j^\dagger \hat{a}_{j'})
\\&= \frac{1}{N} \sum_{j j'} e^{\ii k (j-j')} \sum_{x=1}^{N} \bra{0} \hat{a}_x e^{\ii \hat{\mathcal{H}}^\dagger t} \hat{a}_j^\dagger \hat{a}_{j'} e^{-\ii \hat{\mathcal{H}} t} \hat{a}_x^\dagger \ket{0}
\\&= \sum_{x=1}^{N} \bra{0} \hat{a}_x e^{\ii \hat{\mathcal{H}}^\dagger t} \hat{a}_k^\dagger \hat{a}_k e^{-\ii \hat{\mathcal{H}} t} \hat{a}_x^\dagger \ket{0}
\\&= \sum_{x=1}^{N} \frac{1}{N} \sum_{q q'} e^{\ii (q-q') x} \bra{0} \hat{a}_q e^{\ii \hat{\mathcal{H}}^\dagger t} \hat{a}_k^\dagger \hat{a}_k e^{-\ii \hat{\mathcal{H}} t} \hat{a}_{q'}^\dagger \ket{0}
\\&= \sum_{q} \bra{0} \hat{a}_q e^{\ii \hat{\mathcal{H}}^\dagger t} \hat{a}_k^\dagger \hat{a}_k e^{-\ii \hat{\mathcal{H}} t} \hat{a}_{q}^\dagger \ket{0},
\end{aligned}
\end{equation}
with $\hat{a}_k^\dagger = \frac{1}{\sqrt{N}} \sum_j e^{\ii k j} \hat{a}_j^\dagger $ being the creation operator of a particle on the auxiliary levels with momentum $k$ and $k=2\pi n_k/N$ ($n_k=0,1,...,N-1$).
If the effective Hamiltonian is under PBCs (denoted by $\hat{\mathcal{H}}_\text{p}$) and preserves the translation symmetry, then we have
\begin{equation}
N_{\text{a},k} 
= \bra{0} \hat{a}_k e^{\ii \hat{\mathcal{H}}_\text{p}^\dagger t} \hat{a}_k^\dagger \hat{a}_k e^{-\ii \hat{\mathcal{H}}_\text{p} t} \hat{a}_k^\dagger \ket{0} 
= \bra{0} \hat{a}_k e^{\ii \hat{\mathcal{H}}_\text{p}^\dagger t} \sum_q \hat{a}_q^\dagger \hat{a}_q e^{-\ii \hat{\mathcal{H}}_\text{p} t} \hat{a}_k^\dagger \ket{0} 
= \bra{0} \hat{a}_k e^{\ii \hat{\mathcal{H}}_\text{p}^\dagger t} \hat{N}_\text{a} e^{-\ii \hat{\mathcal{H}}_\text{p} t} \hat{a}_k^\dagger \ket{0},
\end{equation}
with $\hat{N}_\mathrm{a}=\sum_q \hat{a}_q^\dagger \hat{a}_q = \sum_{j=1}^{N} \hat{a}_j^\dagger \hat{a}_j$ being the total particle number operator on the auxiliary levels.
For the effective Hamiltonian that is under OBCs (denoted by $\hat{\mathcal{H}}_\text{o}$), we can also prove that 
$N_{a,k} = \bra{0} \hat{a}_k e^{\ii \hat{\mathcal{H}}_\text{p}^\dagger t} \hat{N}_\text{a} e^{-\ii \hat{\mathcal{H}}_\text{p} t} \hat{a}_k^\dagger \ket{0}$
in the thermodynamic limit.
Let $\ket{\psi_\text{a}^k}=\hat{a}_k^\dagger \ket{0}$, we have
\begin{equation}
N_{\text{a},k} 
= \sum_{q} \bra{\psi_\text{a}^q} e^{\ii \hat{\mathcal{H}}_\text{o}^\dagger t} \ket{\psi_\text{a}^k} \bra{\psi_\text{a}^k} e^{-\ii \hat{\mathcal{H}}_\text{o} t} \ket{\psi_\text{a}^q} 
=  \bra{\psi_\text{a}^k} e^{-\ii \hat{\mathcal{H}}_\text{o} t} \sum_{q} \ket{\psi_\text{a}^q} \bra{\psi_\text{a}^q} e^{\ii \hat{\mathcal{H}}_\text{o}^\dagger t} \ket{\psi_\text{a}^k} 
= \bra{\psi_\text{a}^k} e^{-\ii \hat{\mathcal{H}}_\text{o} t} \hat{N}_\text{a} e^{\ii \hat{\mathcal{H}}_\text{o}^\dagger t} \ket{\psi_\text{a}^k}. 
\end{equation}
Then based on Theorem~\ref{Them3} which states that $\bra{\psi_\text{a}^k} e^{\ii \hat{\mathcal{H}}_\text{o}^\dagger t} \hat{N}_\text{a} e^{-\ii \hat{\mathcal{H}}_\text{o} t} \ket{\psi_\text{a}^k} 
= \bra{\psi_\text{a}^k} e^{\ii \hat{\mathcal{H}}_\text{p}^\dagger t} \hat{N}_\text{a} e^{-\ii \hat{\mathcal{H}}_\text{p} t} \ket{\psi_\text{a}^k}$ in the thermodynamic limit
when the bulk preserves the translation symmetry and the hopping range is finite, we have
\begin{equation}
\begin{aligned}
N_{\text{a},k} 
&= \bra{\psi_\text{a}^k} e^{-\ii \hat{\mathcal{H}}_\text{p} t} \hat{N}_\text{a} e^{\ii \hat{\mathcal{H}}_\text{p}^\dagger t} \ket{\psi_\text{a}^k}
= \bra{\psi_\text{a}^k} e^{-\ii \hat{\mathcal{H}}_\text{p} t} \ket{\psi_\text{a}^k} \bra{\psi_\text{a}^k} e^{\ii \hat{\mathcal{H}}_\text{p}^\dagger t} \ket{\psi_\text{a}^k}
= \bra{\psi_\text{a}^k} e^{\ii \hat{\mathcal{H}}_\text{p}^\dagger t} \ket{\psi_\text{a}^k} \bra{\psi_\text{a}^k} e^{-\ii \hat{\mathcal{H}}_\text{p} t} \ket{\psi_\text{a}^k} 
\\&= \bra{\psi_\text{a}^k} e^{\ii \hat{\mathcal{H}}_\text{p}^\dagger t} \hat{N}_\text{a} e^{-\ii \hat{\mathcal{H}}_\text{p} t} \ket{\psi_\text{a}^k}.
\end{aligned}
\end{equation}
These properties allow us to use a single-particle state $\ket{\psi_0}=\hat{a}_k^\dagger \ket{0}$ evolved by the effective Hamiltonian $\hat{\mathcal{H}}_\text{p}$ and the observable $\hat{N}_\mathrm{a}$ to derive the protocol; the derivation is given in the next section.
We have also numerically confirmed that the master equation and the non-Hermitian approach give the same results.
\end{proof}
	
\subsection{B. Case 2}
For bosons, more than one atoms are allowed to occupy a single state. We thus may consider an initial state with $N_{0,k}$ atoms on the state at momentum $k$,
i.e., $\ket{\psi_0^{(M)}} = \frac{1}{\sqrt{N_{0,k}!}}(\hat{a}_{k}^\dagger)^{N_{0,k}}\ket{0}$. Then
\begin{equation}
C_{j j^\prime}^{(k)}(t=0)=\bra{\psi_0^{(M)}} \hat{c}_j^\dagger \hat{c}_{j^\prime} \ket{\psi_0^{(M)}}=n_{0,k} e^{\ii k(j^\prime-j)},
\end{equation}
where $n_{0,k}=N_{0,k}/N$ for $1\le j,j^\prime \le N$. $C_{j j^\prime}(t=0)$ vanishes when $j > N$ or $j^\prime > N$. In the following, we will prove that
\begin{equation}
C_{j j^\prime}^{(k)}(t)=N_{0,k} \bra{0}\hat{a}_k e^{\ii \hat{\mathcal{H}}^\dagger t} \hat{a}_j^\dagger \hat{a}_{j^\prime} e^{-\ii \hat{\mathcal{H}} t} \hat{a}_k^\dagger\ket{0}.
\end{equation}
\begin{proof}
With the initial correlation function, we write it in the form of a column vector, that is, 
\begin{equation}
V_C^{(k)}(0)=n_{0,k} (\sum_{j_1=1}^N e^{ikj_1} \bm{e}_{j_1})
\otimes (\sum_{j_2=1}^N e^{-ikj_2} \bm{e}_{j_2} ).
\end{equation}
Based on Eq.~(\ref{Vcevolution}), we have
\begin{equation}
V_C^{(k)}(t)=(F^* \otimes F) V_C^{(k)} (0)=n_{0,k}(\sum_{j_1=1}^N e^{ikj_1} F^* \bm{e}_{j_1} )
\otimes (\sum_{j_2=1}^N e^{-ikj_2} F\bm{e}_{j_2} ).
\end{equation}
It follows that
\begin{equation}
C_{j j^\prime}^{(k)}(t)=n_{0,k} 
\sum_{j_1,j_2=1}^N e^{ik(j_1-j_2)} F^*_{j^\prime j_1} F_{j j_2}.
\end{equation}
One can also prove that
\begin{equation}
N_{0,k} \bra{0}\hat{a}_k e^{\ii \hat{\mathcal{H}}^\dagger t} \hat{a}_j^\dagger \hat{a}_{j^\prime} e^{-\ii \hat{\mathcal{H}} t} \hat{a}_k^\dagger\ket{0}
=n_{0,k} 
\sum_{j_1,j_2=1}^N e^{ik(j_1-j_2)} F^*_{j^\prime j_1} F_{j j_2}=C_{j j^\prime}^{(k)}(t).
\end{equation}
\end{proof}

The result indicates that for multiple bosons on a single state, we can still reduce the problem to the dynamics of a single particle with the correlation function multiplied by the initial particle number. 

We are now in a position to consider a finite temperature ensemble for the initial state [the temperature $T=1/(k_B \beta)$], i.e., 
$\hat{\rho}_0 = e^{-\beta(\hat{\mathcal{H}}_\text{a} -\mu \hat{N}_\text{a})}/Z$ with $\hat{\mathcal{H}}_\text{a}$ being the Hamiltonian of the auxiliary levels, $\mu$ being the chemical potential, $\hat{N}_\text{a}$ being the total 
particle number operator on the auxiliary levels,
and 
$Z=\text{Tr}[e^{-\beta(\hat{\mathcal{H}}_\text{a} -\mu \hat{N}_\text{a})}]$ being the partition function. With the initial state, the correlation function at $t=0$ is given by 
\begin{equation}
C_{j j^\prime }(t=0) = \text{Tr}(\hat{\rho}_0 \hat{a}_j^\dagger \hat{a}_{j'}) 
= \sum_n \langle j' | n \rangle  \langle n | j \rangle f_{\text{BE}}(E_n),
\end{equation}
where $\ket{j} = \hat{a}_j^\dagger \ket{0}$, $\ket{n}$ is the $n$th single-particle eigenstate of $\hat{\mathcal{H}}_\text{a}$ with eigenenergy $E_n$, and 
$f_{\text{BE}}(E_n)=\frac{1}{e^{\beta(E_n - \mu)} - 1}$ is the Bose-Einstein distribution with 
the chemical potential $\mu$ determined by the total particle number $N_0$, that is, 
$\sum_{n} f_{\text{BE}}(E_n) = N_0$. Given the fact that the physics for open boundaries is equivalent to the physics for periodic boundaries in the thermodynamic limit, 
we consider the periodic boundaries so that we can replace the $n$ index by the momentum $k$, that is,
\begin{equation}
	C_{j j^\prime }(t=0) = \text{Tr}(\hat{\rho}_0 \hat{a}_j^\dagger \hat{a}_{j'}) 
	= \frac{1}{N}\sum_k e^{\ii k (j^\prime -j)} f_{\text{BE}}(E_k).
\end{equation}
We write $C_{j j^\prime }(t=0)=\sum_{k} C_{j j^\prime}^{(k)}(t=0) [f_{\text{BE}}(E_k)/N_{0,k}]$, then
$V_C(0)=\sum_k V_C^{(k)}(0) [f_{\text{BE}}(E_k)/N_{0,k}]$, so that
$V_C(t)=(F^* \otimes F) V_C (0)=\sum_k (F^* \otimes F) V_C^{(k)}(0) [f_{\text{BE}}(E_k)/N_{0,k}]$.
We thus obtain
\begin{equation}\label{CFtemp}
C_{j j^\prime }(t)=\sum_k f_{\text{BE}}(E_k) \bra{0}\hat{a}_k e^{\ii \hat{\mathcal{H}}^\dagger t} \hat{a}_j^\dagger \hat{a}_{j^\prime} e^{-\ii \hat{\mathcal{H}} t} \hat{a}_k^\dagger\ket{0},
\end{equation}
which reads in momentum space
\begin{equation}
N_{\text{a},k} = \text{Tr}(\rho(t) \hat{a}_k^\dagger \hat{a}_k)=
f_{\text{BE}}(E_k) \bra{0}\hat{a}_k e^{\ii \hat{\mathcal{H}}^\dagger t} \hat{a}_k^\dagger \hat{a}_{k} e^{-\ii \hat{\mathcal{H}} t} \hat{a}_k^\dagger\ket{0}.
\end{equation}
Thus, we can use the finite temperature ensemble for bosons to perform the measurement of the energy spectrum (see Section S-4 B for numerical results).

\subsection{C. Experimental realizations}
Finally, we would like to show that in our protocol, a non-Hermitian Hamiltonian $\hat{\mathcal{H}}$ can be naturally realized by involving atomic levels that experience spontaneous emissions. 
For example, suppose that we want to realize the non-Hermitian Hamiltonian given by
\begin{equation}
\label{EffHam_example}
\hat{\mathcal{H}}_t=\hat{\mathcal{H}}_h-\ii\gamma \sum_j \hat{c}_{j2}^\dagger \hat{c}_{j2},
\end{equation}
where $\hat{\mathcal{H}}_h$ is the Hermitian part of the Hamiltonian involving two hyperfine states and the coupling between one hyperfine level and an auxiliary level used for measurements. To realize the Hamiltonian, we consider three atomic levels at site $j$ described by the creation operators 
$\hat{c}_{j1}^\dagger$, $\hat{c}_{j2}^\dagger$ and $\hat{a}_{j}^\dagger$. With lasers and microwaves, one can realize the Hermitian Hamiltonian $\hat{\mathcal{H}}_h$. For $^{87}$Rb atoms, we may use $|\uparrow\rangle=|F=1,m_F=-1\rangle$ and $|\downarrow\rangle=|F=1,m_F=0\rangle$ as system levels and $|a\rangle=|F=2,m_F=-2\rangle$ 
as the auxiliary level in the ground state manifold $^{2}$S$_{1/2}$. The auxiliary level can be coupled to one of the system levels via microwaves. 
To implement the non-Hermitian part, we consider applying a resonant laser to couple the second hyperfine state to one of the 
P$_{1/2}$ or P$_{3/2}$ energy levels (described by $\hat{p}_j^\dagger$), which spontaneously decays. 
For fermions such as $^{173}$Yb, one may choose $|\uparrow\rangle=|F=5/2,m_F=3/2\rangle$ and $|\downarrow\rangle=|F=5/2,m_F=1/2\rangle$ 
for the system levels and $|a\rangle=|F=5/2,m_F=5/2\rangle$ for 
the auxiliary level in the ground state manifold $^{1}$S$_{0}$. The loss term can be realized by driving atoms on the $|\downarrow\rangle$ level to one of the 
$^3$P$_1$  levels~\cite{JoarXiv2021}.
The entire system is thus described by the master equation with
\begin{equation}
\hat{\mathcal{H}}=\hat{\mathcal{H}}_h+\frac{\Omega_p}{2}(\hat{c}_{j2}^\dagger \hat{p}_j+\hat{p}_j^\dagger \hat{c}_{j2})-i\Gamma \sum_j \hat{p}_j^\dagger \hat{p}_j,
\end{equation}		
and $\hat{L}_j=\sqrt{\Gamma}\hat{p}_j$ where $\Gamma$ denotes the loss rate of the P level.
Based on our proof, the particle number on the auxiliary levels (a special correlation function) 
$N_{\text{a}}=\langle \psi (t)|\hat{N}_{\text{a}}|\psi (t)\rangle$ with $|\psi (t)\rangle=e^{-i\hat{\mathcal{H}}t}\hat{a}_k^\dagger|0\rangle$,
indicating that the dynamics of the particle number is completely determined by $\hat{\mathcal{H}}$.
We can further adiabatically eliminate the P level, leading to an effective target Hamiltonian~(\ref{EffHam_example}) with $\gamma =\Omega_p^2/(4\Gamma$) from $\hat{\mathcal{H}}$,
given the fact that $\Gamma$ (in the order of MHz, e.g., $\Gamma \approx 5.75\ \text{MHz}$ for $5\,^{2}P_{1/2}$ levels of $^{87}$Rb atoms) is usually much larger than the energy scale (in the order of kHz) of the system.
We have also checked the validity of the adiabatic elimination by numerical simulations.

\section{S-2. Derivation of non-Hermitian absorption spectroscopy protocol based on linear response theory}
In this section, we will prove the following theorem using linear response theory~\cite{TormaBook}.
Note that we here consider a time-dependent Hamiltonian without first applying rotating wave approximations
(in the main text, we first apply rotating wave approximations so that the Hamiltonian is time independent).
While the methods are slightly different, the results are the same.

\begin{theorem}\label{Them1}
Consider a one-dimensional (1D) translation-invariant non-Hermitian system consisting of $N$ unit cells described by the Hamiltonian,
\begin{equation}
\label{SysHam}
\hat{\mathcal{H}}_\text{s} = \sum_{i\alpha,j\beta} [H_\text{s}]_{i\alpha,j\beta} \hat{c}_{i\alpha}^{\dagger} \hat{c}_{j\beta},
\end{equation}
where $\hat{c}_{i\alpha}^\dagger$ ($\hat{c}_{i\alpha}$) is the fermionic or bosonic creation (annihilation) operator acting on the $\alpha$th degree of freedom of the $i$th unit cell (each unit cell contains $p$ degrees of freedom).
The non-Hermitian Hamiltonian $\hat{\mathcal{H}}_\text{s}$ is purely dissipative, i.e. $\text{Im}(\lambda) < 0$ for any eigenvalue $\lambda$ of $H_\text{s}$.
Suppose that the first degree of freedom of each unit cell
is coupled to an auxiliary energy level by a microwave field such that the
full Hamiltonian reads
\begin{equation}
\label{full_Hamiltonian}
\hat{\mathcal{H}}(t)=\hat{\mathcal{H}}_\text{s}+\sum_{ij}[T_\text{a}]_{ij}\hat{a}_{i}^{\dagger}\hat{a}_{j}+\omega_\text{a}\sum_{i}\hat{a}_{i}^{\dagger}\hat{a}_{i}+\Omega\cos(\omega t)\sum_{i}(\hat{c}_{i1}^{\dagger}\hat{a}_{i}+\hat{a}_{i}^{\dagger}\hat{c}_{i1}),
\end{equation}
where $\hat{a}_{i}^{\dagger}$ ($\hat{a}_{i}$) creates (annihilates) a fermion or a boson on the $i$th auxiliary level,
$T_a$ describes the hopping between auxiliary levels, $\omega_\text{a}$ is the energy difference
between two energy levels described by $\hat{c}_{i1}$ and $\hat{a}_i$ for an atom, and $\omega$ is
the frequency of the microwave field. The final term depicts the interaction between system and auxiliary
levels with $\Omega$ being the Rabi frequency of the microwave field.
If we initially prepare $N_a(0)$ particles on auxiliary levels at momentum $k$
with $k=2\pi n_k/N$ ($n_k=0,1,...,N-1$), then under periodic boundary conditions, the population of the auxiliary levels is given by
\begin{equation}
\label{thrm1}
N_\text{a}(t)
= N_\text{a}(0)\exp (- \frac{\Omega^{2}t}{2} \sum_m \frac{a_{km}^{(1)}  \gamma_{km} - b_{km}^{(1)} \Delta_{km} }{\Delta_{km}^2 + \gamma_{km}^2})
\end{equation}
when $t \gg 1/\gamma_{km}$ and $\Omega$ is sufficiently small compared with the decay rates of system states.
Here, $\Delta_{km} = E_k - \varepsilon_{km}$ where $E_k=-\delta+\varepsilon_{k}^{(a)}$ with $\delta=\omega-\omega_a$ being
the detuning and $\varepsilon_{k}^{(a)}$ being the energy spectrum of auxiliary levels (without including $\omega_a$),
and $\varepsilon_{km}$ and $-\gamma_{km}$ (labeled by the momentum $k$ and the band index $m$) are the real and imaginary parts
of eigenenergies of ${\mathcal{H}}_\text{s}$, respectively.
The parameters $a_{km}^{(1)}=\mathrm{Re} [c_{km}^{(1)}]$ and $b_{km}^{(1)}=\mathrm{Im} [c_{km}^{(1)}]$ where
$c_{km}^{(\alpha)} = \langle \psi_\text{s}^{k\alpha} | u_\text{R}^{km} \rangle \langle u_\text{L}^{km} | \psi_\text{s}^{k\alpha} \rangle$
with $\ket{\psi_\text{s}^{k \alpha}} = \frac{1}{\sqrt{N}} \sum_j e^{\ii k j} \hat{c}_{j\alpha}^\dagger \ket{0}$ and
$\ket{u_\text{R}^{km}}$ being the right eigenstate of the system Hamiltonian, i.e., $\hat{\mathcal{H}}_\text{s} \ket{u_\text{R}^{km}} = (\varepsilon_{km}-\ii \gamma_{km}) \ket{u_\text{R}^{km}}$, and $\bra{u_\text{L}^{km}}$ being the corresponding left one.
\end{theorem}

\begin{proof}
We write the full Hamiltonian as
\begin{equation}
\hat{\mathcal{H}}(t) = \hat{\mathcal{H}}_{0}' + \hat{\mathcal{V}}(t),
\end{equation}
where
\begin{equation}
\hat{\mathcal{H}}_{0}' = \hat{\mathcal{H}}_\text{s}+\sum_{ij}[T_\text{a}]_{ij}\hat{a}_{i}^{\dagger}\hat{a}_{j}+\omega_\text{a}\sum_{i}\hat{a}_{i}^{\dagger}\hat{a}_{i}
\end{equation}
is time independent,
and
\begin{equation}
\hat{\mathcal{V}}(t)=\Omega\cos(\omega t) \sum_{i} (\hat{c}_{i1}^{\dagger}\hat{a}_{i}+\hat{a}_{i}^{\dagger}\hat{c}_{i1})
\end{equation}
is time dependent.
In the interaction picture, we have
\begin{equation}
\label{H_I}
\hat{\mathcal{V}}^{I}(t)
=
e^{\ii \hat{\mathcal{H}}_{0}' t}
\hat{\mathcal{V}}(t)
e^{-\ii \hat{\mathcal{H}}_{0}' t}
=
\Omega\cos(\omega t)
\sum_{i}
e^{\ii \hat{\mathcal{H}}_{0}' t}
(\hat{c}_{i1}^{\dagger}\hat{a}_{i}+\hat{a}_{i}^{\dagger}\hat{c}_{i1})
e^{-\ii \hat{\mathcal{H}}_{0}' t},
\end{equation}
which is usually non-Hermitian.

We now evaluate the first-order derivative of the population of the auxiliary levels with respect to time,
\begin{equation}
\label{hatNadot}
\hat{\dot{N}}_\text{a} (t) =
\ii \big( \hat{\mathcal{H}}^\dagger(t) \hat{N}_\text{a} - \hat{N}_\text{a} \hat{\mathcal{H}}(t) \big)
= \ii \Omega\cos(\omega t)\sum_{i}(\hat{c}_{i1}^{\dagger}\hat{a}_{i}-\hat{a}_{i}^{\dagger}\hat{c}_{i1}) + \ii (\hat{\mathcal{H}}_\text{s}^\dagger - \hat{\mathcal{H}}_\text{s}) \hat{N}_\text{a}
\end{equation}
where $\hat{N}_\text{a} = \sum_{i}\hat{a}_{i}^{\dagger}\hat{a}_{i}$ satisfying $[\hat{\mathcal{H}}_\text{s}, \hat{N}_\text{a}]=0$.
The operator $\hat{\dot{N}}_\text{a} (t)$ is obtained by
\begin{equation}
\begin{aligned}
\dot{N}_\text{a} (t) &= \frac{d}{dt} \bra{\psi(t)} \hat{N}_\text{a} \ket{\psi(t)} \\
&= \Big( \frac{d}{dt}  \bra{\psi(t)} \Big) \hat{N}_\text{a} \ket{\psi(t)} +  \bra{\psi(t)} \hat{N}_\text{a}  \Big( \frac{d}{dt} \ket{\psi(t)} \Big) \\
& = \bra{\psi(t)} \ii \hat{\mathcal{H}}^\dagger(t) \hat{N}_\text{a} - \ii \hat{N}_\text{a} \hat{\mathcal{H}}(t) \ket{\psi(t)} \\
&= \bra{\psi(t)} \hat{\dot{N}}_\text{a} (t) \ket{\psi(t)},
\end{aligned}
\end{equation}
where $\ket{\psi(t)}$ is a state in the Schr{\"o}dinger picture.
We remark that \Eq{hatNadot} is valid for both fermions and bosons.

In the interaction picture, we have
\begin{equation}
\label{Na_I}
\hat{\dot{N}}_\text{a}^{I}(t)
=
e^{ \ii \hat{\mathcal{H}}_{0}'^\dagger t}
\hat{\dot{N}}_\text{a}(t)
e^{- \ii \hat{\mathcal{H}}_{0}' t}
=
\ii \Omega\cos(\omega t) \sum_{i}
e^{ \ii \hat{\mathcal{H}}_{0}'^\dagger t} (\hat{c}_{i1}^{\dagger}\hat{a}_{i}-\hat{a}_{i}^{\dagger}\hat{c}_{i1})
e^{-\ii \hat{\mathcal{H}}_{0}' t}
+ \ii e^{ \ii \hat{\mathcal{H}}_{0}'^\dagger t} (\hat{\mathcal{H}}_\text{s}^\dagger - \hat{\mathcal{H}}_\text{s}) \hat{N}_\text{a} e^{-\ii \hat{\mathcal{H}}_{0}' t}.
\end{equation}
Note that the transformations of the Hamiltonian (\ref{H_I}) and the observable (\ref{Na_I}) to the interaction representations are different
when the Hamiltonian $\hat{\mathcal{H}}_{0}'$ is non-Hermitian.
Based on the Schr{\"o}dinger equation
$\ii \frac{d}{dt} \ket{\psi(t)} = \hat{\mathcal{H}}(t) \ket{\psi(t)}$
and the transformation of the state vector
$\ket{\psi(t)} = e^{-\ii \hat{\mathcal{H}}_{0}' t} \ket{\psi^{I}(t)}$,
we can derive
$\ii \frac{d}{dt} \ket{\psi^{I}(t)} = \hat{\mathcal{V}}^{I}(t) \ket{\psi^{I}(t)}$
with
$\hat{\mathcal{V}}^{I}(t) = e^{\ii \hat{\mathcal{H}}_{0}' t} \hat{\mathcal{V}}(t) e^{-\ii \hat{\mathcal{H}}_{0}' t}$.
However, if we want $\bra{\psi(t)}\hat{O}(t)\ket{\psi(t)} = \bra{\psi^I(t)}\hat{O}^I(t)\ket{\psi^I(t)}$ for an observable $\hat{O}$,
then we must require that
$\hat{O}^I(t) = e^{ \ii \hat{\mathcal{H}}_{0}'^\dagger t} \hat{O} (t) e^{- \ii \hat{\mathcal{H}}_{0}' t}$.

For an initial state $|\psi_{0}\rangle$, the state at time
$t$ in the interaction picture is given by $|\psi^{I}(t)\rangle=\hat{U}^{I}(t,0)|\psi_{0}\rangle$
where $\hat{U}^{I}(t,t_{0})=\mathcal{T}\exp(-\ii \int_{t_{0}}^{t}dt'\hat{\mathcal{V}}^{I}(t'))$
is the time evolution operator with $\mathcal{T}$ denoting the time-ordering operator.
Expanding the time evolution operator up to the first order with respect to $\Omega$, we get $\hat{U}^{I}(t,t_{0})=1 - \ii \int_{t_{0}}^{t}dt'\hat{\mathcal{V}}^{I}(t')+\mathcal{O}(\Omega^2)$ and
\begin{equation}
\begin{aligned}
\dot{N}_\text{a}(t)&
= \langle \psi(t) | \hat{\dot{N}}_\text{a}(t) | \psi(t) \rangle
= \langle\psi^{I}(t)|\hat{\dot{N}}_\text{a}^{I}(t)|\psi^{I}(t)\rangle
\\& = \langle\hat{\dot{N}}_\text{a}^{I}(t)\rangle - \ii \int_{0}^{t}dt'\langle
\hat{\dot{N}}_\text{a}^{I}(t) \hat{\mathcal{V}}^{I}(t') - \hat{\mathcal{V}}^{I\dagger}(t') \hat{\dot{N}}_\text{a}^{I}(t) \rangle + \mathcal{O}(\Omega^3)
\\& = \Omega^{2} \int_{0}^{t} dt'\cos(\omega t)\cos(\omega t')
\sum_{ij} \langle e^{ \ii \hat{\mathcal{H}}_{0}'^\dagger t}
(\hat{c}_{i1}^{\dagger}\hat{a}_{i}-\hat{a}_{i}^{\dagger}\hat{c}_{i1})
e^{-\ii \hat{\mathcal{H}}_{0}' t}
e^{\ii \hat{\mathcal{H}}_{0}' t'}
(\hat{c}_{j1}^{\dagger}\hat{a}_{j}+\hat{a}_{j}^{\dagger}\hat{c}_{j1})
e^{-\ii \hat{\mathcal{H}}_{0}' t'} \rangle
\\& \ \ \ \ + \Omega \int_{0}^{t} dt' \cos (\omega t') \sum_j \langle
e^{ \ii \hat{\mathcal{H}}_{0}'^\dagger t} (\hat{\mathcal{H}}_\text{s}^\dagger - \hat{\mathcal{H}}_\text{s}) \hat{N}_\text{a} e^{-\ii \hat{\mathcal{H}}_{0}' t} e^{\ii \hat{\mathcal{H}}_{0}' t'}
(\hat{c}_{j1}^{\dagger}\hat{a}_{j}+\hat{a}_{j}^{\dagger}\hat{c}_{j1})
e^{-\ii \hat{\mathcal{H}}_{0}' t'} \rangle + \mathrm{H.c.},
\end{aligned}
\end{equation}
where we have adopted the notation $\langle \hat{A} \rangle=\langle \psi_0 | \hat{A} | \psi_0 \rangle$ and omitted higher order terms
with respect to $\Omega$.

Next, we apply the rotating wave approximation.
Let $\hat{\mathcal{H}}_{0}' = \hat{\mathcal{H}}_0 + \hat{\mathcal{H}}'$ with $\hat{\mathcal{H}}'=\omega \sum_i \hat{a}_i^\dagger \hat{a}_i$, then we have $\hat{\mathcal{H}}_0 = \hat{\mathcal{H}}_\text{s} + \hat{\mathcal{H}}_\text{a}$ with
\begin{equation}
\label{AuxHam}
\hat{\mathcal{H}}_\text{a} = \sum_{ij}[T_\text{a}]_{ij}\hat{a}_{i}^{\dagger}\hat{a}_{j}-(\omega-\omega_\text{a})\sum_{i}\hat{a}_{i}^{\dagger}\hat{a}_{i}.
\end{equation}
Using $[\hat{\mathcal{H}}_0,\hat{\mathcal{H}}']=[\hat{\mathcal{H}}_0 ^{\dagger},\hat{\mathcal{H}}']=0$ and $e^{\ii \hat{\mathcal{H}}' t} \hat{a}_i e^{-\ii \hat{\mathcal{H}}' t} = e^{-\ii \omega t} \hat{a}_i$, we obtain
\begin{equation}
\label{Eq11}
\begin{aligned}
\dot{N}_\text{a}(t) &=
\frac{\Omega^{2}}{4} \int_{0}^{t} dt'
\sum_{ij}
\langle
e^{ \ii \hat{\mathcal{H}}_{0}^\dagger t}
(\hat{c}_{i1}^{\dagger}\hat{a}_{i}-\hat{a}_{i}^{\dagger}\hat{c}_{i1})
e^{-\ii \hat{\mathcal{H}}_{0} (t - t')}
(\hat{c}_{j1}^{\dagger}\hat{a}_{j}+\hat{a}_{j}^{\dagger}\hat{c}_{j1})
e^{-\ii \hat{\mathcal{H}}_{0} t'}
\rangle
\\&  \ \ \ \ \ \  + \frac{\Omega}{2} \int_{0}^{t} dt' \sum_j \langle
e^{ \ii \hat{\mathcal{H}}_{0}^\dagger t} (\hat{\mathcal{H}}_\text{s}^\dagger - \hat{\mathcal{H}}_\text{s}) \hat{N}_\text{a} e^{-\ii \hat{\mathcal{H}}_{0} (t-t')}
(\hat{c}_{j1}^{\dagger}\hat{a}_{j}+\hat{a}_{j}^{\dagger}\hat{c}_{j1})
e^{-\ii \hat{\mathcal{H}}_{0} t'} \rangle
 + \mathrm{H.c.}
\\&= \frac{\Omega^{2}}{4} \int_{0}^{t} dt'  \sum_{ij} \Gamma_{ij} (t,t')
+ \frac{\Omega}{2} \int_{0}^{t} dt' \sum_{j}  \Lambda_{j} (t,t')
 + \mathrm{H.c.},
\end{aligned}
\end{equation}
where we have neglected the fast oscillating terms like $e^{-2\ii \omega t'}$.
Note that here we have already omitted the term which contains only $e^{2\ii \omega t}$; otherwise, it can be omitted in the calculation of $N_\text{a} (t)$.

We first consider the initial state with one particle at momentum $k$, that is,
$\ket{\psi_0} =\ket{\psi_\text{a}^k}= \hat{a}_k^\dagger|0\rangle=\frac{1}{\sqrt{N}} \sum_j e^{\ii kj} \hat{a}_j^\dagger \ket{0}$, and we have
$\hat{\mathcal{H}}_\text{a} \ket{\psi_\text{a}^k} = E_k \ket{\psi_\text{a}^k}$ and
$\hat{\mathcal{H}}_\text{s} \ket{\psi_\text{a}^k}=\hat{\mathcal{H}}_\text{s}^{\dagger} \ket{\psi_\text{a}^k} = 0$. Using
these properties including $[\hat{\mathcal{H}}_\text{s},\hat{\mathcal{H}}_\text{a}]=[\hat{\mathcal{H}}_\text{s}^{\dagger},\hat{\mathcal{H}}_\text{a}] = 0$ and $\hat{\mathcal{H}}_\text{a}^\dagger = \hat{\mathcal{H}}_\text{a}$, we obtain
\begin{equation}
\begin{aligned}
\Gamma_{ij} (t,t') &=   \langle
e^{ \ii \hat{\mathcal{H}}_{0}^\dagger t}
(\hat{c}_{i1}^{\dagger}\hat{a}_{i}-\hat{a}_{i}^{\dagger}\hat{c}_{i1})
e^{-\ii \hat{\mathcal{H}}_{0} (t - t')}
(\hat{c}_{j1}^{\dagger}\hat{a}_{j}+\hat{a}_{j}^{\dagger}\hat{c}_{j1})
e^{-\ii \hat{\mathcal{H}}_{0} t'}
\rangle\\
& =
-\langle
e^{ \ii \hat{\mathcal{H}}_\text{a} t}
\hat{a}_{i}^{\dagger}\hat{c}_{i1}
e^{-\ii \hat{\mathcal{H}}_{0} (t - t')}
\hat{c}_{j1}^{\dagger}\hat{a}_{j}
e^{-\ii \hat{\mathcal{H}}_\text{a} t'}
\rangle \\
& =
-e^{\ii E_k (t-t')}
\langle
\hat{a}_{i}^{\dagger} \hat{c}_{i1} e^{-\ii \hat{\mathcal{H}}_{0} (t-t')} \hat{c}_{j1}^{\dagger} \hat{a}_{j} \rangle \\
&= -e^{\ii E_k (t-t')}
\bra{0}
(\frac{1}{\sqrt{N}} e^{-\ii k i} \hat{c}_{i1}) \cdot
e^{-\ii \hat{\mathcal{H}}_\text{s} (t-t')} \cdot
(\frac{1}{\sqrt{N}} e^{\ii k j} \hat{c}_{j1}^{\dagger})
\ket{0}
\end{aligned}
\end{equation}
and
\begin{equation}
\label{Eq15}
\begin{aligned}
\Lambda_{j} (t,t') &= \langle
e^{ \ii \hat{\mathcal{H}}_{0}^\dagger t} (\hat{\mathcal{H}}_\text{s}^\dagger - \hat{\mathcal{H}}_\text{s}) \hat{N}_\text{a} e^{-\ii \hat{\mathcal{H}}_{0} (t-t')}
(\hat{c}_{j1}^{\dagger}\hat{a}_{j}+\hat{a}_{j}^{\dagger}\hat{c}_{j1})
e^{-\ii \hat{\mathcal{H}}_{0} t'} \rangle
\\&= \langle
e^{ \ii \hat{\mathcal{H}}_\text{a} t} (\hat{\mathcal{H}}_\text{s}^\dagger - \hat{\mathcal{H}}_\text{s}) \hat{N}_\text{a} e^{-\ii \hat{\mathcal{H}}_{0} (t-t')}
(\hat{c}_{j1}^{\dagger}\hat{a}_{j}+\hat{a}_{j}^{\dagger}\hat{c}_{j1})
e^{-\ii \hat{\mathcal{H}}_\text{a} t'} \rangle
\\&= \langle
(\hat{\mathcal{H}}_\text{s}^\dagger - \hat{\mathcal{H}}_\text{s}) e^{ \ii \hat{\mathcal{H}}_\text{a} t}  \hat{N}_\text{a} e^{-\ii \hat{\mathcal{H}}_{0} (t-t')}
(\hat{c}_{j1}^{\dagger}\hat{a}_{j}+\hat{a}_{j}^{\dagger}\hat{c}_{j1})
e^{-\ii \hat{\mathcal{H}}_\text{a} t'} \rangle
\\&= 0.
\end{aligned}
\end{equation}
Since $\ket{\psi_\text{s}^{k\alpha}} = \sum_{j} \frac{1}{\sqrt{N}} e^{\ii k j} \hat{c}_{j\alpha}^{\dagger} \ket{0}$ is a k-space basis vector of the system, we obtain
\begin{equation}
\label{Eq2}
\begin{aligned}
\dot{N}_\text{a}(t)
& = -\frac{\Omega^{2}}{4} \int_{0}^{t}dt' e^{\ii E_k (t-t')}
\bra{0}
(\sum_{i} \frac{1}{\sqrt{N}} e^{-\ii k i} \hat{c}_{i1}) \cdot
e^{-\ii\hat{\mathcal{H}}_\text{s} (t-t')} \cdot
(\sum_{j} \frac{1}{\sqrt{N}} e^{\ii k j} \hat{c}_{j1}^{\dagger})
\ket{0}
+ \mathrm{H.c.} \\
& = -\frac{\Omega^{2}}{4} \int_{0}^{t}dt' e^{\ii E_k (t-t')}  \bra{\psi_\text{s}^{k1}} e^{-\ii \hat{\mathcal{H}}_\text{s} (t-t')} \ket{\psi_\text{s}^{k1}} + \mathrm{H.c.}.
\end{aligned}
\end{equation}
We write $\ket{\psi_\text{s}^{k\alpha}}=\sum_{k^\prime m} \ket{u_\text{R}^{k^\prime m}}\langle u_\text{L}^{k^\prime m} | \psi_\text{s}^{k\alpha} \rangle=\sum_{m} \ket{u_\text{R}^{km}}\langle u_\text{L}^{km} | \psi_\text{s}^{k\alpha} \rangle$
where $\ket{u_\text{R}^{km}}$ is the right eigenstate of the Hamiltonian $\hat{\mathcal{H}}_\text{s}$ satisfying $\hat{\mathcal{H}}_\text{s} \ket{u_\text{R}^{km}} = (\varepsilon_{km} - \ii \gamma_{km}) \ket{u_\text{R}^{km}}$, and $\bra{u_\text{L}^{km}}$ is the corresponding left one.
We further reduces Eq.~(\ref{Eq2}) to
\begin{equation}
\label{Eq1}
\begin{aligned}
\dot{N}_\text{a}(t)
& = -\frac{\Omega^{2}}{4} \sum_m \int_{0}^{t}dt' e^{\ii E_k (t-t')}  \bra{\psi_\text{s}^{k1}} e^{-\ii \hat{\mathcal{H}}_\text{s} (t-t')} \ket{u_\text{R}^{km}}\langle u_\text{L}^{km} | \psi_\text{s}^{k1} \rangle + \mathrm{H.c.} \\
& = -\frac{\Omega^{2}}{4} \sum_m \int_{0}^{t}dt' e^{-\gamma_{km} (t-t')} e^{\ii \Delta_{km} (t-t')} (a_{km}^{(1)} + \ii b_{km}^{(1)}) + \mathrm{H.c.} \\
& = -\frac{\Omega^{2}}{2} \sum_m \int_{0}^{t}dt' e^{-\gamma_{km} (t-t')} \{ a_{km}^{(1)} \cos [\Delta_{km} (t-t')] - b_{km}^{(1)} \sin [\Delta_{km} (t-t')] \} \\
& = -\frac{\Omega^{2}}{2} \sum_m \frac{a_{km}^{(1)}  \gamma_{km} - b_{km}^{(1)} \Delta_{km} - e^{-\gamma_{km} t} [(a_{km}^{(1)} \gamma_{km} - b_{km}^{(1)} \Delta_{km} ) \cos (\Delta_{km} t) - (a_{km}^{(1)} \Delta_{km} + b_{km}^{(1)} \gamma_{km}) \sin (\Delta_{km} t)] }{\Delta_{km}^2 + \gamma_{km}^2}.
\end{aligned}
\end{equation}
Since we only consider the first-order contribution of the time evolution operator $\hat{U}^{I}(t,t_{0})$ in the derivation,
the results agree well with numerical results only in a short time [see blue and black lines in Fig.~\ref{figS1}(a)].
In a longer time scale $t \gg 1/\gamma_{km} $ such that $e^{-\gamma_{km} t} \approx 0$, we neglect the oscillation terms
and approximate $\dot{N}_\text{a}(t)$ by
\begin{equation}
\label{1storder}
\dot{N}_\text{a}(t) \approx -\kappa
= -\frac{\Omega^{2}}{2} \sum_m \frac{a_{km}^{(1)}  \gamma_{km} - b_{km}^{(1)} \Delta_{km} }{\Delta_{km}^2 + \gamma_{km}^2}.
\end{equation}
If an initial state is $\ket{\psi_0} = \sqrt{N_0} \ket{\psi_\text{a}^k}$ with $\langle \psi_0 | \psi_0 \rangle = N_0 \le 1$, then by a similar derivation we would get $\dot{N}_\text{a}(t) = -\kappa N_0$.
In light of this result and the fact that $N_\text{a} (t)$ decreases with time, we conclude that
\begin{equation}
\label{Neq}
\dot{N}_\text{a}(t) = -\kappa N_\text{a}(t).
\end{equation}
Solving Eq.~(\ref{Neq}), we get $N_\text{a}(t) = N_\text{a}(0) e^{-\kappa t}$ when $\gamma_{km} t \gg 1$.

We remark that while we consider the 1D case for simplicity, Eq.~(\ref{thrm1}) can be easily generalized to two or three dimensional case by changing
the momentum $k$ to a vector $\bm k$.
\end{proof}

\bigskip
For a Hermitian system, we have $\gamma_{km} = 0$, $\ket{u_\text{R}^{km}} = \ket{u_\text{L}^{km}}$ and $b_{km} = 0$.
Using the formula $\lim_{K \rightarrow \infty} \sin(Kx)/x = \pi \delta (x)$ reduces \Eq{Eq1} to
\begin{equation}
\label{Eq18}
\dot{N}_\text{a}(t)
= -\frac{\Omega^{2}}{2} \sum_m \frac{a_{km} \sin (\Delta_{km} t) }{\Delta_{km}}
\overset{t \rightarrow \infty}{=\joinrel=}
-\frac{\pi \Omega^{2}}{2} \sum_m a_{km} \delta (\Delta_{km}),
\end{equation}
which is consistent with the result derived in a continuum model \cite{TormaBook}.

We have also numerically confirmed the theorem. Figure~\ref{figS1} illustrates that $\dot{N}_\text{a}(t)$
and $N_\text{a}(t)$ obtained from Eq.~(\ref{thrm1}) agree well with numerical results in a long time.
In a short time, $\dot{N}_\text{a}(t)$ from Eq.~(\ref{thrm1}) deviates from the numerical results which exhibit strong oscillations.
In Fig.~\ref{figS1}(b), we also plot the population of the system which is defined as
$N_\text{s}(t) = \bra{\psi(t)} \sum_{i\alpha} \hat{c}_{i\alpha}^\dagger \hat{c}_{i\alpha} \ket{\psi(t)}$,
indicating that the population is close to zero during the time evolution.

\begin{figure}[t]
\includegraphics[width=4.8in]{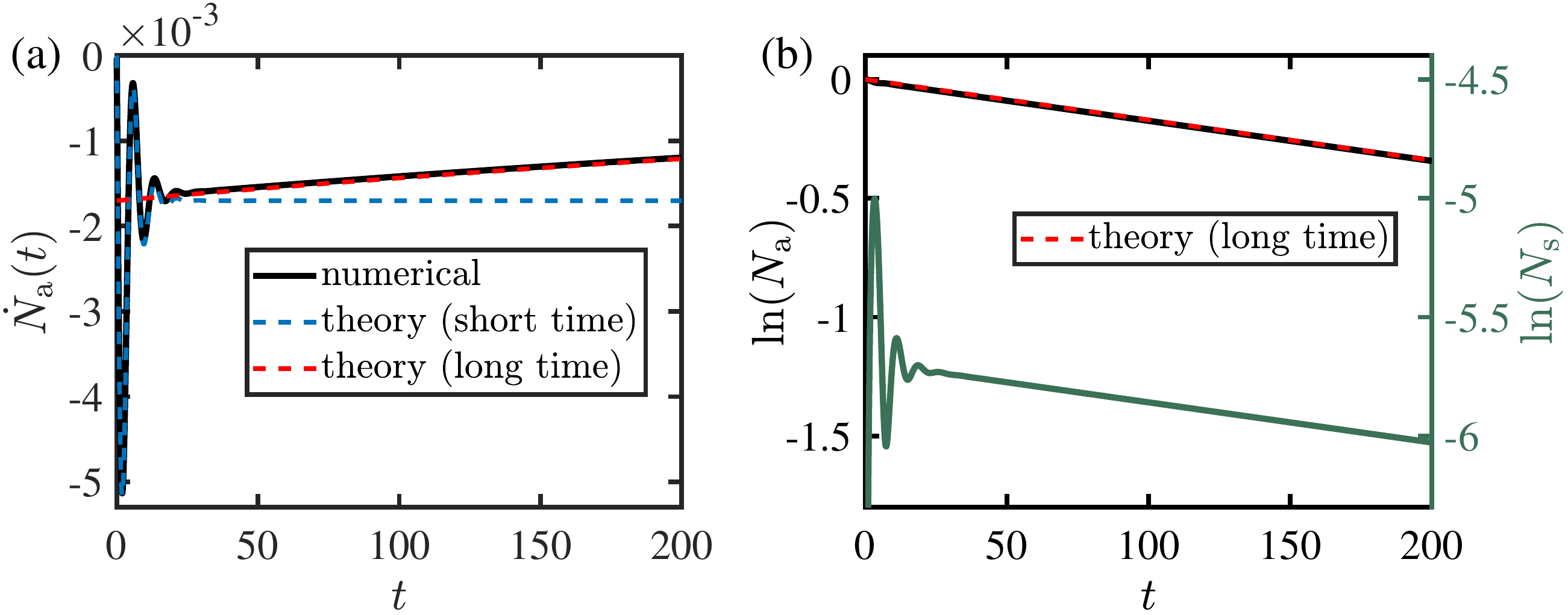}
\caption{
(a) $\dot{N}_\text{a}(t)$ versus $t$ with the black line obtained by numerical calculations,
the dashed red line from Eq.~(\ref{thrm1}) and the dashed blue line from
Eq.~(\ref{Eq1}).
(b) $\ln(N_\text{a})$ and $\ln(N_\text{s})$ as a function of $t$.
The black line is obtained by numerical calculations, and the dashed red line is obtained from Eq.~(\ref{thrm1}).
The results are obtained for the Hatano-Nelson (HN) model with parameters
$J_s=-1$, $g=-0.5$, $\gamma=0.6$, $\delta=-0.2$, $\Omega=0.1$ and $J=0$ for an initial state with $k=8\pi/5$ and $N_{\text{a}}(0)=1$.
}
\label{figS1}
\end{figure}

\section{S-3. Proof of insensitivity of non-Hermitian absorption spectroscopy on boundary conditions in the thermodynamic limit}

In this section, we will prove that the dynamics of $N_\text{a}(t)$ is independent of boundary conditions in the thermodynamic limit.

\subsection{A. Insensitivity of momentum-space propagators on boundary conditions in the thermodynamic limit}

We first prove that the time evolution operators in $k$-space under open boundary conditions
and periodic boundary conditions are the same in the thermodynamic limit.

\begin{lemma}\label{lemma1}
For a 1D translation-invariant system, if the hopping range is finite, then we have
\begin{equation}
D = \bra{k \alpha} e^{-\ii \hat{\mathcal{H}}_\text{o} t} - e^{-\ii \hat{\mathcal{H}}_\text{p} t} \ket{k' \alpha'} \propto \frac{1}{N},
\end{equation}
where $\hat{\mathcal{H}}_\text{o}$ and $\hat{\mathcal{H}}_\text{p}$ denote the Hamiltonian (which can be non-Hermitian) under
OBCs and PBCs, respectively, $\ket{k \alpha}$ denotes the k-space basis vector with momentum $k$ on the $\alpha$th degree of freedom, and $N$ is the number of unit cells.
Since $D$ is proportional to $1/N$, we have $D = 0$ when $N \rightarrow \infty$.
\end{lemma}

\begin{proof}
We first write down the general form of a Hamiltonian under PBCs as
\begin{equation}
\hat{\mathcal{H}}_\text{p} = \sum_{i=1}^{N} \sum_{x=-L}^{L} \sum_{\alpha,\beta=1}^{p} t_{x,\alpha \beta}
\hat{c}_{i+x,\alpha}^\dagger \hat{c}_{i,\beta},
\end{equation}
where $p$ is the number of degrees of freedom within a unit cell and $L$ is the hopping range.
The difference between the Hamiltonians under OBCs and PBCs is given by
\begin{equation}
\hat{B} = \hat{\mathcal{H}}_\text{p} - \hat{\mathcal{H}}_\text{o} =
\sum_{i=1}^{L} \sum_{x=-L}^{-i} \sum_{\alpha,\beta=1}^{p}
t_{x,\alpha \beta} \hat{c}_{i+x+N,\alpha}^\dagger \hat{c}_{i,\beta}
+
\sum_{i=N-L+1}^{N} \sum_{x=N-i+1}^{L} \sum_{\alpha,\beta=1}^{p}
t_{x,\alpha \beta} \hat{c}_{i+x-N,\alpha}^\dagger \hat{c}_{i,\beta}.
\end{equation}
We calculate the matrix element of $\hat{B}$ in $k$-space and find that
\begin{equation}
\bra{k\alpha} \hat{B} \ket{k' \alpha'} =
\sum_{i=1}^{L} \sum_{x=-L}^{-i} \frac{1}{N}
t_{x,\alpha \alpha'} e^{-\ii k (i+x)} e^{ \ii k' i}
+
\sum_{i=N-L+1}^{N} \sum_{x=N-i+1}^{L} \frac{1}{N}
t_{x,\alpha \alpha'} e^{-\ii k (i+x)} e^{ \ii k' i}
= \frac{f_{\alpha \alpha'}(k,k')}{N},
\end{equation}
where $f_{\alpha \alpha'}(k,k')$ is independent of $N$ if the hopping range $L$ is finite.

We now prove that
\begin{equation}
\label{Eq5}
D \equiv \sum_{n=1}^{\infty} \frac{(-\ii t)^n}{n!}
\bra{k \alpha} (\hat{\mathcal{H}}_\text{p} - \hat{B})^n - (\hat{\mathcal{H}}_\text{p})^n \ket{k' \alpha'}
=\sum_{n=1}^{\infty} \frac{(-\ii t)^n}{n!} D_n \propto \frac{1}{N}
\end{equation}
with $D_n = \bra{k \alpha} (\hat{\mathcal{H}}_\text{p} - \hat{B})^n - (\hat{\mathcal{H}}_\text{p})^n \ket{k' \alpha'}$.
In fact, each $D_n$ is proportional to $1/N$.
We will show this for $D_1$ and $D_2$.
For $D_1$, we have
\begin{equation}
D_1 = - \bra{k \alpha} \hat{B} \ket{k' \alpha'} = -\frac{f_{\alpha \alpha'} (k,k')}{N} \propto \frac{1}{N}.
\end{equation}
$D_2$ can be separated into three parts:
\begin{equation}
D_2 = \bra{k \alpha} -\hat{B} \hat{\mathcal{H}}_\text{p} - \hat{\mathcal{H}}_\text{p} \hat{B} + \hat{B}^2 \ket{k' \alpha'}.
\end{equation}
For the first part,
\begin{equation}
\begin{aligned}
\bra{k \alpha} \hat{B} \hat{\mathcal{H}}_\text{p} \ket{k' \alpha'}
&= \sum_{k_1 \alpha_1 q \beta} \bra{k \alpha} \hat{B} \ket{k_1 \alpha_1} \bra{k_1 \alpha_1} \hat{\mathcal{H}}_\text{p} \ket{u_\text{R}^{q\beta}} \langle u_\text{L}^{q\beta} | k' \alpha' \rangle \\
&= \sum_{\alpha_1 \beta} \frac{1}{N} f_{\alpha \alpha_1} (k,k') E_{\beta}(k') \langle k' \alpha_1 | u_\text{R}^{k' \beta} \rangle \langle u_\text{L}^{k' \beta} | k' \alpha' \rangle \\
&= \sum_{\alpha_1 \beta} \frac{1}{N} f_{\alpha \alpha_1} (k,k') E_{\beta}(k') g_{\alpha_1 \alpha' \beta} (k')
\propto \frac{1}{N},
\end{aligned}
\end{equation}
where $\ket{u_\text{R}^{k \alpha}}$ is the right eigenvector of $\hat{\mathcal{H}}_\text{p}$ with eigenenergy $E_{\alpha} (k)$, $g_{\alpha \alpha' \beta} (k) = \langle k \alpha | u_\text{R}^{k \beta} \rangle \langle u_\text{L}^{k \beta} | k \alpha' \rangle$. In the derivation, we have also used
the fact that $\sum_{k\alpha} \ket{u_\text{R}^{k\alpha}} \bra{u_\text{L}^{k\alpha}} = \mathbf{1}$.
Similarly, one can derive that $\bra{k \alpha} \hat{\mathcal{H}}_\text{p} \hat{B}  \ket{k' \alpha'} \propto 1/N$.
For the third part, we obtain that
\begin{equation}
\begin{aligned}
\bra{k\alpha} \hat{B}^2 \ket{k' \alpha'} &= \sum_{q \beta} \bra{k\alpha} \hat{B} \ket{q \beta} \bra{q \beta} \hat{B} \ket{k' \alpha'} \\
&= \sum_{q \beta} \frac{1}{N^2} f_{\alpha \beta} (k,q) f_{\beta \alpha'} (q,k') \\
&\simeq \frac{1}{N} \int_{0}^{2\pi} dq \, \sum_{\beta} \frac{1}{2\pi} f_{\alpha \beta} (k,q) f_{\beta \alpha'} (q,k') \propto \frac{1}{N},
\end{aligned}
\end{equation}
where in the final step we have substituted $\sum_k$ with $\frac{N}{2\pi} \int_{0}^{2\pi} dk$.
Thus, we have $D_2 \propto 1/N$.
Other terms are also proportional to $1/N$. For example, in $D_3$,
\begin{equation}
\begin{aligned}
\bra{k\alpha} \hat{B} \hat{\mathcal{H}}_\text{p} \hat{B} \ket{k' \alpha'}
&= \sum_{k_1 \alpha_1  q\beta k_2 \alpha_2} \bra{k \alpha} \hat{B} \ket{k_1 \alpha_1} \bra{k_1 \alpha_1} \hat{\mathcal{H}}_\text{p} \ket{u_\text{R}^{q\beta}}  \langle u_\text{L}^{q\beta} | k_2 \alpha_2 \rangle \bra{k_2 \alpha_2} \hat{B} \ket{k' \alpha'} \\
&= \sum_{q \alpha_1 \alpha_2 \beta} \frac{1}{N^2} f_{\alpha \alpha_1} (k,q) \, E_{\beta} (q) \, g_{\alpha_1 \alpha_2 \beta} (q) \,  f_{\alpha_2 \alpha'} (q,k') \\
&\simeq \frac{1}{N} \int_{0}^{2\pi} dq \, \sum_{\alpha_1 \alpha_2 \beta} \frac{1}{2\pi} f_{\alpha \alpha_1} (k,q) \, E_{\beta} (q) \, g_{\alpha_1 \alpha_2 \beta} (q) \, f_{\alpha_2 \alpha'} (q,k') \propto \frac{1}{N}.
\end{aligned}
\end{equation}
So we have $D \propto 1/N$ because each term in $D$ is proportional to $1/N$.

We conclude that in the thermodynamic limit ($N \rightarrow \infty$), $D$ approaches zero and the time evolution operators in k-space are the same for the Hamiltonians under OBCs and PBCs.
We remark that for a $d$-dimensional system, one can easily find that $D\propto 1/N^d$.
\end{proof}

\bigskip

Our results may be beyond the intuition that since the Hatano-Nelson model exhibits real energy spectra for open boundaries [the eigenenergy can be approximated by 
$\varepsilon_k^\text{o} \simeq 2\, \text{sgn}(J_s) \sqrt{J_s^2-g^2} \cos k$ (for $|J_s|>|g|$) obtained through similar transformations]
if we set $\gamma=0$,
the dynamics would be bounded in sharp contrast to the case for periodic boundaries. 
In fact, even though the eigenenergies are purely real, the dynamics of $D_{\text{o}}=\bra{k}e^{-i\hat{\mathcal{H}}_{\text{o}} t}\ket{k}$ may still lead to gain or loss. 

To illustrate the fact, 
we first consider the state $\ket{\phi(t)} = e^{-\ii \hat{\mathcal{H}}_\text{o}^{\text{HN}} t} \ket{k}$ and see whether its norm $n_\phi (t) = \langle \phi (t) | \phi(t) \rangle$ remains at one under the evolution of the Hatano-Nelson model $\hat{\mathcal{H}}_\text{o}^{\text{HN}}$ under OBCs [see the Hamiltonian (3) in the main text]. 
Since the energies of $\hat{\mathcal{H}}_\text{o}^{\text{HN}}$ are real, one may think that this norm should remain unchanged as the Hermitian case, and thus $D_\text{o}$ is bounded. 
However, this is not true.
For simplicity, we calculate the time derivative of $n_\phi (t)$ at $t=0$ which is $n_{\phi}' (0) = \ii \bra{k} [\hat{\mathcal{H}}_\text{o}^{\text{HN}}]^\dagger - \hat{\mathcal{H}}_\text{o}^{\text{HN}} \ket{k} = 4 g \frac{N-1}{N} \sin k$.
In the Hermitian case ($g=0$), $[\hat{\mathcal{H}}_\text{o}^{\text{HN}}]^\dagger = \hat{\mathcal{H}}_\text{o}^{\text{HN}}$ so that $n_{\phi}' (0) = 0$.
But in this non-Hermitian case $g \neq 0$, when $k\neq 0$ or $\pi$, $n_\phi'(0) \neq 0$, indicating that the state $\phi (t)$ may grow (when $g\sin k>0$) 
or decay (when $g\sin k<0$) under the evolution of $\hat{\mathcal{H}}_\text{o}^{\text{HN}}$. 
It implies that $D_\text{o}$ may also grow or decay over time and thus may be unbounded, in  contrast to Hermitian cases where $n_\phi (t) = 1$ and $D_\text{o}$ is bounded by $|D_\text{o}| \leq 1$.
This example suggests that the reality of the energy spectrum cannot guarantee the conservation of the norm of a state under time evolution in the non-Hermitian case.

\begin{figure*}[t]
\includegraphics[width=2.5in]{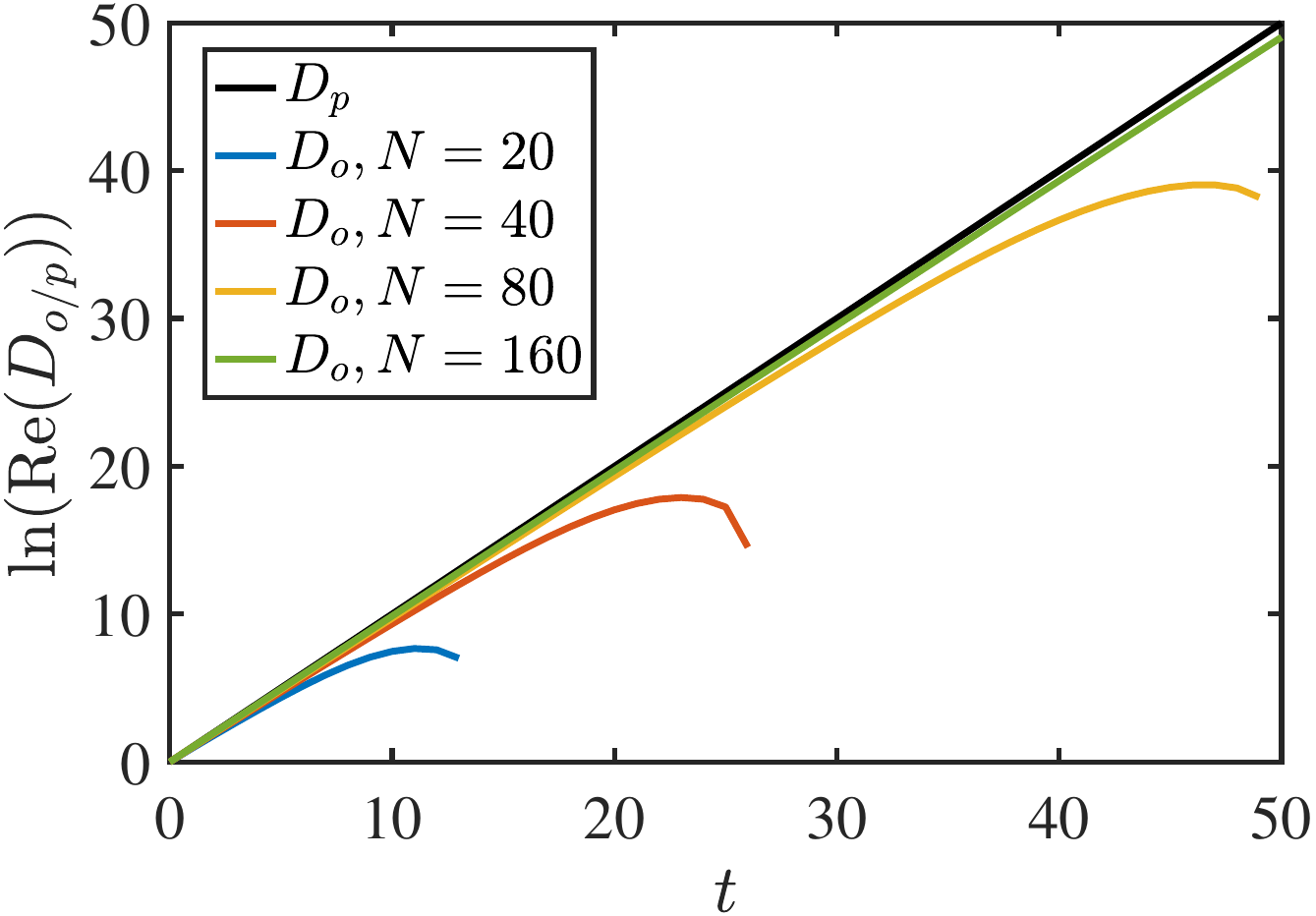}
\caption{
Comparison between $D_\text{o}$ and $D_\text{p}$ at various system sizes $N$ for the Hatano-Nelson model. Here $J_s=-1$, $g=-0.5$, $\gamma=0$ and $k=3\pi/2$.
Note that $D_\text{p} = \bra{k} e^{-\ii \hat{\mathcal{H}}_\text{p} t} \ket{k} = e^{-\ii \lambda_k^\text{p} t} = e^{2gt\sin k} e^{-\ii 2J_s t \cos k} = e^{t}$ for the parameters here, indicating that $D_\text{p}$ grows exponentially.
}
\label{figS2}
\end{figure*}

We further show the numerical results of $D_\text{o}$ and $D_\text{p}$ for the Hatano-Nelson model in Fig.~\ref{figS2}. We see that
$D_\text{o}$ also grows exponentially with time, and as $N$ increases, $D_\text{o}$ approaches $D_\text{p}$, suggesting that $D_\text{o} = D_\text{p}$ when $N$ goes to infinity, which is consistent with Lemma~\ref{lemma1}.

\subsection{B. Insensitivity of $N_\text{a}(t)$ on boundary conditions in the thermodynamic limit}

We are now in a position to prove that the dynamics of $N_\text{a}(t)$ is independent of boundary conditions in the thermodynamic limit.
\begin{theorem}\label{Them3}
For a translation-invariant system, Theorem~\ref{Them1} holds in the thermodynamic limit even if
$\hat{\mathcal{H}}_\text{s}$ and $\hat{\mathcal{H}}_\text{a}$ are under OBCs, provided that
the hopping range is finite in $\hat{\mathcal{H}}_\text{s}$ and $\hat{\mathcal{H}}_\text{a}$.
\end{theorem}

\begin{proof}
We first show that in Eq.~(\ref{Eq11}),
\begin{equation}
\label{Eq29}
\begin{aligned}
\sum_{ij} \Gamma_{ij} (t,t') &=
\sum_{ij}  \bra{\psi_\text{a}^k}
e^{ \ii \hat{\mathcal{H}}_{0}^\dagger t}
(\hat{c}_{i1}^{\dagger}\hat{a}_{i}-\hat{a}_{i}^{\dagger}\hat{c}_{i1})
e^{-\ii \hat{\mathcal{H}}_{0} (t - t')}
(\hat{c}_{j1}^{\dagger}\hat{a}_{j}+\hat{a}_{j}^{\dagger}\hat{c}_{j1})
e^{-\ii \hat{\mathcal{H}}_{0} t'}
\ket{\psi_\text{a}^k} \\
& = - \sum_{ij} \bra{\psi_\text{a}^k}
e^{ \ii \hat{\mathcal{H}}_\text{a} t}
\hat{a}_{i}^{\dagger}\hat{c}_{i1}
e^{-\ii \hat{\mathcal{H}}_{0} (t - t')}
\hat{c}_{j1}^{\dagger}\hat{a}_{j}
e^{-\ii \hat{\mathcal{H}}_\text{a} t'}
\ket{\psi_\text{a}^k} \\
& = - \sum_{ij q_1 k_1 \alpha k_2 \beta q_2} \bra{\psi_\text{a}^k}
e^{ \ii \hat{\mathcal{H}}_\text{a} t}
\ket{\psi_\text{a}^{q_1}} \bra{\psi_\text{a}^{q_1}}
\hat{a}_{i}^{\dagger}\hat{c}_{i1}
\ket{\psi_\text{s}^{k_1 \alpha}} \bra{\psi_\text{s}^{k_1 \alpha}}
e^{-\ii \hat{\mathcal{H}}_\text{s} (t - t')}
\ket{\psi_\text{s}^{k_2 \beta}} \bra{\psi_\text{s}^{k_2 \beta}}
\hat{c}_{j1}^{\dagger}\hat{a}_{j}
\ket{\psi_\text{a}^{q_2}} \bra{\psi_\text{a}^{q_2}}
e^{-\ii \hat{\mathcal{H}}_\text{a} t'}
\ket{\psi_\text{a}^k} \\
& = - \sum_{ij q_1 k_1 k_2 q_2}
\frac{1}{N^2} e^{\ii [(k_1 - q_1) i - (k_2 - q_2)j]}
\bra{\psi_\text{a}^k}
e^{ \ii \hat{\mathcal{H}}_\text{a} t}
\ket{\psi_\text{a}^{q_1}}
\bra{\psi_\text{s}^{k_1 1}}
e^{-\ii \hat{\mathcal{H}}_\text{s} (t - t')}
\ket{\psi_\text{s}^{k_2 1}} \bra{\psi_\text{a}^{q_2}}
e^{-\ii \hat{\mathcal{H}}_\text{a} t'}
\ket{\psi_\text{a}^k} \\
& = - \sum_{q_1 q_2}
\bra{\psi_\text{a}^k}
e^{ \ii \hat{\mathcal{H}}_\text{a} t}
\ket{\psi_\text{a}^{q_1}}
\bra{\psi_\text{s}^{q_1 1}}
e^{-\ii \hat{\mathcal{H}}_\text{s} (t - t')}
\ket{\psi_\text{s}^{q_2 1}} \bra{\psi_\text{a}^{q_2}}
e^{-\ii \hat{\mathcal{H}}_\text{a} t'}
\ket{\psi_\text{a}^k}
\end{aligned}
\end{equation}
with $\ket{\psi_\text{a}^k}$ being the initial state.

Since $\Lambda_j (t,t')=0$ [\Eq{Eq15}], based on Eq.~(\ref{Eq11}), we only need to prove that
$\sum_{ij} \Gamma_{ij}^\text{o} (t,t') = \sum_{ij} \Gamma_{ij}^\text{p} (t,t')$
in the thermodynamic limit.
Here we use a superscript `o' (`p') to denote quantities calculated under OBCs (PBCs).
Let $F_{k k'}^{\text{o}(\text{p})} = \bra{\psi_\text{a}^k} e^{ \ii \hat{\mathcal{H}}_\text{a}^{\text{o}(\text{p})} t} \ket{\psi_\text{a}^{k'}}$, $S_{k k'}^{\text{o}(\text{p})} = \bra{\psi_\text{s}^{k 1}} e^{-\ii \hat{\mathcal{H}}_\text{s}^{\text{o}(\text{p})} (t - t')} \ket{\psi_\text{s}^{k' 1}}$ and $G_{k k'}^{\text{o}(\text{p})} = \bra{\psi_\text{a}^{k}} e^{-\ii \hat{\mathcal{H}}_\text{a}^{\text{o}(\text{p})} t'} \ket{\psi_\text{a}^{k'}}$.
Using Lemma~\ref{lemma1}, we have
$F_{k k'}^{\text{o}} - F_{k k'}^{\text{p}} = \frac{1}{N} A_1 (k,k')$,
$S_{k k'}^{\text{o}} - S_{k k'}^{\text{p}} = \frac{1}{N} A_2 (k,k')$ and
$G_{k k'}^{\text{o}} - G_{k k'}^{\text{p}} = \frac{1}{N} A_3 (k,k')$,
where $A_1 (k,k')$, $A_2 (k,k')$ and $A_3 (k,k')$ are independent of $N$.
Then
\begin{equation}
\begin{aligned}
\sum_{ij} \Gamma_{ij}^\text{o} (t,t') - \Gamma_{ij}^\text{p} (t,t')
&= - \Big( \sum_{q_1 q_2} F_{k q_1}^\text{o} S_{q_1 q_2}^\text{o} G_{q_2 k}^\text{o} - F_{k q_1}^\text{p} S_{q_1 q_2}^\text{p} G_{q_2 k}^\text{p} \Big) \\
&= - \Big[ \sum_{q_1 q_2} (F_{k q_1}^\text{o} - F_{k q_1}^\text{p})
S_{q_1 q_2}^\text{o} G_{q_2 k}^\text{o} +
F_{k q_1}^\text{p} (S_{q_1 q_2}^\text{o} G_{q_2 k}^\text{o} - S_{q_1 q_2}^\text{p} G_{q_2 k}^\text{p}) \Big].
\end{aligned}
\end{equation}

Since
\begin{equation}
\begin{aligned}
& \ \ \ \, \sum_{q_1 q_2} (F_{k q_1}^\text{o} - F_{k q_1}^\text{p})
S_{q_1 q_2}^\text{o} G_{q_2 k}^\text{o} \\
& = \sum_{q_1 q_2} (F_{k q_1}^\text{o} - F_{k q_1}^\text{p})
(S_{q_1 q_2}^\text{o} - S_{q_1 q_2}^\text{p}) G_{q_2 k}^\text{o}
+ (F_{k q_1}^\text{o} - F_{k q_1}^\text{p}) S_{q_1 q_2}^\text{p}
G_{q_2 k}^\text{o} \\
& \
\begin{aligned}
= \sum_{q_1 q_2} &(F_{k q_1}^\text{o} - F_{k q_1}^\text{p})
(S_{q_1 q_2}^\text{o} - S_{q_1 q_2}^\text{p})
(G_{q_2 k}^\text{o} - G_{q_2 k}^\text{p})
+ (F_{k q_1}^\text{o} - F_{k q_1}^\text{p})
(S_{q_1 q_2}^\text{o} - S_{q_1 q_2}^\text{p})
G_{q_2 k}^\text{p} \\
&+ (F_{k q_1}^\text{o} - F_{k q_1}^\text{p}) S_{q_1 q_2}^\text{p}
(G_{q_2 k}^\text{o} - G_{q_2 k}^\text{p})
+ (F_{k q_1}^\text{o} - F_{k q_1}^\text{p}) S_{q_1 q_2}^\text{p}
G_{q_2 k}^\text{p}
\end{aligned}\\
& \
\begin{aligned}
= \Big[ \sum_{q_1 q_2} & (F_{k q_1}^\text{o} - F_{k q_1}^\text{p})
(S_{q_1 q_2}^\text{o} - S_{q_1 q_2}^\text{p})
(G_{q_2 k}^\text{o} - G_{q_2 k}^\text{p}) \Big]
+ \Big[ \sum_{q_1} (F_{k q_1}^\text{o} - F_{k q_1}^\text{p})
(S_{q_1 k}^\text{o} - S_{q_1 k}^\text{p})
G_{k k}^\text{p} \Big] \\
&+ \Big[ \sum_{q} (F_{k q}^\text{o} - F_{k q}^\text{p})  S_{q q}^\text{p}
(G_{q k}^\text{o} - G_{q k}^\text{p}) \Big]
+ (F_{k k}^\text{o} - F_{k k}^\text{p}) S_{k k}^\text{p} G_{k k}^\text{p}
\end{aligned}\\
& \
\begin{aligned}
\simeq \Big[ \frac{1}{N} \int &dq_1  dq_2 \, \frac{1}{4\pi^2} A_1 (k,q_1) A_2 (q_1,q_2) A_3 (q_2,k) \Big]
+
\frac{1}{N} G_{k k}^\text{p} \Big[ \int dq_1 \, \frac{1}{2\pi} A_1 (k,q_1) A_2(q_1,k) \Big]
\\
&+ \frac{1}{N} \Big[ \int dq \, \frac{1}{2\pi} A_1 (k,q) S^\text{p} (q,q) A_3 (q,k) \Big]
+ \frac{1}{N} A_1 (k,k) S_{k k}^\text{p} G_{k k}^\text{p}
\end{aligned}\\
& \propto \frac{1}{N},
\end{aligned}
\end{equation}
where $S^\text{p} (q,q) = S_{qq}^\text{p}$, and
\begin{equation}
\begin{aligned}
& \ \ \ \, \sum_{q_1 q_2} F_{k q_1}^\text{p} (S_{q_1 q_2}^\text{o} G_{q_2 k}^\text{o} - S_{q_1 q_2}^\text{p} G_{q_2 k}^\text{p})
= \sum_{q_2}
F_{k k}^\text{p} (S_{k q_2}^\text{o} G_{q_2 k}^\text{o} - S_{k q_2}^\text{p} G_{q_2 k}^\text{p})
\\
& = \sum_{q_2}
F_{k k}^\text{p} \big[ (S_{k q_2}^\text{o} - S_{k q_2}^\text{p}) G_{q_2 k}^\text{o} + S_{k q_2}^\text{p} (G_{q_2 k}^\text{o} - G_{q_2 k}^\text{p}) \big] \\
& = \sum_{q_2}
F_{k k}^\text{p} \big[ (S_{k q_2}^\text{o} - S_{k q_2}^\text{p})
(G_{q_2 k}^\text{o} - G_{q_2 k}^\text{p}) +
(S_{k q_2}^\text{o} - S_{k q_2}^\text{p}) G_{q_2 k}^\text{p}
+ S_{k q_2}^\text{p} (G_{q_2 k}^\text{o} - G_{q_2 k}^\text{p}) \big] \\
& = \Big[ \sum_{q_2}
F_{k k}^\text{p} (S_{k q_2}^\text{o} - S_{k q_2}^\text{p})
(G_{q_2 k}^\text{o} - G_{q_2 k}^\text{p}) \Big]
+ F_{k k}^\text{p} \big[
(S_{k k}^\text{o} - S_{k k}^\text{p}) G_{k k}^\text{p}
+ S_{k k}^\text{p} (G_{k k}^\text{o} - G_{k k}^\text{p}) \big] \\
&\simeq \frac{1}{N} F_{k k}^\text{p} \int dq_2 \, \frac{1}{2\pi} A_2(k,q_2) A_3(q_2,k) +
\frac{1}{N} F_{k k}^\text{p} \big[A_2 (k,k) G_{k k}^\text{p} + S_{k k}^\text{p} A_3 (k,k) \big] \\
& \propto \frac{1}{N},
\end{aligned}
\end{equation}
we conclude that
\begin{equation}
\sum_{ij} \Gamma_{ij}^\text{o} (t,t') - \Gamma_{ij}^\text{p} (t,t')
\propto \frac{1}{N}
\end{equation}
and thus $\dot{N}_\text{a}^\text{o} (t) - \dot{N}_\text{a}^\text{p} (t) \propto 1/N$.
Therefore, Theorem~\ref{Them1} holds in the thermodynamic limit even if
$\hat{\mathcal{H}}_\text{s}$ and $\hat{\mathcal{H}}_\text{a}$ are under OBCs.
\end{proof}

\bigskip

\begin{figure*}[t]
\includegraphics[width=4.4in]{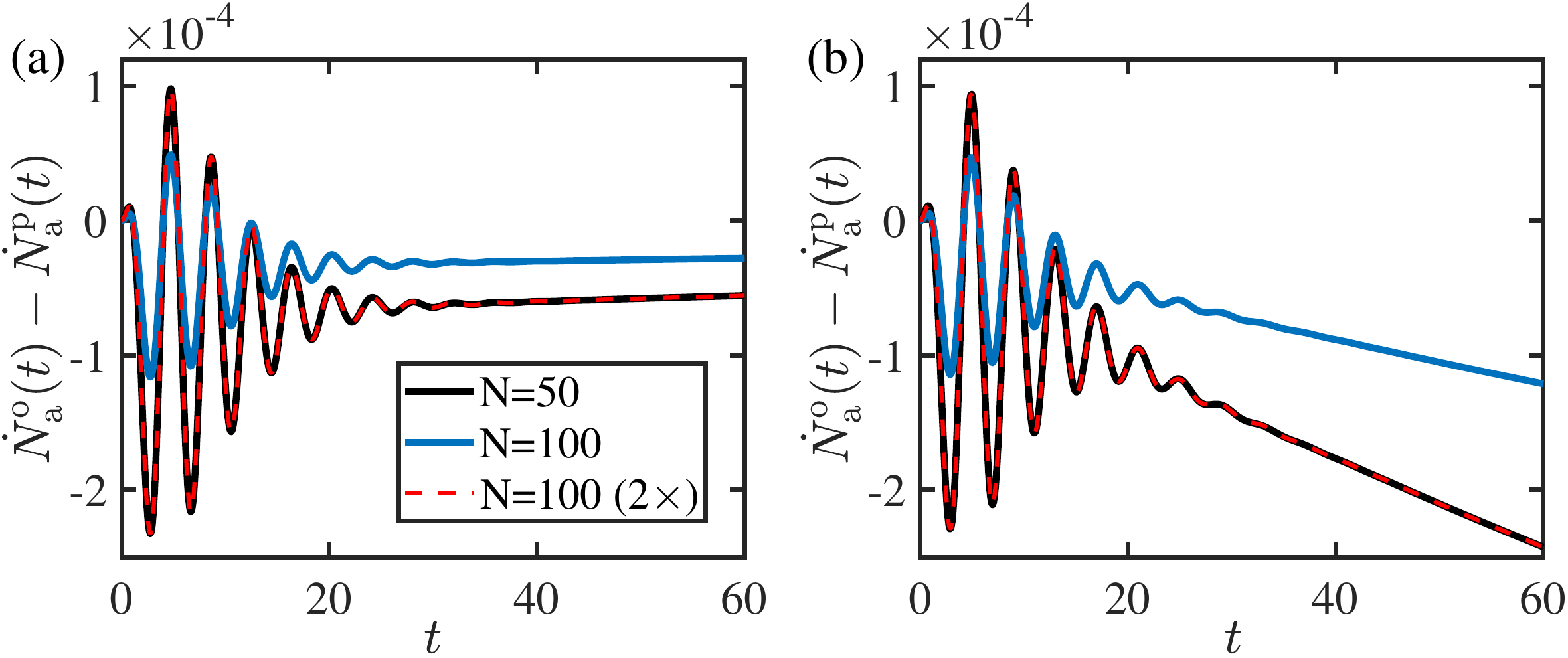}
\caption{
$\dot{N}_\text{a}^\text{o} (t) - \dot{N}_\text{a}^\text{p} (t)$ as time evolves calculated under $N=50$ and $100$ for the HN model with system parameters $J_s=-1$, $g=-0.5$, $\gamma=0.6$, $\Omega=0.1$ and $\delta=-1$ when $k=8\pi/5$.
In (a), $J=0$ and in (b), $J=-0.1$.
For comparison, we also show the doubled value for $N=100$ as a red line.
}
\label{figS3}
\end{figure*}

The theorem is also numerically confirmed.
In Fig.~\ref{figS3}, we show $\dot{N}_\text{a}^\text{o} (t) - \dot{N}_\text{a}^\text{p} (t)$ calculated under $N=50$ and $100$ for the
HN model.
The coincidence of the black line ($N=50$) and the red line (doubled value for $N=100$) agrees with the fact that $\dot{N}_\text{a}^\text{o} (t) - \dot{N}_\text{a}^\text{p} (t)$ is proportional to $1/N$.



\section{S-4. Non-Hermitian Rice-Mele model}

\begin{figure*}[t]
\includegraphics[width=7in]{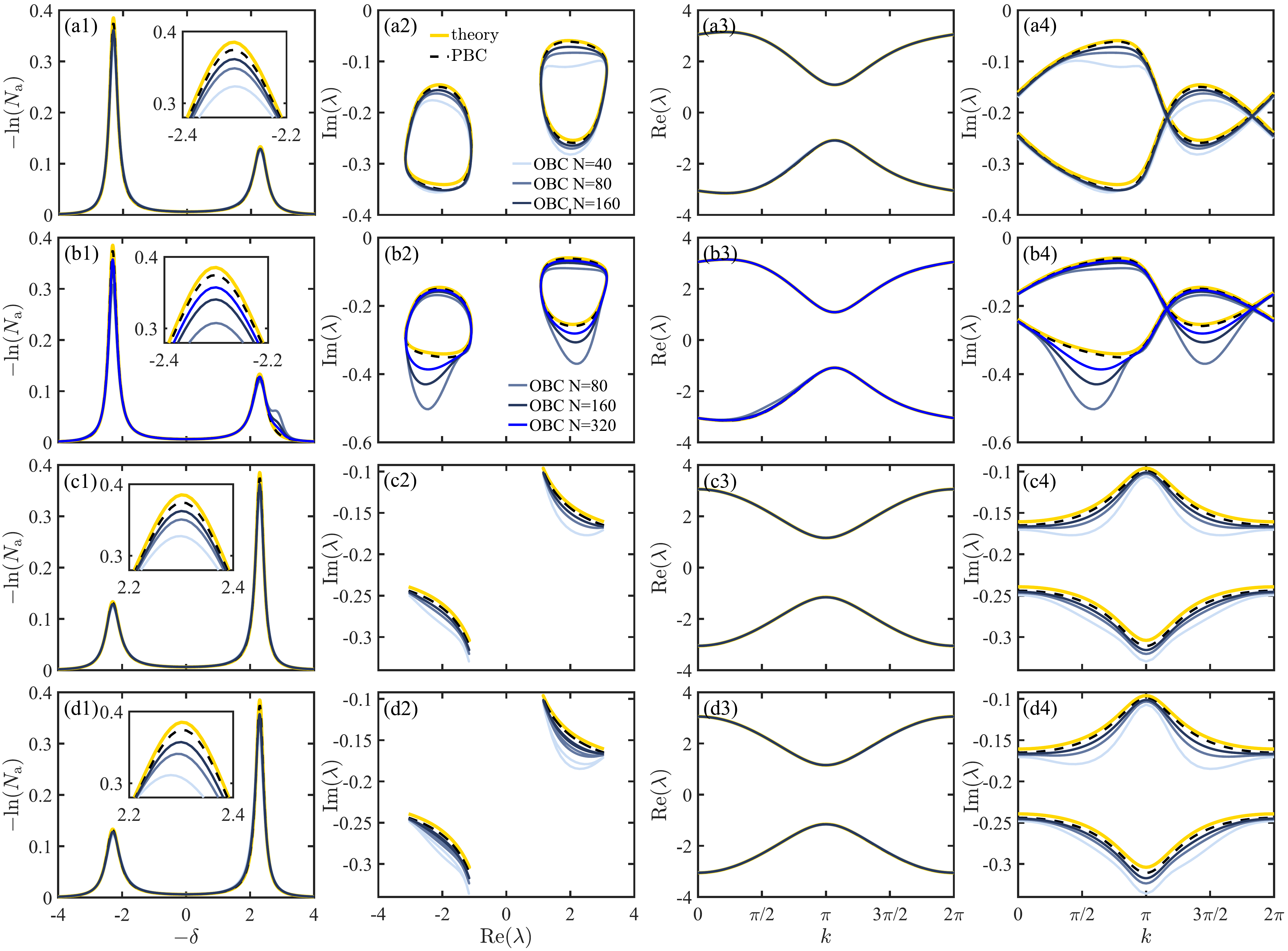}
\caption{
First column: The spectral lines of the non-Hermitian Rice-Mele model when $k=3\pi/2$.
The yellow line is obtained from Eq.~(\ref{RM_theo}),
and the dashed black lines and a set of blue lines are calculated for a system
under PBCs and OBCs, respectively.
The inset shows the zoomed-in view of the spectral line near its highest absorption peak.
The energy spectrum in the complex energy plane (second column), real (third column) and imaginary parts (fourth column)
of energies with respect to $k$.
The energies are extracted by fitting the numerically calculated spectral lines
under PBCs (dashed black lines) and OBCs (solid blue lines) based on Eq.~(\ref{RM_theo}) in comparison with
the theoretical result $\lambda_{km}$ (the yellow lines).
In the first and second rows, $J_1 = 2$, $J_2 = 1$, $m_z = 0.6$, $J_3=1.2$, $\gamma=0.2$, $\Omega=0.03$, $t=200$ and $J=-0.05$
corresponding to a system with NHSEs.
In the third and fourth rows, $J_3=0$ and the values of other parameters are the same as those in the first and second rows
corresponding to a system without NHSEs.
In the first and third rows, the auxiliary levels are under PBCs when the system is under OBCs.
In the second and fourth rows, the auxiliary levels are under OBCs when the system is under OBCs.
}
\label{figS4}
\end{figure*}

In this section, we will
discuss the influence of NHSEs on non-Hermitian absorption spectroscopy
based on the non-Hermitian Rice-Mele model~\cite{NHRM2020PRL}, which is more experimentally relevant.
The model is described by
\begin{equation}
\label{RMHam}
\hat{\mathcal{H}}_\text{s}^\text{NHRM} = \sum_j \big[ J_1 \hat{c}_{j1}^\dagger \hat{c}_{j2} + J_2 \hat{c}_{j2}^\dagger c_{j+1, 1} - \ii \frac{J_3}{2} (\hat{c}_{j1}^\dagger \hat{c}_{j+1, 1} - \hat{c}_{j2}^\dagger \hat{c}_{j+1, 2}) + \text{H.c.} \big] + m_z \hat{c}_{j1}^\dagger \hat{c}_{j1} - (m_z + 2\ii \gamma) \hat{c}_{j2}^\dagger \hat{c}_{j2},
\end{equation}
where $J_1$, $J_2$ and $J_3$ represent the hopping strength,
$m_z$ denotes the mass term, and $\gamma$ describes the decay strength of particles on the second degree of freedom in a unit cell,
which can induce the NHSEs.
The Hamiltonian in momentum space reads
\begin{equation}
h_k^{\text{NHRM}} = (J_1 + J_2 \cos k) \sigma_x + J_2 \sin k \, \sigma_y + (J_3 \sin k  + m_z) \sigma_z + \ii \gamma (\sigma_z - \sigma_0).
\end{equation}
Due to the breaking of both the time-reversal and inversion symmetry, the NHSE occurs when $J_1 J_2 J_3 \neq 0$~\cite{NHRM2020PRL}.

In the presence of auxiliary levels, the full Hamiltonian becomes
\begin{equation}
\label{full_Ham_NHRM}
\hat{\mathcal{H}}^{\text{NHRM}} = \hat{\mathcal{H}}_\text{s}^\text{NHRM} 
+ \sum_j [J(\hat{a}_{j}^{\dagger}\hat{a}_{j+1}+\hat{a}_{j+1}^{\dagger}\hat{a}_{j})-\delta \hat{a}_{j}^{\dagger}\hat{a}_{j} + \frac{\Omega}{2} (\hat{c}_{j1}^{\dagger}\hat{a}_{j}+\hat{a}_{j}^{\dagger}\hat{c}_{j1})] .
\end{equation}

\subsection{A. Case 1}
In this subsection, we will consider the case with the many particle initial state $|\psi_0^{(M)}\rangle=\prod_k\hat{a}_k^\dagger|0\rangle$ (see 
the next subsection for the finite temperature ensemble case for bosons).
Based on Eq.~(\ref{thrm1}), the population of auxiliary levels is given by
\begin{equation}
\label{RM_theo}
-\ln [N_\text{a}(t)]
= \frac{\Omega^{2}t}{2} (\frac{a_{k1}^{(1)}  \gamma_{k1} - b_{k1}^{(1)} \Delta_{k1} }{\Delta_{k1}^2 + \gamma_{k1}^2} +
\frac{a_{k2}^{(1)}  \gamma_{k2} - b_{k2}^{(1)} \Delta_{k2} }{\Delta_{k2}^2 + \gamma_{k2}^2}),
\end{equation}
where $\Delta_{km}=E_k-\varepsilon_{km}$ ($m=1,2$) with $E_k=-\delta+2J\cos k$,
$\varepsilon_{km}=\text{Re}(\lambda_{km})$,
$\gamma_{km}=-\text{Im}(\lambda_{km})$, and
$\lambda_{km}=(-1)^{m}\sqrt{ (J_1 + J_2 \cos k)^2+J_2^2 \sin^2 k+(J_3 \sin k  + m_z + \ii \gamma)^2}-\ii \gamma$.

In Fig.~\ref{figS4}, we show the boundary effects of $\hat{\mathcal{H}}_\text{s}$ and $\hat{\mathcal{H}}_\text{a}$
on non-Hermitian absorption spectroscopy in the cases with or without NHSEs.
The upper two panels refer to the results for a system with NHSEs while
the lower two ones refer to those for a system without NHSEs.
We see clearly that when the system size increases, the difference between the theoretical results and
the open boundary results declines significantly no matter whether a system has NHSEs, implying that
NHSEs cannot prevent us from measuring the complex energy spectrum in momentum space in the thermodynamic limit.
Yet, NHSEs indeed have large effects on the spectroscopy, especially
on measurements of imaginary parts of the energy spectra, when a system size is not sufficiently large
[see the second (a system with NHSEs) and forth panels (a system without NHSEs) in Fig.~\ref{figS4}].
Meanwhile, we find that when the auxiliary levels are under PBCs and the system is under OBCs,
the spectroscopy works much better than the case where the auxiliary levels are also under OBCs even though
a system has NHSEs (see the first panel in Fig.~\ref{figS4}).
The result suggests that i) the NHSE itself has minor effects on the fitted energy spectrum and ii) with the NHSE, the inaccuracy of the fitted energy spectrum under OBCs is mainly caused by the boundary effects of $\hat{\mathcal{H}}_\text{a}$, which will scatter the initial state $\ket{\psi_\text{a}^k}$ to other momentum states. Such effects may be suppressed in a realistic case where all momentum states are occupied in an initial state.
Overall, one can reduce the boundary effects of $\hat{\mathcal{H}}_\text{a}$ by increasing the system size
or decreasing $J$.

\begin{figure*}[t]
\includegraphics[width=4.5in]{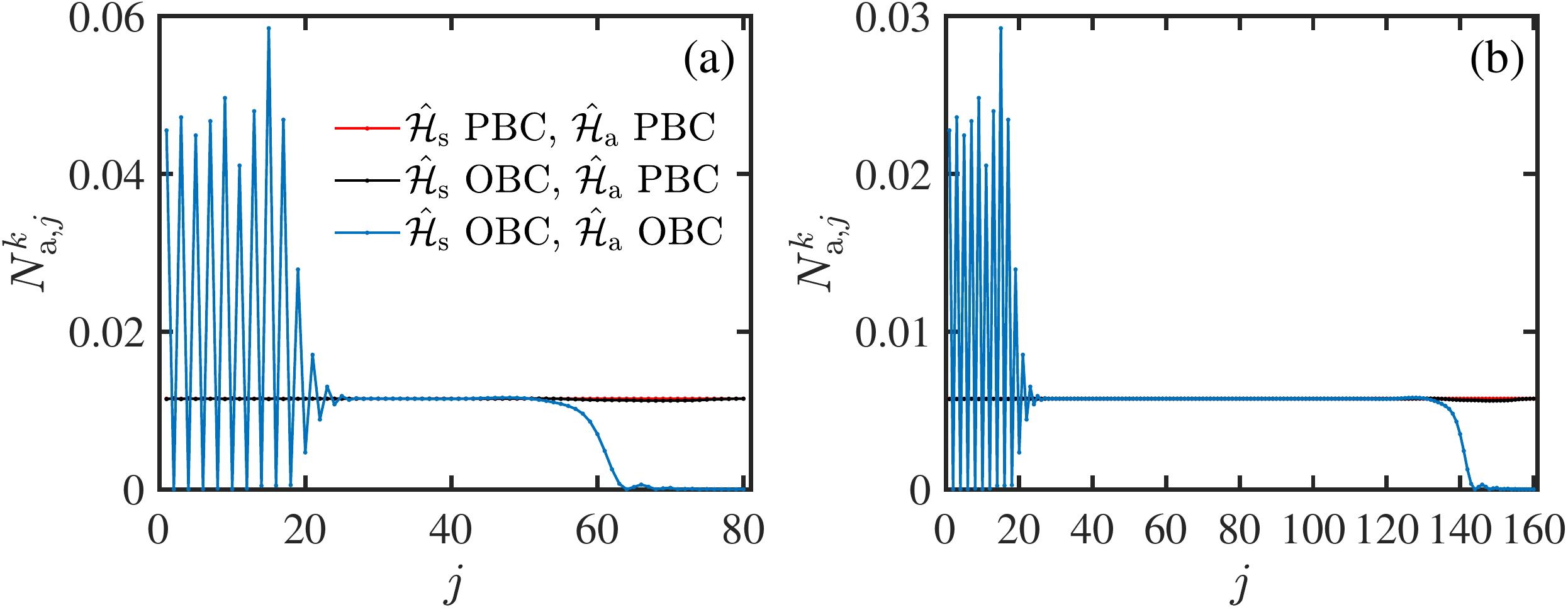}
\caption{
The occupancy of each auxiliary level for the non-Hermitian Rice-Mele model under different boundary conditions after a period of time $t$.
In (a), the system size $N=80$ and in (b), $N=160$.
Here, $J_1=2$, $J_2=1$, $J_3=1.2$, $m_z=0.6$, $\gamma=0.2$, $\Omega=0.03$, $\delta=2$, $k=3\pi/2$, $t=200$ and $J=-0.05$.
}
\label{figS5}
\end{figure*}

The boundary effects of $\hat{\mathcal{H}}_\text{s}$ and $\hat{\mathcal{H}}_\text{a}$
can also be illustrated by the occupancy of each auxiliary level,
$N_{\text{a},j}^k = \bra{\psi_\text{a}^k} e^{\ii \hat{\mathcal{H}}^\dagger t} \hat{a}_j^\dagger \hat{a}_j e^{-\ii \hat{\mathcal{H}} t} \ket{\psi_\text{a}^k}$ under different boundary conditions (see Fig.~\ref{figS5}).
We see that when we choose OBCs and PBCs for the system and auxiliary levels, respectively,
the occupancy is in good agreement with the result obtained by imposing PBCs on both systems,
which is consistent with the results in the first and third panel in Fig.~\ref{figS4}. 
Further imposing OBCs on the auxiliary levels in fact induces large discrepancies for $N_{\text{a},j}^k$ near boundaries compared with the results under PBCs; 
however, the discrepancies are mainly restricted to regions near boundaries, $j \lesssim 20$ or $j \gtrsim N-20$, implying that the boundary effects would vanish when we take an infinite size limit.
The drop (oscillations) near the right (left) boundary are mainly caused by reflection of the state $\ket{\psi_\text{a}^k}$ off an edge evolved by  $\hat{\mathcal{H}}_\text{a}$ for OBC auxiliary levels.
We thus can estimate the distance as $\Delta x \approx |\frac{\partial E_\text{a}(k)}{\partial k}| t = 2Jt |\sin k| = 20$ for the parameters in Fig.~\ref{figS4}, 
which agrees well with the fact that the oscillations and the drop of $N_{\text{a},j}^k$ mainly appear at $j \leq 20$ and $j \geq N-20$, respectively.

\subsection{B. Case 2}
In this subsection, we will use the finite temperature ensemble $\hat{\rho}_0 = e^{-\beta(\hat{\mathcal{H}}_\text{a} -\mu \hat{N}_\text{a})}/Z$ as
an initial density matrix to calculate the dynamics of the single-particle correlation function 
using Eq.~(\ref{CMatrixDynamics}). In this case, the non-Hermitian Rice-Mele Hamiltonian is the effective non-Hermitian Hamiltonian $\hat{\mathcal{H}}$ in the master equation 
with $\hat{L}_j=\sqrt{2\gamma} \hat{c}_{j2}$ as 
detailed in Section S-1.
We have also proved that the dynamics is governed by the single-particle non-Hermitian Hamiltonian in Subsection S-1 B [see Eq.~(\ref{CFtemp})]. 

We plot our numerical results in Fig.~\ref{figS11}.
We find that the spectral lines as well as the fitted spectra under OBCs also approach the PBC ones as $N$ increases, indicating that we can obtain the PBC spectra under OBCs in the thermodynamic limit for non-interacting bosons at a finite temperature.

\begin{figure*}[t]
	\includegraphics[width=7in]{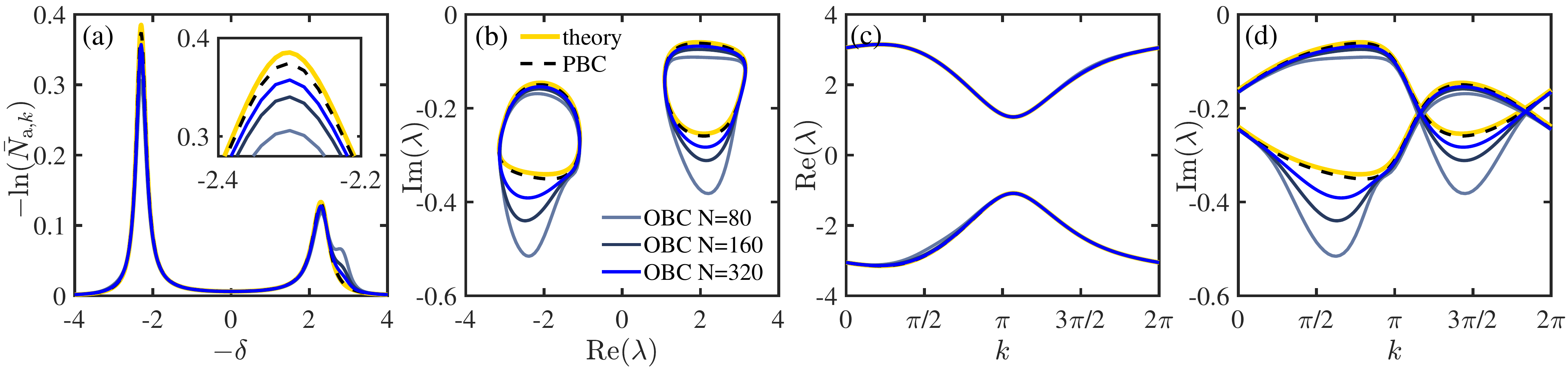}
	\caption{
		(a) The spectral lines of the non-Hermitian Rice-Mele model at $k=3\pi/2$ for bosonic systems at a finite temperature, where
		$\bar{N}_{\text{a},k} = N_{\text{a},k}/N_{0,k}$ with $N_{0,k} = \text{Tr}({\rho}_0 \hat{a}_k^\dagger \hat{a}_k)$.
		Both $N_{\text{a},k}$ and $N_{0,k}$ can be obtained by time-of-flight measurements in cold atom systems. 
		The yellow line is obtained based on Eq.~(\ref{RM_theo}) (here we need to change $N_{\text{a},k}$ to $\bar{N}_{\text{a},k}$),
		and the dashed black lines and a set of blue lines are calculated for a system
		under PBCs and OBCs, respectively.
		The inset shows the zoomed-in view of the spectral line near its highest absorption peak.
		The energy spectrum in the complex energy plane (b), real (c) and imaginary parts (d) of energies with respect to $k$.
		The energies are extracted by fitting the numerically calculated spectral lines
		under PBCs (dashed black lines) and OBCs (solid blue lines) in comparison with the theoretical result (the yellow lines).
		Here, $\beta=1$, $N_0 = 1000$ and other parameters are the same as Fig.~\ref{figS4}(b1)--(b4).
	}
	\label{figS11}
\end{figure*}

\section{S-5. Probing the non-Bloch energy}
In this section, we will discuss the possibility of probing non-Bloch energies for systems that exhibit NHSEs.
We start by preparing a single-particle non-Bloch state $\ket{\psi_0} = \ket{\psi_\text{a,NB}^\beta} = \hat{S}_\text{a}(r) \ket{\psi_\text{a}^q}/C_q (r)$ where the trajectory of $\beta$ forms the generalized Brillouin zone with the non-Bloch momentum defined as $k_\text{NB}=-\ii \ln \beta = q-\ii \ln r$, $\hat{S}_\text{a}(r) = \sum_{j=1}^N r^j \hat{a}_j^\dagger \ket{0} \bra{0} \hat{a}_j$, and $C_q (r)=\sqrt{\bra{\psi_\text{a}^q} \hat{S}_\text{a}^\dagger(r) \hat{S}_\text{a}(r) \ket{\psi_\text{a}^q}}$ is the normalization factor. 
This state with $r=\sqrt{|J_s -g|/|J_s + g|}$ is an approximate eigenstate of the Hatano-Nelson Hamiltonian under OBCs. 
In addition, we turn off the hopping between auxiliary levels such that $\hat{\mathcal{H}}_\text{a} = -\delta \sum_i \hat{a}_i^\dagger \hat{a}_i$ so that $\ket{\psi_0}$ is also an eigenstate of $\hat{\mathcal{H}}_\text{a}$.
We now focus on a non-Hermitian system that can be mapped to a Hermitian one by a similar transformation $\hat{S}_\text{s}(r)$ in terms of the same $r$ (we will clarify the reason in the following).
Suppose that the auxiliary levels are coupled to the first site of each unit cell of the system such that the full Hamiltonian has the form of Eq.~(1) in the main text.
Similar to Eq.~(2), we can also derive the population of auxiliary levels as
\begin{equation}
N_\text{a} (t) = \exp\left[-\frac{\Omega^2 t}{2} \sum_n \frac{a_{qn}^{(1)} \gamma_{qn} - b_{qn}^{(1)} (\varepsilon_{qn} + \delta)}{(\varepsilon_{qn} + \delta)^2 + \gamma_{qn}^2}\right],
\end{equation}
where $\varepsilon_{qn} - \ii \gamma_{qn}$ is the non-Bloch energy of the system Hamiltonian, and $a_{qn}^{(1)}$ and $b_{qn}^{(1)}$ are parameters depending on the system. 
For a non-Hermitian model considered here, $\gamma_{qn} = \gamma$ is a constant decay term added to ensure that the imaginary parts of energies are negative. 
In the following, we will first present the derivation details and some numerical results and then discuss the problems of this protocol.

We start from Eq.~(\ref{Eq11}) where $\Lambda_{ij} = 0$ [Eq.~(\ref{Eq15})] and
\begin{equation}
\begin{aligned}
\Gamma_{ij} (t,t') &= -\bra{\psi_0} e^{\ii \hat{\mathcal{H}}_\text{a} t} \hat{a}_i^\dagger \hat{c}_{i1} e^{-\ii \hat{\mathcal{H}}_\text{s} (t-t')}  \hat{c}_{j1}^\dagger  \hat{a}_j  e^{-\ii \hat{\mathcal{H}}_\text{a} t'}  \ket{\psi_0}
\\&= - e^{-\ii \delta(t-t')} \bra{\psi_0} \hat{a}_i^\dagger \hat{c}_{i1} \hat{S}_\text{s}(r') e^{-\ii \hat{\mathcal{H}}_\text{s}' (t-t') } \hat{S}_\text{s}^{-1} (r') \hat{c}_{j1}^\dagger  \hat{a}_j \ket{\psi_0}
\\& \approx - e^{-\ii \delta(t-t')} \bra{\psi_0} \hat{a}_i^\dagger \hat{c}_{i1} \hat{S}_\text{s}(r') e^{-\ii \hat{\mathcal{H}}_\text{s,PBC}' (t-t') } \hat{S}_\text{s}^{-1} (r') \hat{c}_{j1}^\dagger  \hat{a}_j \ket{\psi_0}.
\end{aligned}
\end{equation}
Here $\hat{\mathcal{H}}_\text{s}' = \hat{S}_\text{s}^{-1} (r') \hat{\mathcal{H}}_\text{s} \hat{S}_\text{s}(r')$ is a Hermitian Hamiltonian (up to a constant shift $-\ii \gamma$) under OBCs, where $\hat{S}_\text{s}(r)=\sum_{j=1}^{N} r^j \sum_\alpha f_\alpha \hat{c}_{j\alpha}^\dagger \ket{0} \bra{0} \hat{c}_{j\alpha}$ and $\hat{S}_\text{s}^{-1}(r) = \sum_{j=1}^{N} r^{-j} \sum_\alpha f_\alpha^{-1} \hat{c}_{j\alpha}^\dagger \ket{0} \bra{0} \hat{c}_{j\alpha}$ only acts on the system part. 
In the final step, we approximate the OBC Hermitian Hamiltonian $\hat{\mathcal{H}}_\text{s}'$ by its PBC version $\hat{\mathcal{H}}_\text{s,PBC}'$.
Note that using this approximation, a non-Bloch state appears as an approximate eigenstate of a non-Hermitian OBC Hamiltonian. 
Then we insert an identity $\sum_{kn} \ket{\psi_\text{PBC}^{kn}} \bra{\psi_\text{PBC}^{kn}} + \sum_k \ket{\psi_\text{a}^k} \bra{\psi_\text{a}^k} = 1$ into the above equation. 
Here, $\ket{\psi_\text{PBC}^{kn}}$ is an eigenstate of $\hat{\mathcal{H}}_\text{s,PBC}'$ labeled by the momentum $k$ and the band index $n$ with the corresponding eigenenergy being $\varepsilon_{kn} - \ii \gamma_{kn}$ (which is the non-Bloch energy of the OBC system Hamiltonian) and $\ket{\psi_\text{a}^k}$ is the momentum state on the auxiliary levels. 
Using the fact that $\bra{\psi_\text{a}^k} \hat{S}_\text{s}^{-1} (r')  \hat{c}_{j1}^\dagger  \hat{a}_j \ket{\psi_0} = 0$, we arrive at
\begin{equation}
\begin{aligned}
\Gamma_{ij} (t,t') & = - e^{-\ii \delta(t-t')} \bra{\psi_0} \hat{a}_i^\dagger \hat{c}_{i1} \hat{S}_\text{s}(r') e^{-\ii \hat{\mathcal{H}}_\text{s,PBC}' (t-t') } \sum_{kn} \ket{\psi_\text{PBC}^{kn}} \bra{\psi_\text{PBC}^{kn}} \hat{S}_\text{s}^{-1} (r') \hat{c}_{j1}^\dagger  \hat{a}_j \ket{\psi_0}
\\&= - \sum_{kn} e^{-\ii (\varepsilon_{kn} + \delta - \ii \gamma_{kn})(t-t')} \bra{\psi_0} \hat{a}_i^\dagger \hat{c}_{i1} \hat{S}_\text{s}(r') \ket{\psi_\text{PBC}^{kn}} \bra{\psi_\text{PBC}^{kn}} \hat{S}_\text{s}^{-1} (r') \hat{c}_{j1}^\dagger  \hat{a}_j \ket{\psi_0}.
\end{aligned}
\end{equation}
We can also derive the result that $\sum_j \hat{c}_{j1}^\dagger \hat{a}_j \ket{\psi_0} = \hat{S}_\text{s} (r) \ket{\psi_\text{s}^{q1}}/Z_{q1} (r)$ where $\ket{\psi_\text{s}^{q1}} = \frac{1}{\sqrt{N}} \sum_x e^{\ii q x} \hat{c}_{x1}^\dagger \ket{0}$ is a momentum state located on the first site of each unit cell in the system and $Z_{q1} (r) = \sqrt{\bra{\psi_\text{s}^{q1}} \hat{S}_\text{s}^\dagger (r) \hat{S}_\text{s} (r) \ket{\psi_\text{s}^{q1}}}$ is the normalization factor. 
Then we have
\begin{equation}
\label{NB_sum_kn_equation}
\sum_{ij} \Gamma_{ij} (t,t') = - \sum_{kn} e^{-\ii (\varepsilon_{kn} + \delta - \ii \gamma_{kn})(t-t')} \frac{1}{Z_{q1}^2 (r)} \bra{\psi_\text{s}^{q1}} \hat{S}_\text{s}^\dagger (r) \hat{S}_\text{s}(r') \ket{\psi_\text{PBC}^{kn}} \bra{\psi_\text{PBC}^{kn}} \hat{S}_\text{s}^{-1} (r') \hat{S}_\text{s} (r) \ket{\psi_\text{s}^{q1}}.
\end{equation}
If $r'=r$ (the imaginary parts of the non-Bloch momentum for the initial state and the non-Hermitian system are equal), then $\hat{S}_\text{s}^{-1} (r') \hat{S}_\text{s} (r) = 1$ so that we can significantly simplify the equation with only one momentum component remaining, that is,
\begin{equation}
\begin{aligned}
\sum_{ij} \Gamma_{ij} (t,t') &= - \sum_{n} e^{-\ii (\varepsilon_{qn} + \delta - \ii \gamma_{qn})(t-t')} \frac{1}{Z_{q1}^2 (r)} \bra{\psi_\text{s}^{q1}} \hat{S}_\text{s}^\dagger (r) \hat{S}_\text{s}(r) \ket{\psi_\text{PBC}^{qn}} 
\langle \psi_\text{PBC}^{qn} | \psi_\text{s}^{q1} \rangle
\\ &= - \sum_{n} e^{-\ii (\varepsilon_{qn} + \delta - \ii \gamma_{qn})(t-t')} 
(a_{qn}^{(1)} + \ii b_{qn}^{(1)}),
\end{aligned}
\end{equation}
where $a_{qn}^{(1)} = \text{Re}(c_{qn}^{(1)})$ and $b_{qn}^{(1)} = \text{Im}(c_{qn}^{(1)})$ with 
$c_{qn}^{(1)} = \frac{1}{Z_{q1}^2 (r)} \bra{\psi_\text{s}^{q1}} \hat{S}_\text{s}^\dagger (r) \hat{S}_\text{s}(r) \ket{\psi_\text{PBC}^{qn}} \langle \psi_\text{PBC}^{qn} | \psi_\text{s}^{q1} \rangle$.
Now we evaluate $\dot{N}_\text{a} (t)$ which is given by 
\begin{equation}
\begin{aligned}
\dot{N}_\text{a} (t) &= 
\frac{\Omega^{2}}{4} \int_{0}^{t} dt'  \sum_{ij} \Gamma_{ij} (t,t') + \mathrm{H.c.}
\\&= -\frac{\Omega^{2}}{4} \sum_{n} \int_{0}^{t} dt' e^{-\gamma_{qn} (t-t')} e^{-\ii (\varepsilon_{qn} + \delta)(t-t')} (a_{qn}^{(1)} + \ii b_{qn}^{(1)}) + \text{H.c.}.
\end{aligned}
\end{equation}
Since the above equation has the same form as Eq.~(\ref{Eq1}), we conclude that
\begin{equation}
\label{NB_conclusion}
N_\text{a} (t) = \exp(-\frac{\Omega^2 t}{2} \sum_n \frac{a_{qn}^{(1)} \gamma_{qn} - b_{qn}^{(1)} (\varepsilon_{qn} + \delta)}{(\varepsilon_{qn} + \delta)^2 + \gamma_{qn}^2})
\end{equation}
based on the derivation in Section S-2, which may allow us to extract the non-Bloch energy by fitting the spectral lines.

\begin{figure*}[t]
\includegraphics[width=5.5in]{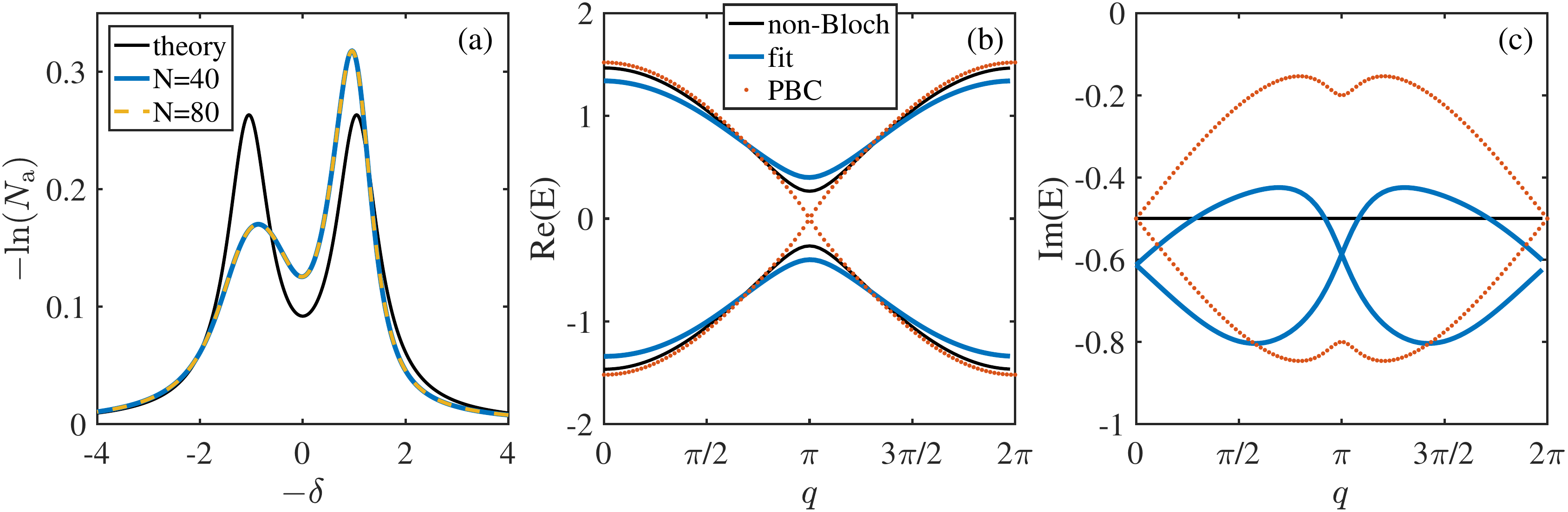}
\caption{
(a) The spectral lines of the non-Hermitian SSH model at $q=\pi/2$ with the black line obtained from Eq.~(\ref{NB_conclusion}) and the colored lines numerically computed. 
The real (b) and the imaginary (c) parts of the non-Bloch energy with respect to $k$, either obtained by non-Bloch theorem (black lines) or extracted by fitting the spectral lines (blue lines). 
In comparison, we also plot the eigenenergies under PBCs as red dots. 
Here we set $J_s = -1$, $g=-0.5$, $J_2 = -0.6$ and $\gamma = 0.5$ so that the system does not have zero-energy edge states.
}
\label{figS7}
\end{figure*}

We now take the non-Hermitian SSH model [Eq.~(6) in the main text] as an example to show whether this protocol can be used to measure the non-Bloch energy spectrum and then discuss the underlying problems for this protocol. 
We find that the measured results described by blue and yellow lines (numerically simulated) deviate significantly from the theoretical result (the black line) as shown in Fig.~\ref{figS7}(a). 
Such a deviation does not decrease as we increase the system size due to the fact that the dynamics is mainly dominated by the evolution of a state mainly localized at a boundary. 
The discrepancy may be attributed to the approximation ($\hat{\mathcal{H}}_\text{s}' \approx \hat{\mathcal{H}}_\text{s,PBC}'$) we employed to derive $N_\text{a}(t)$. 
Although one can safely use this approximation for treating the dynamics of bulk states, a slight difference at a boundary between $\ket{\psi_\text{PBC}^{kn}}$ (eigenstate of $\hat{\mathcal{H}}_\text{s,PBC}'$) and a corresponding eigenstate of $\hat{\mathcal{H}}_\text{s}'$ may be amplified due to the fact that the initial state is exponentially localized at a boundary, which is also the reason why the results do not improve with the system size.

If we insist on fitting the measured spectral line based on the formula of $N_\text{a} (t)$, we can extract the real and imaginary parts of the non-Bloch energy spectrum as shown in the Figs.~\ref{figS7}(b) and (c).
Despite some discrepancy, we find that the fitted real energy spectrum (blue lines) is quite close to the theoretical results (black lines). The fitted energy spectrum also reveals the gapped feature of the non-Bloch energy spectrum (or OBC energy spectrum) at $q=\pi$ in stark contrast to the gapless feature of the PBC energy spectrum (red dots). But for the imaginary parts, the measured ones exhibit significant difference from the theoretical ones.

To sum up, if we prepare a non-Bloch state as an initial state and suppose that its imaginary part of the non-Bloch momentum is equal to the imaginary part of the non-Bloch momentum for a system Hamiltonian, then we may obtain an energy spectrum whose real parts resemble those of the non-Bloch energy spectrum and whose imaginary parts exhibit significant difference.

\emph{We now discuss the problems of this protocol.}

(1) In the derivation, we have assumed that the imaginary part $-\ln r$ of the non-Bloch momentum of the initial state is equal to the imaginary part $-\ln r^\prime$ of the non-Bloch momentum of a system Hamiltonian, i.e., $r=r^\prime$. 
This means that we need to know $r'$ in advance.
In other words, we need to first measure the generalized Brillouin zone using other methods; how to obtain the information in fact still remains an open question in cold atom systems. 
Note that our method cannot be used to extract $r'$ because when $r\neq r'$, $\bra{\psi_\text{PBC}^{kn}} \hat{S}_\text{s}^{-1} (r') \hat{S}_\text{s} (r) \ket{\psi_\text{s}^{q1}} \neq  \langle \psi_\text{PBC}^{kn}  | \psi_\text{s}^{q1} \rangle$ since $\hat{S}_\text{s}^{-1} (r') \hat{S}_\text{s} (r) \neq 1$. 
Then based on Eq.~(\ref{NB_sum_kn_equation}), we have
\begin{equation}
\sum_{ij} \Gamma_{ij} (t,t') = - \sum_{kn} e^{-\ii (\varepsilon_{kn} + \delta - \ii \gamma_{kn})(t-t')} (a_{kqn}^{(1)} + \ii b_{kqn}^{(1)}),
\end{equation}
where $a_{kqn}^{(1)} + \ii b_{kqn}^{(1)} = \frac{1}{Z_{q1}^2} \bra{\psi_\text{s}^{q1}} \hat{S}_\text{s}^\dagger (r) \hat{S}_\text{s}(r') \ket{\psi_\text{PBC}^{kn}} \bra{\psi_\text{PBC}^{kn}} \hat{S}_\text{s}^{-1} (r') \hat{S}_\text{s} (r) \ket{\psi_\text{s}^{q1}}$.
Similarly, one can obtain
\begin{equation}
N_\text{a} (t) = \exp(-\frac{\Omega^2 t}{2} \sum_{kn} \frac{a_{kqn}^{(1)} \gamma_{kn} - b_{kqn}^{(1)} (\varepsilon_{kn} + \delta)}{(\varepsilon_{kn} + \delta)^2 + \gamma_{kn}^2})
\end{equation}
For each $q$, this expression does not resolve the momentum, but instead involves the contribution from all momenta, from which it is very hard to extract the information of the non-Bloch energy $\varepsilon_{kn} - \ii \gamma_{kn}$.
It suggests that without knowing the generalized Brillouin zone of the system Hamiltonian, it is very hard (if not impossible) to extract the non-Bloch energy using this method.

(2) Even we know the generalized Brillouin zone in advance, it also raises a significant challenge to prepare a non-Bloch state on the auxiliary levels in experiments. 
If we consider non-Hermitian auxiliary levels, it remains an open question of whether we can prepare a non-Bloch state by slowly tuning a system parameter, such as $g$ in the HN model. 
Even we could do this, atoms on the auxiliary levels will decay, and a significant loss of atoms after a long period of time would make it impossible to extract the energy spectrum.

(3) Finally, this protocol also requires that the hopping between auxiliary levels must vanish. 
Otherwise, the non-Bloch state is no longer an eigenstate of the auxiliary system, and an analytical results for $N_\text{a} (t)$ would become very complicated, making it hard to extract the spectrum information.

\section{S-6. More results about the non-Hermitian SSH model}
In this section, we show more results about the non-Hermitian SSH model in order to address the following three questions: 
\begin{enumerate}
\item[i)] Does the fact of the equivalence between the results of $k$-resolved measurements for open and periodic boundaries in the thermodynamic limit
mean that the information about zero-energy edge states is not measurable given that the boundary effects are invisible?
\item[ii)] Why is there no signal of the right edge state in Fig.~4 in the main text?
\item[iii)] How well does the protocol work under OBCs in the case where the eigenspectra host exceptional points?
\end{enumerate}

\subsection{A. Boundary effects in the non-Hermitian SSH model}

\begin{figure*}[t]
\includegraphics[width=4.0in]{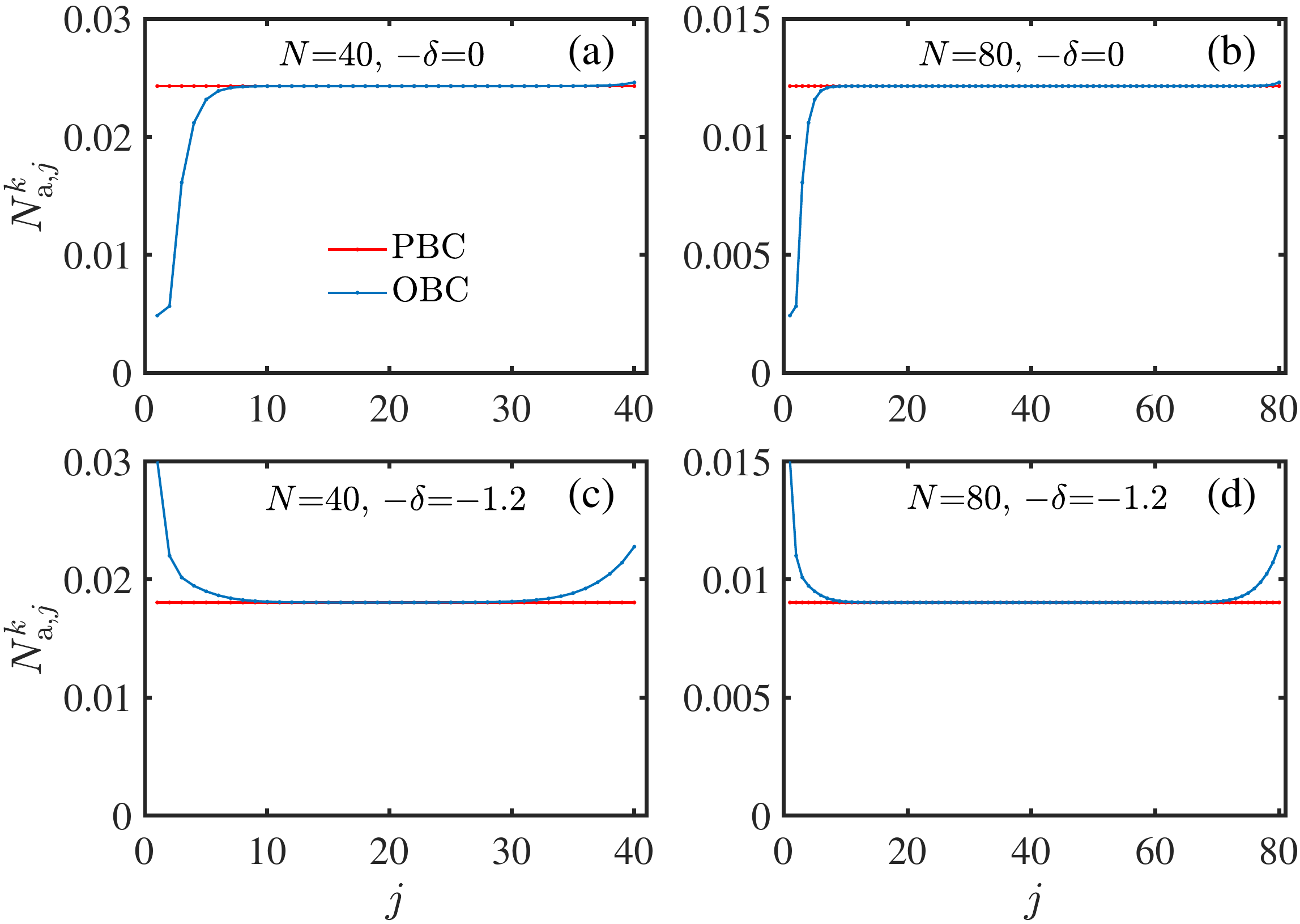}
\caption{
The occupancy of each auxiliary level $N_{\text{a},j}^k$ for the non-Hermitian SSH model under OBCs and PBCs at different system sizes $N$ and detuning $\delta$.
In (a) and (b), the detuning ($-\delta=0$) matches the real energy of zero-energy edge states, leading to a large absorption at the left boundary.
In (c) and (d), there is no contribution of the edge states due to the large detuning $-\delta=-1.2$.
Here, the parameters are $J_s=-1$, $J_2=-2.5$, $g=-0.1$, $\gamma=0.1$ and $J=0$, which are the same as those in Fig.~4(a) in the main text.
}
\label{figS8}
\end{figure*}

To address the first question, we prepare a momentum state $\ket{\psi_\text{a}^k}$ (here we take $k=\pi$) as an initial state on the auxiliary levels and observe the local occupancy 
($N_{\text{a},j}^k = \bra{\psi_\text{a}^k} e^{\ii \hat{\mathcal{H}}^\dagger t} \hat{a}_j^\dagger \hat{a}_j e^{-\ii \hat{\mathcal{H}} t} \ket{\psi_\text{a}^k}$)
of each auxiliary level for the non-Hermitian SSH model under PBCs and OBCs after a long period of time (see Fig.~\ref{figS8}).

Fig.~\ref{figS8} illustrates that $N_{\text{a},j}^k$ under PBCs and OBCs indeed differ significantly at the boundaries, indicating that the OBC Hamiltonian has contributions at the edges. In the presence of zero-energy edge states, $N_{\text{a},j}^k$ exhibits a drop at the left boundary when $-\delta=0$ [see (a) and (b)], showing that one can measure the isolated edge states by measuring $N_{\text{a},j}^k$.

However, if we focus on the $N_{\text{a},j}^k$ in the bulk, we see that the results under PBCs and OBCs agree very well. In fact, the difference only occurs in the vicinity of boundaries. These regions are restricted to $j \lesssim 10$ or $j \gtrsim N-10$ and do not increase as we increase the system size $N$. For the $k$-resolved measurement, we probe the particle number, $N_\text{a} = \sum_j N_{\text{a},j}^k$. It is clear to see that in the thermodynamic limit, this quantity is completely determined by the particle number distribution in the bulk as the boundary contribution to 
$N_{\text{a}}$ becomes vanishingly small compared with the contribution of the bulk. If we consider many particle occupations, we can resolve the momentum state through time-of-flight measurements so that the problem is reduced to the single-particle one (see our proof in Supplementary Section S-1).

We thus conclude that bulk and edge states mainly affect the dynamics of the particle distribution in the bulk and near boundaries, respectively. If we measure the occupancy near a boundary, we may obtain the information about zero-energy edge states; if we measure the occupancy at each momentum through time-of-flight measurements, we can extract the Bloch energy spectrum under PBCs.

\subsection{B. Probing the right edge state of the non-Hermitian SSH model}
In the main text, we find that only the left zero-energy edge state responds to our probe although
there are two zero-energy edge states [one (the other) is localized at the left (right) boundary] for the system parameters we considered in Fig.~4.
The reason why we only probe the left edge state is due to the fact that the edge states at the left and right edges are 
mainly localized at the A and B sites, respectively.
In Fig.~4, we only couple the auxiliary levels to the A sites, and thus only the left edge state will respond to the probe. 
If we want to probe the right edge state, we need to couple the auxiliary level to the B sites and then the response of the edge state would appear at the right boundary (see the red line in Fig.~\ref{figS9}).
The results for both cases can be found in Fig.~\ref{figS9}, corresponding to the inset of Fig.~4(a) in the main text.

\begin{figure*}[t]
\includegraphics[width=2.2in]{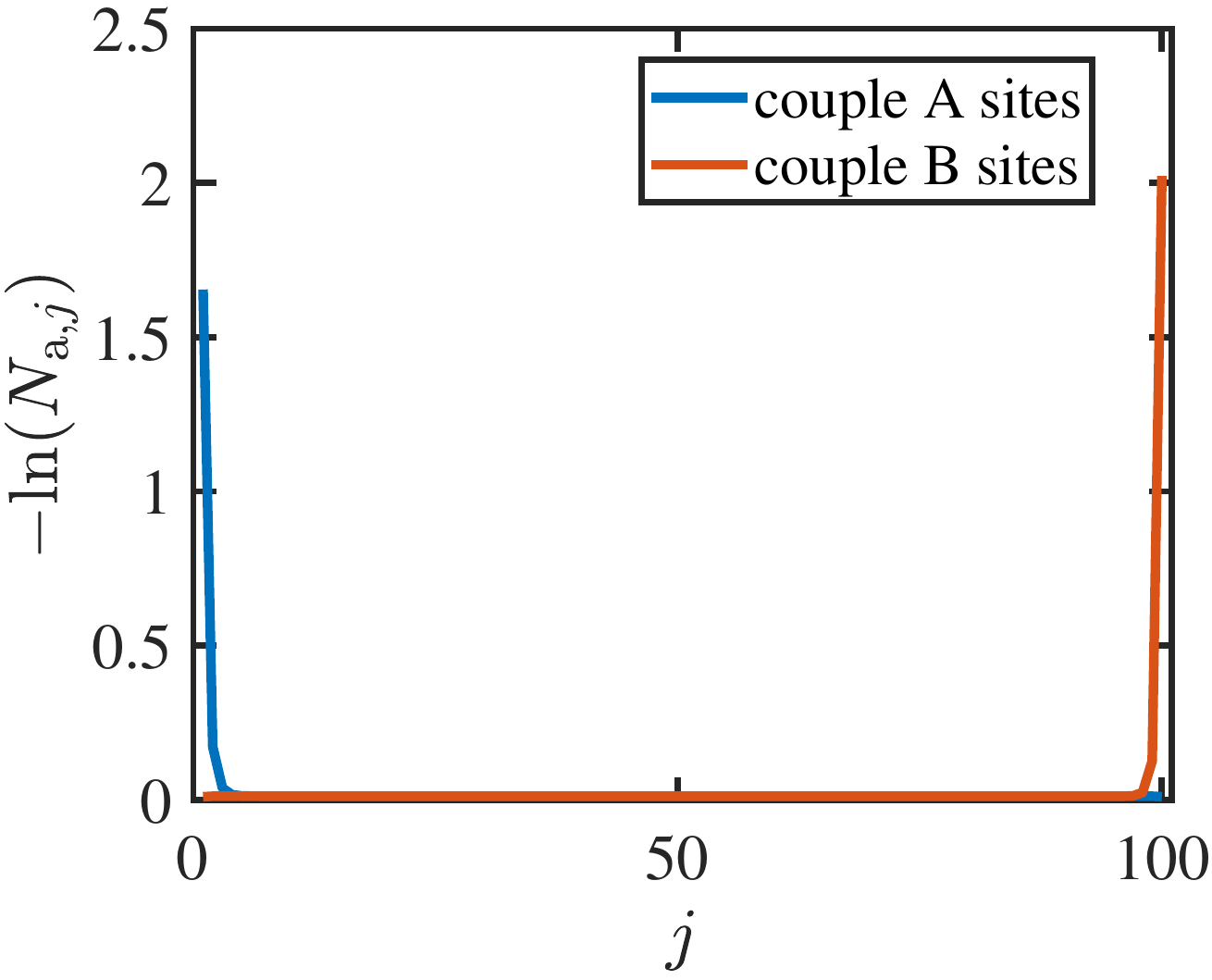}
\caption{
The indicator of local absorption $-\ln (N_{\text{a},j})$ on the $j$th auxiliary level at $\delta=0$ when the auxiliary levels are coupled to A sites (blue line) or B sites (red line).
Here, the parameters are $N=100$, $J_s=-1$, $J_2=-2.5$, $g=-0.1$, $\gamma=0.1$ and $J=0$.
}
\label{figS9}
\end{figure*}

\subsection{C. Extracting the bulk spectra for a system under OBCs with exceptional points}

\begin{figure*}[t]
\includegraphics[width=6.1in]{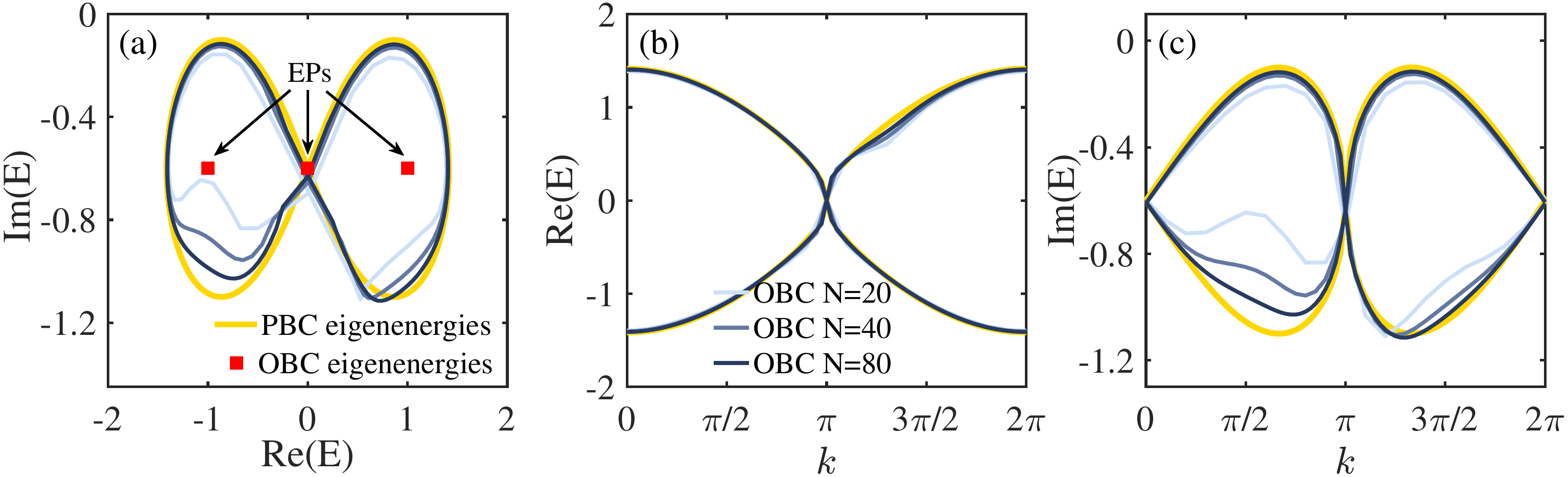}
\caption{
(a) The energy spectra under PBCs in the complex energy plane. The eigenenergies under OBCs are plotted as red squares. 
Real (b) and imaginary parts (c) of eigenenergies with respect to $k$.
The yellow lines are
obtained by diagonalizing the k-space Hamiltonian, and the blue lines are extracted by fitting the spectral lines under OBCs for different system sizes. 
Here $J_s = -0.5$, $g=-0.5$, $J_2 = -1$ and $\gamma = 0.6$.
}
\label{figS10}
\end{figure*}

In this subsection, we will show that the bulk spectra can be probed in a limiting case where all OBC eigenenergies collapse into multiple exceptional points.
Here we choose a specific set of parameters for the non-Hermitian SSH model where all the OBC eigenenergies collapse into three exceptional points [shown in Fig.~\ref{figS10}(a) by red squares].
We find that although the eigenenergies under PBCs and OBCs show significant differences, the probed spectra in an open boundary geometry (blue lines) still approaches the  eigenenergies under PBCs when increasing $N$.
We also note that the exceptional point under PBCs at $E=-\ii \gamma$ can also be captured by the fitted spectra.
Such an exceptional point marks a topological phase transition between $W=1$ (two zero modes) and $W=1/2$ (one or no zero mode), where the winding number $W$ is defined as \cite{Ueda2018PRX,Xu2020PRB} $W=(w_1 - w_2)/2$ with $w_{1,2} = \int_{0}^{2\pi} \frac{dk}{2\pi \ii} \partial_k \ln \det h_{1,2}(k)$ and $h_{1,2} (k)$ denoting the two off-diagonal block of chiral-symmetric Hamiltonians (the non-Hermitian SSH model preserves the chiral symmetry with respect to a constant shift $-\ii \gamma$).

\section{S-7. Absorption spectroscopy for interacting systems}
In this section, we show that the absorption spectroscopy can also be applied to cases where there are interactions for atoms on the system levels. 
In light of the fact that with interactions the problem becomes very complicated,
we consider two simplified cases.
(1) Consider that there are interactions between atoms in the non-Hermitian system (atoms on the auxiliary levels do not interact). Atoms are initially prepared on the auxiliary levels. In this case, we find that our spectroscopy can measure the single-particle part of the system Hamiltonian. In other words, the measurement outcomes are independent of interactions. A brief explanation for this result is that atoms decay rapidly after being transferred to the system part so that it is not possible to have two atoms in the system at the same time, leading to the fact that the interaction part of the system Hamiltonian has no contribution to the dynamics.
(2) Suppose that there is one atom on the system levels for the initial state. We find that our spectroscopy can give the two-particle energy spectrum.

\subsection{A. Case 1}
We now elaborate on the first case. Consider a general density-density interaction
$\hat{\mathcal{H}}_\text{int}=\frac{1}{2}\sum_{m,n\in \text{sys}} U_{mn} \hat{c}_m^\dagger \hat{c}_n^\dagger \hat{c}_n \hat{c}_m$ with $U_{mn} = U_{nm}$ and $U_{nn} = 0$.
Here we use ``sys'' and ``aux'' to denote degrees of freedom in the system and auxiliary levels, respectively. 
Based on the derivation in Section S-1 [Eqs.~(\ref{derivation_CMatrixDynamics_first})--(\ref{CMatrixDynamics})], 
the dynamics of the correlation function is described by
\begin{equation}
\label{CMatrixDynamics_MB}
\begin{aligned}
\frac{dC_{ij}(t)}{dt} 
&= [X C(t) + C(t) X^\dagger]_{ij} 
- \frac{\ii}{2} \sum_{m,n\in \text{sys}} U_{mn} \text{Tr} \big( \rho(t) 
[\hat{c}_i^\dagger \hat{c}_j,\hat{c}_m^\dagger \hat{c}_n^\dagger \hat{c}_n \hat{c}_m] \big)
\\&= [X C(t) + C(t) X^\dagger]_{ij} 
- \ii \sum_{m,n\in \text{sys}} U_{mn} \text{Tr}\big(  \rho(t) (
\delta_{mj} \hat{c}_i^\dagger \hat{c}_n^\dagger \hat{c}_n \hat{c}_m
- \delta_{im} \hat{c}_m^\dagger \hat{c}_n^\dagger \hat{c}_n \hat{c}_j
) \big).
\end{aligned}
\end{equation}
Let $D_{ij}^\text{int} (t)$ be the interaction term in the above equation. 
We apply Wick's theorem to approximate the four-point correlation function, which gives
\begin{equation}
\begin{aligned}
D_{ij}^\text{int} (t) &= - \ii \sum_{m,n\in \text{sys}} U_{mn} \text{Tr}\big(  \rho(t) (
\delta_{mj} \hat{c}_i^\dagger \hat{c}_n^\dagger \hat{c}_n \hat{c}_m
- \delta_{im} \hat{c}_m^\dagger \hat{c}_n^\dagger \hat{c}_n \hat{c}_j
) \big)
\\& \approx - \ii \sum_{m,n\in \text{sys}} U_{mn} \big[  
\delta_{mj} \big( C_{nn}(t) C_{im}(t) - C_{nm}(t) C_{in}(t) \big) 
- \delta_{im} \big( C_{nn}(t) C_{mj}(t) - C_{mn}(t) C_{nj}(t) \big)
\big].
\end{aligned}
\end{equation}
Next, we replace $C_{mn}(t)$ with $C_{mn}^0(t)$ where $C_{mn}^0(t)$ is determined by the non-interacting part of the effective Hamiltonian:
\begin{equation}
\frac{d C_{mn}^0(t)}{dt} = [X C^0 (t) - C^0 (t) X^\dagger]_{mn}.
\end{equation}
Here for clarity and simplicity, we only take, for example, $C_{nn}(t) C_{im}(t) \approx C_{nn}^0 (t) C_{im} (t)$.
The physics will not change if we let $C_{nn}(t) C_{im}(t) \approx [C_{nn}^0 (t) C_{im} (t) + C_{nn} (t) C_{im}^0 (t)]/2$, which is more accurate. 
Then we have
\begin{equation}
\begin{aligned}
D_{ij}^\text{int} (t) &\approx -\ii \sum_{m,n\in \text{sys}} U_{mn} \big[ 
\big( \delta_{mj}  C_{nn}^0(t) - \delta_{nj} C_{mn}^0(t) \big) C_{im}(t)
- \big( \delta_{im} C_{nn}^0(t) - \delta_{in} C_{nm}^0(t) \big) C_{mj}(t)
\big]
\\&= \sum_{m\in\text{sys}} Y_{im} (t) C_{mj} (t) + C_{im} (t) Y_{jm}^* (t),
\end{aligned}
\end{equation}
where $Y_{im} (t) = \ii \sum_{n \in \text{sys}} U_{mn} \big( \delta_{im} C_{nn}^0 (t) - \delta_{in} C_{nm}^0 (t) \big)$ if $m\in \text{sys}$ and $Y_{im} (t) = 0$ if $m \in \text{aux}$.
Then the full dynamics can be approximated by
\begin{equation}
\frac{d C_{ij} (t)}{dt} \approx [\tilde{X}(t) C(t) + C(t) \tilde{X}^\dagger (t)]_{ij}
\end{equation}
with $\tilde{X} (t) = X + Y(t)$.
Now we compare the magnitude of $X$ and $Y(t)$.
We first note that $C_{ij}^0 (t) \sim 0$ if $i$ or $j$ belongs to the system levels due to the dissipative nature of the system.
Given that $U_{mn}$ and $H_{mn}$ typically have the same order of magnitude, we would have $\tilde{X}(t) \approx X$ since the elements of $Y(t)$ is negligibly small compared with $X$. 
Thus, the dynamics contributed by the interaction term is insignificant compared with the non-interacting part, suggesting that we can obtain the single-particle spectrum of the non-interacting part of the effective Hamiltonian by absorption spectroscopy.
We remark that the result will not change if we consider terms like $C_{nn}(t) C_{im}^0 (t)$. 
Even if we take $i\in \text{aux}$, $C_{im}^0 (t)$ is still close to zero because it involves a system level $m$ (here $m$ can only be system levels since there is no interaction between system and auxiliary levels).
However, if the system interacts with the auxiliary levels, then $C_{im}^0 (t)$ can be finite because we can take both $i$ and $m$ belonging to the auxiliary levels, leading to $C_{im}^0 (t) \sim \delta_{im}$.
Then by the arguments the interaction cannot be neglected, since the contribution of the interaction part may have the same order of magnitude as the non-interaction part $X$.

\begin{figure*}[t]
\includegraphics[width=4.2in]{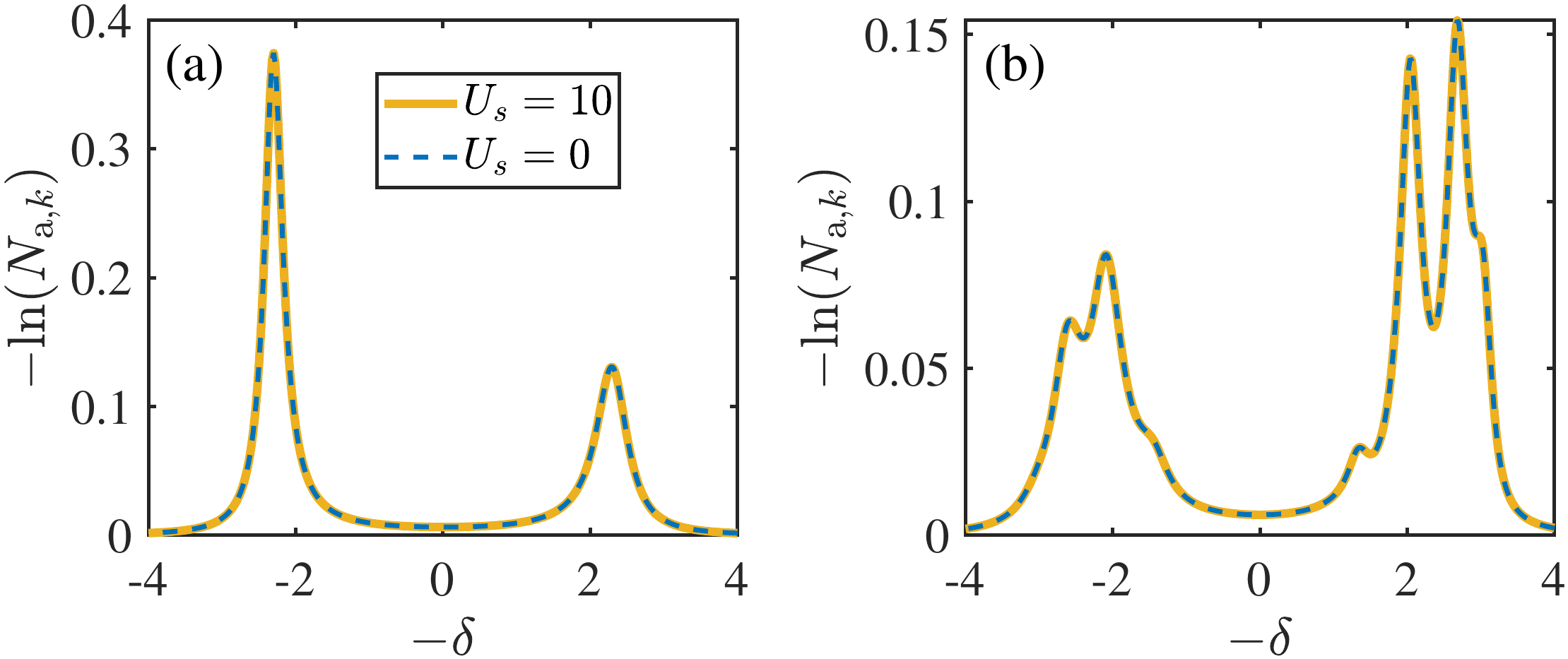}
\caption{
The spectral lines for the non-Hermitian Rice-Mele model under PBCs (a) and OBCs (b) with (solid lines) and without interactions (dashed lines).
Here $N=4$, $J_1 = 2$, $J_2 = 1$, $m_z = 0.6$, $J_3=1.2$, $\gamma=0.2$, $\Omega=0.03$, $t=200$ and $J=-0.05$.
}
\label{figS12}
\end{figure*}

We now numerically confirm that the measured spectra are independent of interactions by numerically simulating the master equation in many-particle bases
using the Krylov subspace method for the initial state $|\psi_0^{(M)}\rangle=\prod_j\hat{a}_j^\dagger|0\rangle$. 
We consider the non-Hermitian Rice-Mele model [Eq.~(\ref{full_Ham_NHRM})] with an interaction described by
$\hat{\mathcal{H}}_\text{int} = U_s \sum_j \hat{c}_{j1}^\dagger \hat{c}_{j1} \hat{c}_{j2}^\dagger \hat{c}_{j2}$ (which can be realized by short-range interactions between two hyperfine states in cold atoms). Fig.~\ref{figS12} shows the spectral lines for a system with $U_s=10$ (yellow lines) and $U_s$=0 (blue lines) under periodic (left) and open (right) boundary conditions. We see that the results with or without interactions exhibit no difference, confirming our previous argument that the measured energy spectra are independent of interactions. In other words, our method can be used to measure the single-particle energy spectrum of a dissipative interacting system. 
Here for the sake of computational simplicity, we take the system to be fermionic.
For bosons, the conclusion would be the same similar to the fermionic case if atoms on the auxiliary levels do not interact.

\subsection{B. Case 2}
To demonstrate that our methods may also be generalized to measure the many-particle energy levels in a dissipative interacting system, 
we consider a simple case where there is an atom in the system Hamiltonian and $N-1$ atoms on the auxiliary levels for the initial state, that is, 
the initial state is changed to $\ket{\psi_{0,k\alpha}^{(M)}} = \hat{c}_{k\alpha}^\dagger \hat{a}_k \ket{\psi_{0}^{(M)}}$. 
The initial state $\ket{\psi_{0,k\alpha}^{(M)}}$ may be realized in experiments by first preparing the $\ket{\psi_{0}^{(M)}}$ state and 
then transferring one particle, say $\hat{a}_k^\dagger \ket{0}$, to the system levels. We note that for more particles present in the system, the problem will become much more complicated, and we thus leave the general case for our future research.

Since the initial state contains a particle on the system levels, other particles will interact with this particle after being transferred to the system levels. As a result, the occupation $N_{\text{a},q} = \text{Tr} (\rho(t) \hat{a}_q^\dagger \hat{a}_q)$ on each momentum state after the evolution via the master equation may give some signals related to the interaction. To validate this argument, we simulate the many-body dynamics of a $N=3$ non-Hermitian Rice-Mele model with the initial state being 
$\ket{\psi_{0,k\alpha}^{(M)}}$ (here we set $k=0$ and $\alpha=1$) under PBCs.

Figure~\ref{figS13} plots the spectral lines for a system with (solid blue lines) and without (solid red lines) interactions. We see that the presence of interactions leads to additional peaks besides a sharp peak corresponding to a single-particle energy in momentum space with momentum $q$. Such small peaks result from two-particle energies due to interactions (see the following discussion). These peaks are much smaller than the peaks for the single-particle energy because the atom on the system level will decay to zero after a period of time $t_0$. When $t>t_0$, the interaction term will not contribute to the dynamics, leading to a smaller spectral response of the interaction terms at long times compared with that of the single-particle part.

\begin{figure*}[t]
\includegraphics[width=4.2in]{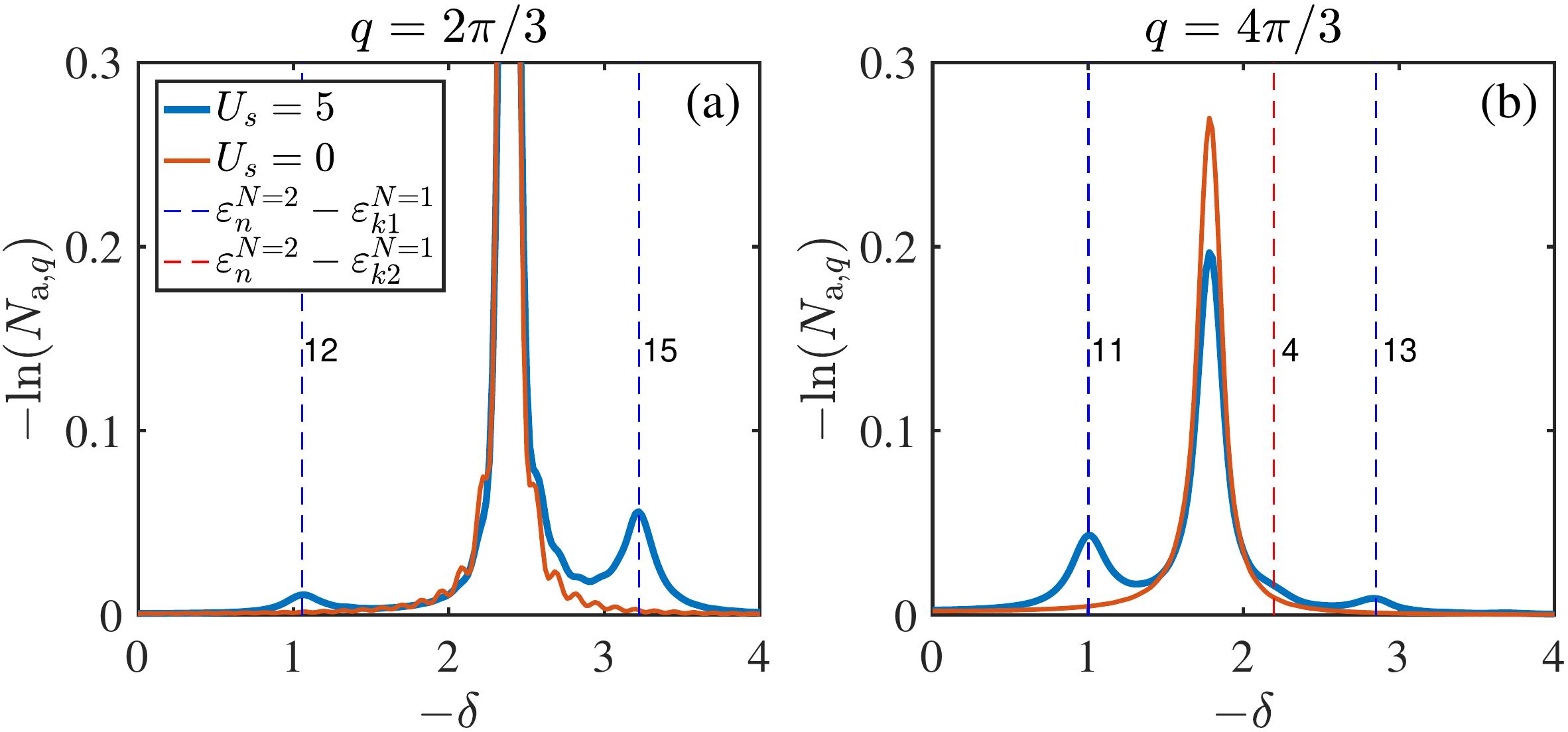}
\caption{
The spectral lines for the non-Hermitian Rice-Mele model under PBCs at (a) $q=2\pi/3$ and (b) $q=4\pi/3$ with (blue solid lines) and without interactions (red solid lines). 
We also plot the real part of two-particle eigenenergies $\varepsilon_{n}^{N=2}$ measured with respect to the real part of single-particle eigenenergies $\varepsilon_{k1}^{N=1}$ and $\varepsilon_{k2}^{N=1}$ as vertical dashed lines, with the numbers beside the vertical dashed lines denoting the index $n$ of two-particle eigenenergies.
Here $N=3$, $J_1 = 2$, $J_2 = 1$, $m_z = 0.6$, $J_3=1.2$, $\gamma=0.05$, $\Omega=0.05$, $t=200$ and $J=0$.
}
\label{figS13}
\end{figure*}

We now use the Fermi's golden rule to give a quantitative explanation. For Hermitian systems, the Fermi's golden rule states that, the transition rate from the initial 
state $\ket{i}$ to the final state $\ket{f}$ is given by  $\Gamma_{i\rightarrow f} = \frac{2\pi}{\hbar} |\bra{f} \hat{V} \ket{i}|^2 \delta(E_f - E_i -\hbar \omega)$
as a result of a weak perturbation $\hat{V}$ with $\omega$ being the frequency of the perturbation field. Here, $\ket{i}$ and $\ket{f}$ are the eigenstates of the unperturbed Hamiltonian with eigenenergies $E_i$ and $E_f$, respectively. 
For non-Hermitian systems, since the eigenenergies are complex, we write them as $E_{i/f} = \varepsilon_{i/f} - \ii \gamma_{i/f}$. 
In light of the fact that the imaginary parts of eigenenergies are related to the broadening of spectral lines, for non-Hermitian systems, we may write the transition rate as 
$\Gamma_{i\rightarrow f} \propto |\bra{f} \hat{V} \ket{i}|^2/\big[ (\varepsilon_f - \varepsilon_i -\hbar \omega)^2 + (\gamma_i + \gamma_f)^2 \big]$.
Let $\ket{u_{\text{R},km}^{N=1}}$ be the single-particle right eigenstate of the system Hamiltonian $\hat{\mathcal{H}}_\text{s}$ with eigenenergy 
$E_{km}^{N=1} = \varepsilon_{km}^{N=1} - \ii \gamma_{km}^{N=1}$ labelled by the momentum $k$ and the band index $m$. 
The initial particle on the system levels can be decomposed as $\hat{c}_{k\alpha}^\dagger \ket{0} = d_{k1} \ket{u_{\text{R},k1}^{N=1}} + d_{k2} \ket{u_{\text{R},k2}^{N=1}}$, 
where $d_{k1}$ and $d_{k2}$ are two complex coefficients.
When we measure the occupation $N_{\text{a},q}$ on momentum $q$, the initial state is effectively 
$\ket{i} = \hat{c}_{k\alpha}^\dagger \ket{0} \otimes \ket{\psi_\text{a}^q} = d_{k1} \ket{u_{\text{R},k1}^{N=1}} \otimes \ket{\psi_\text{a}^q} + d_{k2} \ket{u_{\text{R},k2}^{N=1}} \otimes \ket{\psi_\text{a}^q}$, and the final state can be any two-particle eigenstate $\ket{f_n} = \ket{u_{\text{R},n}^{N=2}}$ which is labelled by the index $n$. Since the initial state $\ket{i}$ is not an eigenstate of the system, we need to consider $\ket{i_m} = d_{km} \ket{u_{\text{R},km}^{N=1}} \otimes \ket{\psi_\text{a}^q}$ for $m=1$ and $m=2$ separately. 
The energies of the initial state $\ket{i_m}$ and the final state $\ket{f_n}$ are given by $E_{i_m} = E_{km}^{N=1} = \varepsilon_{km}^{N=1} - \ii \gamma_{km}^{N=1}$ (for simplicity, we set $J=0$ so that the energy of $\ket{\psi_\text{a}^q}$ is 0) and $E_{f_n} = E_{n}^{N=2} = \varepsilon_n^{N=2} - \ii \gamma_n^{N=2}$, respectively. In our case, 
$\omega$ is replaced by the detuning $-\delta$, and thus  
\begin{equation}
\Gamma_{i_m \rightarrow f_n} \propto |\bra{f_n} \hat{V} \ket{i_m}|^2/\big[ (\varepsilon_n^{N=2} - \varepsilon_{km}^{N=1} + \delta)^2 + (\gamma_{km}^{N=1} + \gamma_n^{N=2})^2 \big].
\end{equation}
As a result, we can observe an absorption peak at $-\delta = \varepsilon_n^{N=2} - \varepsilon_{km}^{N=1}$ if $|\bra{f_n} \hat{V} \ket{i_m}|^2 \neq 0$. 
In Fig.~\ref{figS13}, we have plotted $\varepsilon_n^{N=2} - \varepsilon_{km}^{N=1}$ as vertical dashed lines for $m=1$ (blue) and $m=2$ (red), which coincide with the center of the small peaks (the numbers beside the vertical dashed lines denote the index $n$). The imaginary parts $\tilde{\gamma}_{n}^{N=2}$ extracted by curve fitting are in good agreement with $\gamma_{n}^{N=2}$ obtained by diagonalization. For example, for $n=11$, $\tilde{\gamma}_{11}^{N=2} = 0.093$ and $\gamma_{11}^{N=2} = 0.087$, and for $n=15$, $\tilde{\gamma}_{15}^{N=2} = 0.084$ and $\gamma_{15}^{N=2} = 0.074$. These results suggest that we can obtain two-particle energies of an interacting non-Hermitian system. We may also expect that more particle energies may be measured through the method by initially loading more particles on the system levels.